\documentclass[12pt]{article}
\usepackage[margin=1in]{geometry}
\usepackage[round]{natbib}
\usepackage{amsmath,amssymb,authblk,graphicx,amsthm}
\usepackage{natbib}

\DeclareMathOperator*{\argmin}{arg\,min}

\newtheorem{theorem}{Theorem}
\newtheorem{proposition}{Proposition}
\newtheorem{assumption}{Assumption}

\title{Estimation and inference for the indirect effect in high-dimensional linear mediation models}
\author{Ruixuan Rachel Zhou}
\affil{Department of Statistics, University of Illinois at Urbana-Champaign}

\author{Liewei Wang}
\affil{Division of Clinical Pharmacology, Department of Molecular Pharmacology and Experimental Therapeutics, Mayo Clinic}

\author{Sihai Dave Zhao}
\affil{Department of Statistics and Carl R. Woese Institute for Genomic Biology, University of Illinois at Urbana-Champaign}

\begin{document}
\date{}
\maketitle

\begin{abstract}
Mediation analysis is difficult when the number of potential mediators is larger than the sample size. In this paper we propose new inference procedures for the indirect effect in the presence of high-dimensional mediators for linear mediation models. We develop methods for both incomplete mediation, where a direct effect may exist, as well as complete mediation, where the direct effect is known to be absent. We prove consistency and asymptotic normality of our indirect effect estimators. Under complete mediation, where the indirect effect is equivalent to the total effect, we further prove that our approach gives a more powerful test compared to directly testing for the total effect. We confirm our theoretical results in simulations, as well as in an integrative analysis of gene expression and genotype data from a pharmacogenomic study of drug response. We present a novel analysis of gene sets to understand the molecular mechanisms of drug response, and also identify a genome-wide significant noncoding genetic variant that cannot be detected using standard analysis methods.
\end{abstract}

\section{\label{sec:intro}Introduction}
Mediation analysis is of great interest in many areas of research, such as psychology, epidemiology, and genomics \citep{mackinnon2008introduction, hayes2013introduction, huang2014joint}. A major goal is to understand the direct and indirect effects of an exposure variable on an outcome variable, potentially mediated through several intervening variables. Statistical methods for estimating and testing direct and indirect effects are well-developed when the number of mediator variables is relatively small \citep{hayes2013introduction, vanderweele2014mediation, vanderweele2015explanation}, but problems arise when the number of potential mediators exceeds the sample size. This high-dimensional scenario is common in genomics applications. For example, the effects of genetic variants may be mediated through the regulation of gene expression, but it is usually not known \textit{a priori} which genes are regulated, so the total number of potential mediators can be very large. 

General methods for high-dimensional inference are currently the subject of intense research. Techniques based on debiasing penalized regression estimators have been shown to provide asymptotically normal and unbiased estimators for certain parametric sparse regression models \citep{van2014asymptotically, zhang2014confidence, javanmard2014confidence, javanmard2018debiasing}. The sparsity level of the regression parameter is not typically known. \citet{cai2017confidence} discussed the construction of confidence intervals that can adapt to this unknown sparsity, and \citet{zhu2018linear} proposed a test that avoids the sparsity requirement by instead assuming that the precision matrix is known or has certain sparsity properties. While these methods can be used for testing direct effects, they cannot be directly applied to perform inference on indirect effects. One approach is to use them to extend low-dimensional mediation analysis methods such as \citet{vanderweele2014mediation}, but it may be difficult to achieve valid inference, for reasons that will be explained in Section~\ref{sec:model}.

Several semiparametric high-dimensional methods have recently been proposed in the causal inference literature, for the purposes of doing inference on causal effects in the presence of high-dimensional controls \citep{belloni2017program, athey2018approximate}. In particular, the procedure of \citet{athey2018approximate} is closely related to the method proposed here, and is discussed in detail in Section~\ref{sec:connections}. However, these approaches do not directly apply to estimating indirect effects in high-dimensional mediation models. \citet{chen2015high} and \citet{huang2016hypothesis} use principal components analysis to reduce the dimensionality of the mediators, and employ the bootstrap for inference. \citet{hanson2016computational} and \citet{zhang2016estimating} first screen the mediators according to their marginal correlations with the response.

In this paper, we propose, and provide asymptotic guarantees for, two new inferential procedures for the indirect effect in high-dimensional linear mediation analysis models. We first consider the incomplete mediation setting, where both direct and indirect effects might exist. This is a common scenario, for example in genome-wide methylation studies that investigate whether environmental exposures exert their effects on phenotype by altering DNA methylation patterns. The exposures may also act through a non-methylation pathway, giving rise to potential direct effects. We illustrate another application in Section~\ref{sec:data}, where we identify gene sets that may mediate the effect of a gene of interest on a drug response phenotype.

We then consider the complete mediation setting, when it is known that a direct effect does not exist. This setting is common when studying genetic variants located in noncoding regions of the genome, which typically can only exert their effects on a phenotype by regulating gene expression. Recent work has shown that in the low dimensional case, testing for the indirect effect can be much more powerful than directly testing the total effect, even though both are equal under complete mediation \citep{kenny2014power, zhao2014more, loeys2015cautionary}. We show theoretically and in simulations that this is also true for our proposed high dimensional method. Our work can thus be useful in genome-wide association studies where powerful tests are required to detect important variants. In an analysis of the genetics of drug response in Section~\ref{sec:data}, our method was able to identify a genome-wide significant noncoding genetic variant that could not be detected by the standard approach.
\section{\label{sec:proposed}Proposed methods}
\subsection{\label{sec:model}Mediation model and notation}

For the $i$th subject, $i=1, \dots ,n$, let $Y_i$ be the outcome, $G_i$ be a vector of $p$ mediators, and $S_i$ be a vector of $q$ exposures. We allow $p$ to be larger than the sample size $n$, but we assume that $S_i$ is low-dimensional. Finally, assume that the $Y_i$, $G_i$, and $S_i$ have all been centered to have zero mean. We consider the following linear mediation model:
\begin{equation}
  \label{eq:model}
  Y_i = G_i^\top \alpha_0 + S_i^\top \alpha_1 + \epsilon_{1i},
  \quad
  G_i = \gamma S_i + E_{i},
\end{equation}
where $\epsilon_{1i}$ are mean-zero random variables and $E_i$ are mean-zero random vectors that are independent of $G_i$ and $S_i$. Model \eqref{eq:model} implies $G_i^\top \alpha_0 = S_i^\top \gamma^\top \alpha_0 + \epsilon_{2i}$, where $\epsilon_{2i} = E_i^\top \alpha_0$. Let $\sigma_1^2$ denote the variance of $\epsilon_{1i}$ and $\sigma_2^2$ denote the variance of $\epsilon_{2i}$.

We are interested in performing inference on the indirect effect
\begin{equation}
  \label{eq:beta0}
  \gamma^\top \alpha_0 \equiv \beta_0
\end{equation}
of $S_i$ on $Y_i$ when the dimension of $G_i$ exceeds the sample size. We will describe separate methods for the incomplete mediation setting, where $S_i$ may have a direct effect on $Y_i$ through $\alpha_1$, and the complete mediation setting, where $\alpha_1$ is assumed to equal zero. We will assume throughout that $\alpha_0$ is sparse, so that only a small number of variables mediate the effect of $S_i$ on $Y_i$.

Assuming that model \eqref{eq:model} is correctly specified with no unmeasured confounders, $\beta_0$ and the direct effect $\alpha_1$ admit causal interpretations under a counterfactual framework, analogous to low-dimensional mediation models \citep{huang2014joint, vanderweele2014mediation}. See Section \ref{sec:causal_inter} of the Supplementary Materials for a detailed discussion. Our method can also accommodate measured confounders or covariates. If $Z_i$ is a low-dimensional vector of potential confounders, we could write
\[
Y_i = G_i^\top \alpha_0 + S_i^\top \alpha_1 + Z_i^\top \alpha_z + \epsilon_{1i},
\quad
G_i = \gamma S_i + \gamma_z Z_i + E_{i}.
\]
For example, in our data analysis in Section~\ref{sec:data}, we let $Z_i$ be a set of principal components to adjust for population stratification. Section \ref{sec:covs} of the Supplementary Materials describes how our proposed procedures can be modified for this setting. In this paper we do not consider more complicated models, such as interactions between the $Z_i$ and $G_i$ or between $S_i$ and $G_i$. These may require additional methodological development, which we leave for future work. 

The remainder of the paper will use the following notation. Let $S$ be an $n \times q$ matrix of the $S_i$, $G$ be an $n \times p$ matrix of the $G_i$, $Y$ be an $n \times 1$ vector of the $Y_i$, $\epsilon_1$ be an $n \times 1$ vector of the $\epsilon_{1i}$, and $E_i$ be an $n \times p$ matrix of the $E_i$. Define the vector $X_i = (G_i^\top, S_i^\top)^\top$. Also define the sample matrices $\hat{\Sigma}_{SS} = n^{-1} \sum_i (S_i S_i^\top)$, $\hat{\Sigma}_{SG} = n^{-1} \sum_i (S_i G_i^\top)$, $\hat{\Sigma}_{GG} = n^{-1} \sum_i (G_i G_i^\top)$, $\hat{\Sigma}_{GY} = n^{-1} \sum_i (G_i Y_i)$, $\hat{\Sigma}_{XY} = n^{-1} \sum_i X_i Y_i$, and $\hat{\Sigma}_{XX} = n^{-1} \sum_i X_i X_i^\top$, as well as their population-level versions $\Sigma_{SS}$, $\Sigma_{SG}$, $\Sigma_{GG}$, $\Sigma_{GY}$, $\Sigma_{XX}$, and $\Sigma_{XY}$. Finally, for any matrix $A$, let $a_{ij}$ denote the $ij$th entry and let $\Vert A \Vert_{L_1} = \max_j \sum_i \vert a_{ij} \vert$, $\Vert \cdot \Vert_1$ denote the element-wise $\ell_1$ norm, and $\Vert \cdot \Vert_\infty$ denote the element-wise $\ell_\infty$ norm of either a vector or a matrix.

\subsection{\label{sec:intuitions}Intuitions}

This section provides an intuitive description of the challenges of performing inference on the indirect effect $\beta_0$ \eqref{eq:beta0} with high-dimensional mediators. For simplicity, in this subsection we assume that the direct effect $\alpha_1 = 0$. In the low-dimensional problem when $p < n$, $\gamma^\top$ and $\alpha_0$ can be estimated using the ordinary least squares estimates $\tilde{\gamma}^\top = \hat{\Sigma}_{SS}^{-1} \hat{\Sigma}_{SG}$ and $\tilde{\alpha}_0 = \hat{\Sigma}_{GG}^{-1} \hat{\Sigma}_{GY}$, respectively. Then $\beta_0$ can be estimated using $\tilde{\gamma}^\top \tilde{\alpha}_0$. Inference is straightforward because this product estimator typically has an asymptotically normal distribution \citep{sobel1982asymptotic, zhao2014more}, though see the last paragraph of Section \ref{subsec:incomplete}. In high-dimensions, when $p$ exceeds $n$, the challenge is that the ordinary least squares estimator of $\alpha_0$ does not exist. Since $\alpha_0$ is sparse, one solution would be to use penalized regression, such as the lasso, to estimate $\alpha_0$. However these do not have tractable limiting distributions, so inference on $\beta_0$ using this approach is difficult. We illustrate this in simulations in Section \ref{sec:numerical}.

An alternative might be to instead use a debiased lasso estimator $\check \alpha_0$ of $\alpha_0$, whose components do have nice asymptotic distributions \citep{javanmard2014confidence, van2014asymptotically, zhang2014confidence}. We first briefly introduce $\check{\alpha}_0$ following \citet{javanmard2014confidence}. In high-dimensions, the ordinary least squares estimator $\hat{\Sigma}_{GG}^{-1} \hat{\Sigma}_{GY}$ is not feasible because $\hat{\Sigma}_{GG}$ is no longer invertible, but we can still consider estimators of the form $\hat{\Omega} \hat{\Sigma}_{GY}$ for a different data-dependent matrix $\hat{\Omega}$. By model \eqref{eq:model},
  \begin{equation}
    \label{eq:pilot}
    \hat{\Omega} \hat{\Sigma}_{GY} - \alpha_0=  (\hat{\Omega} \hat{\Sigma}_{GG} - I) \alpha_0 + \frac{1}{n} \hat{\Omega} G^\top \epsilon_1,
  \end{equation}
  where $I$ is the $p \times p$ identity matrix. In general, $\hat{\Omega} \hat{\Sigma}_{GY}$ will therefore be a biased estimator, with bias equal to $(\hat{\Omega} \hat{\Sigma}_{GG} - I) \alpha_0$. When $\alpha_0$ is sparse, it turns out that this bias can be well-estimated by $(\hat{\Omega} \hat{\Sigma}_{GG} - I) \tilde{\alpha}_0$, if we carefully construct $\hat\Omega$ so that $\Vert \hat\Omega \Sigma_{GG} - I \Vert_\infty$ is small and $\tilde{\alpha}_0$ is a lasso estimate of $\alpha_0$ so that $\Vert\tilde\alpha_0 - \alpha_0\Vert_1$ is small; for more details see \citet{javanmard2014confidence}. The debiased lasso estimator is then constructed by substracting the estimated bias from $\hat{\Omega} \hat{\Sigma}_{GY}$:
  \begin{equation}
    \label{eq:debiased}
    \check \alpha_0
    =
    \hat{\Omega} \hat{\Sigma}_{GY} - (\hat{\Omega} \hat{\Sigma}_{GG} - I ) \tilde{\alpha}_0
    =
    \tilde \alpha_0 + \frac{1}{n} \hat\Omega G^\top (Y - G \tilde\alpha)
    =
    \alpha_0 + \frac{1}{n} \hat{\Omega} G^\top \epsilon_1 + \Delta,
  \end{equation}
  where $\Delta = (\hat{\Omega} \hat{\Sigma}_{GG} - I) (\alpha_0 - \tilde \alpha)$. It can be shown for suitably constructed $\hat{\Omega}$ that each component of $\Delta$ is $o_P(n^{-1/2})$, so that each component of $n^{1/2} (\check{\alpha}_0 - \alpha_0)$ is asymptotically normal. \citet{javanmard2014confidence} chose $\hat{\Omega}$ to minimize the variance of $\check{\alpha}_0$ while \citet{van2014asymptotically} and \cite{zhang2014confidence} chose $\hat{\Omega}$ to estimate the precision matrix $\Sigma_{GG}^{-1}$.

Despite these encouraging properties, inference using the corresponding estimator $\tilde{\gamma}^\top \check{\alpha}_0$ for $\beta_0$ is still not always possible. Using \eqref{eq:debiased},
  \[
  \tilde{\gamma}^\top \check{\alpha}_0
  =
  \tilde{\gamma}^\top \hat{\Omega} \hat{\Sigma}_{GY} - \tilde{\gamma}^\top (\hat{\Omega} \hat{\Sigma}_{GG} - I) \tilde{\alpha}_0
  =
  \beta_0 + (\tilde{\gamma}^\top - \gamma) \alpha_0 + \frac{1}{n} \tilde{\gamma}^\top \hat{\Omega} G^\top \epsilon_1
  +
  \tilde{\gamma}^\top \Delta,
  \]
which can be interpreted as a debiased version of $\tilde{\gamma}^\top \hat{\Omega} \hat{\Sigma}_{GY}$ for $\beta_0$. 

However, the error $\tilde{\gamma}^\top \Delta$ is no longer negligible: even though each component of $\Delta$ is $o_P(n^{-1/2})$, the linear combination $\tilde{\gamma}^\top \Delta$ may not be, so $n^{-1/2} (\tilde{\gamma}^\top \check{\alpha}_0 - \beta_0)$ may not have an easily characterized asymptotic distribution. We argue in Section \ref{sec:sparse_Sigma_SG} of the Supplementary Materials that we would need to at least assume either that $p \log(p) / n^{1/2} \rightarrow 0$ or that $\gamma$ is sparse in order for $\tilde{\gamma}^\top \Delta = o_P(n^{-1/2})$. However, these conditions are restrictive.

In this paper we propose an estimate of $\beta_0$ under the weaker assumption that $\log (p) / n^{1/2} \rightarrow 0$, and without assumptions on the sparsity of $\gamma$. Our central idea is to develop a debiased estimator not of $\alpha_0$ or $\beta_0$, but of $\Sigma_{SG} \alpha_0$. We will show that the bias of our initial estimator for this quantity can be estimated sufficiently accurately as long as we construct the matrix $\hat{\Omega}$ appropriately. By pre-multiplying our debaised estimate of $\Sigma_{SG} \alpha_0$ by the low-dimensional quantity $\hat{\Sigma}_{SS}^{-1}$, we will obtain an asymptotically normal estimate of $\beta_0$.

\subsection{\label{sec:incomplete}Inference for indirect effect under incomplete mediation}

We first estimate the indirect effect $\beta_0$ \eqref{eq:beta0} under incomplete mediation, where $\alpha_1$ is allowed to be non-zero. Let $X = (G, S)$ be the $n \times (p + q)$ design matrix and $\alpha = (\alpha_0^\top, \alpha_1^\top) ^\top$. As described in Section \ref{sec:intuitions}, our strategy is to first obtain a debiased estimator for $\Sigma_{SG} \alpha_0$, which we will then pre-multiply by $\hat{\Sigma}_{SS}^{-1}$. First define
\begin{equation}
  \label{eq:D}
  D =
  \begin{pmatrix}
    \Sigma_{SG} & 0 \\
    0 & \Sigma_{SS}
  \end{pmatrix},
  \quad
  \hat{D} =
  \begin{pmatrix}
    \hat{\Sigma}_{SG} & 0 \\
    0 & \hat{\Sigma}_{SS}
  \end{pmatrix}.
\end{equation}
Following \eqref{eq:pilot}, we first consider estimators of the form $\hat{\Omega}_I \hat{\Sigma}_{XY}$, for a matrix $\hat{\Omega}_I$ that we will construct later. By model \eqref{eq:model},
\[
\hat{\Omega}_I \hat{\Sigma}_{XY} - 
\begin{pmatrix}
  \Sigma_{SG} \alpha_0\\
  \Sigma_{SS} \alpha_1
\end{pmatrix}
=
(\hat{\Omega}_I \hat{\Sigma}_{XX} - D) \alpha  
+
\frac{1}{n} \hat{\Omega}_I X^\top \epsilon_1,
\]
Since $\alpha$ is sparse, we can estimate the bias term using $(\hat{\Omega}_I \hat{\Sigma}_{XX} - \hat{D}) \tilde{\alpha}$ as in \eqref{eq:pilot}, using a carefully construct $\hat\Omega$ and where $\tilde{\alpha}$ is a lasso estimate of $\alpha$. We will use the scaled lasso of \citet{sun2012scaled} because it also provides a consistent estimate of the variance of the $Y_i$, which will be useful later. We may also leave $\alpha_1$ unpenalized, which is further discussed in Section \ref{sec:dis}.

We can therefore construct a debiased estimate of $\Sigma_{SG} \alpha_0$ by subtracting the estimated bias from $\hat{\Omega}_I \hat{\Sigma}_{XY}$, analogous to \eqref{eq:debiased}. We then pre-multiply the debiased estimator by $I_2 \otimes \hat{\Sigma}_{SS}^{-1}$, where $I_2$ denotes the $2 \times 2$ identity matrix and $\otimes$ denotes the Kronecker product. This gives our proposed estimator $\hat{b}$ for the indirect effect $\beta_0$ under incomplete mediation, as well as an estimate $\hat{a}$ of the direct effect $\alpha_1$:
\begin{align}
  \begin{pmatrix}
    \hat{b} \\
    \hat{a}
  \end{pmatrix}
   =\,&
   (I_2 \otimes \hat{\Sigma}_{SS}^{-1}) \{ \hat{\Omega}_I \hat{\Sigma}_{XY} - (\hat{\Omega}_I \hat{\Sigma}_{XX} - \hat{D}) \tilde{\alpha} \}
   \nonumber \\
  =\,&
  \label{eq:incomplete}
  \begin{pmatrix}
    \hat{\Sigma}_{SS}^{-1} \hat{\Sigma}_{SG} \tilde{\alpha}_0 \\
    \tilde{\alpha}_1
  \end{pmatrix}
  +
  (I_2 \otimes \hat{\Sigma}_{SS}^{-1}) \frac{1}{n} \hat{\Omega}_I X^\top (Y - X \tilde{\alpha}),
\end{align}
where $\tilde{\alpha}_1$ is the component of $\tilde{\alpha}$ that estimates $\alpha_1$.

Analogous to \eqref{eq:debiased}, it remains to find a suitable matrix $\hat{\Omega}_I$ so that $(\hat{\Omega}_I \hat{\Sigma}_{XX} - \hat{D})$ is small. We propose to choose $\hat{\Omega}_I$ to estimate the matrix $D \Sigma_{XX}^{-1}$, for $D$ defined in \eqref{eq:D} and $\Sigma_{XX} = E(X_i X_i^\top)$.
Our estimator is based on constrained $\ell_1$ optimization, similar to the precision matrix estimation procedure of \citet{cai2011constrained}:
\begin{equation}
  \label{eq:M_I}
  \hat{\Omega}_I
  =
  \argmin_\Omega \Vert \Omega \Vert_1
  \text{ subject to }
  \Vert \Omega \hat\Sigma_{XX} - \hat{D} \Vert_\infty \leq \tau_n,
\end{equation}
where $\tau_n$ is a tuning parameter. We will show in Section~\ref{sec:theory} that $\hat{\Omega}_I$ will converge to $D \Sigma_{XX}^{-1}$ under the condition that $D \Sigma_{XX}^{-1}$ is sparse.

We show in Section~\ref{sec:theory} that under certain conditions, $(\hat{b}, \hat{a})$ is asymptotically normal and centered at the true $(\beta_0, \alpha_1)$. We also provide estimates of the asymptotic variance of $\hat{b}$, which will allow us to construct confidence intervals and conduct Wald tests for the indirect effects. Though this paper focuses on the indirect effect, \eqref{eq:incomplete} also gives an estimate $\hat{a}$ for the direct effect. As pointed out by a referee, the direct effect could also be estimated by subtracting $\hat b$ from the ordinary least squares estimate of the total effect of $S_i$ on $Y_i$. We show in Section \ref{compare_incomplete} in the Supplementary Materials that these two approaches are asymptotically equivalent.

\subsection{\label{sec:complete}Inference for indirect effect under complete mediation}

In some applications, for example in the analysis of noncoding genetic variants, it may be known that exposure does not act directly on the outcome, and only acts through mediators. We can make use of the extra information that $\alpha_1 = 0$ to develop a more efficient procedure for estimating the indirect effect $\beta_0$ \eqref{eq:beta0}. As above, we first obtain a debiased estimator for $\Sigma_{SG} \alpha_0$ and then pre-multiply by $\hat{\Sigma}_{SS}^{-1}$. We again first consider estimators of the form $\hat{\Omega}_C \hat{\Sigma}_{GY}$, which by model \eqref{eq:model} satisfy
\[
\hat{\Omega}_C \hat{\Sigma}_{GY} - \Sigma_{SG} \alpha_0
=
(\hat{\Omega}_C \hat{\Sigma}_{GG} - \Sigma_{SG}) \alpha_0 + \frac{1}{n} \hat{\Omega}_C G^\top \epsilon_1.
\]
We construct $\hat{\Omega}_C$ to estimate $\Sigma_{SG} \Sigma_{GG}^{-1}$, analogous to $\hat{\Omega}_I$ \eqref{eq:M_I} above:
\begin{equation}
  \label{eq:M_C}
  \hat{\Omega}_C = \argmin_\Omega \Vert \Omega \Vert_1
  \text{ subject to }
  \Vert \Omega \hat{\Sigma}_{GG} - \hat{\Sigma}_{SG} \Vert_\infty \leq \tau'_n,
\end{equation}
where $\tau'_n$ is a tuning parameter. We show in Section \ref{sec:theory} that $\hat{\Omega}_C$ will converge to $\Sigma_{SG} \Sigma_{GG}^{-1}$ if the latter is sparse. If $\tilde{\alpha}_0$ is the scaled lasso estimate of $\alpha_0$, we can estimate the bias of $\hat{\Omega}_C \hat{\Sigma}_{GY}$ using $(\hat{\Omega}_C \hat{\Sigma}_{GG} - \hat{\Sigma}_{SG}) \tilde{\alpha}_0$. Subtracting this from $\hat{\Omega}_C \hat{\Sigma}_{GY}$ and premultiplying by $\hat{\Sigma}_{SS}^{-1}$ gives our proposed estimate of $\beta_0$ under complete mediation:
\begin{equation}
  \label{eq:complete}
  \tilde{b}
  =
  \hat{\Sigma}_{SS}^{-1} \{ \hat{\Omega}_C \hat{\Sigma}_{GY} - (\hat{\Omega}_C \hat{\Sigma}_{GG} - \hat{\Sigma}_{SG}) \tilde{\alpha}_0 \}
  =
  \hat{\Sigma}_{SS}^{-1} \hat{\Sigma}_{SG} \tilde{\alpha}_0 + \hat{\Sigma}_{SS}^{-1} \frac{1}{n} \hat{\Omega}_C G^\top (Y - G^\top \tilde{\alpha}_0).
\end{equation}
We show in Section~\ref{sec:theory} that $\tilde{b}$ is asymptotically normal and centered at the true $\beta_0$, and provide estimates for its asymptotic variance.

This estimator has an interesting efficiency property. Under complete mediation, $\beta_0$ can also be estimated by directly regressing $Y_i$ on $S_i$ and ignoring the mediating gene expression information. We will show that the asymptotic variance of the ordinary least squares estimator of $\beta_0$ is always greater than or equal to the variance of our $\tilde{b}$. The same phenomenon has been observed in a low-dimensional mediation model \citep{kenny2014power, zhao2014more, loeys2015cautionary}. Intuitively, our procedure achieves this efficiency gain by denoising the outcome $Y_i$, replacing it with an estimate $G_i^\top \tilde \alpha_0$ of its conditional expectation $G_i^\top \alpha_0$ and thus removing much of the variation from the error term $\epsilon_{1i}$. 

\subsection{\label{sec:connections}Connections to existing work}

Estimating the indirect effect in high dimensions is challenging because $\beta_0$ \eqref{eq:beta0} is a linear combination of the high-dimensional vector $\alpha_0$. \citet{athey2018approximate} encountered a similar issue studying inference for a causal effect in the presence of high-dimensional controls and also took a debiasing approach. Both of our approaches can be viewed as debiasing a pilot estimator by subtracting a weighted sum of the residuals from a fitted penalized regression model for $Y_i$. \citet{athey2018approximate} chose the weights in this weighted sum to minimize the estimation error of the desired linear combination, while our weights are equal to $(I_2 \otimes \hat \Sigma_{SS}^{-1}) \hat{\Omega}_I$ in \eqref{eq:incomplete} and $\Sigma_{SS}^{-1} \hat{\Omega}_C$ in \eqref{eq:complete}. The coefficients of the desired linear combination are known in the setting of \citet{athey2018approximate}, while in our approach they are equal to $\Sigma_{SG}$ and must be estimated, so the method of \citet{athey2018approximate} is not directly applicable here. It would be interesting to apply their strategy to our mediation framework in the future.

There are alternative approaches to constructing the matrices $\hat{\Omega}_I$ \eqref{eq:M_I} and $\hat{\Omega}_C$ \eqref{eq:M_C}. One method might be to choose them to minimize the variances of the resulting estimators $\hat a$, $\hat b$, and $\tilde{b}$ while controlling their biases. In the standard linear regression setting with high-dimensional covariates, \citet{javanmard2014confidence} showed that this strategy can give asymptotically optimal inference without requiring the precision matrix of the covariates to be sparse. As pointed out by a referee, applying this strategy to the present mediation setting may obviate the need to assume sparsity of $D \Sigma_{XX}^{-1}$ and $\Sigma_{SG} \Sigma_{GG}^{-1}$. This is an important direction for future work, and Section \ref{sec:sparse_est} of the Supplementary Materials contains a detailed discussion and simulation study exploring the robustness of our procedure to the accuracy of estimating $D \Sigma_{XX}^{-1}$ and $\Sigma_{SG} \Sigma_{GG}^{-1}$. On the other hand, our current strategy of choosing $\hat \Omega_I$ and $\hat \Omega_C$ to estimate $D \Sigma_{XX}^{-1}$ and $\Sigma_{SG} \Sigma_{GG}^{-1}$ allows us to characterize the asymptotic variances of our proposed estimators in terms of population-level quantities, as well as to construct consistent estimates of those variances. \citet{hirshberg2017augmented} studied a similar approach for a more general class of debiased estimators.

\section{\label{sec:theory}Theoretical Results}

\subsection{\label{subsec:incomplete}Incomplete Mediation}

This section presents the theoretical properties of our proposed indirect effect inference procedure under incomplete mediation. We first require $G_{ij}$, $S_i$, and residual error $\epsilon_{1i}$ to have exponential-type tails and make several sparsity assumptions.

\begin{assumption}
  \label{a:tails}
  For each $j = 1, \ldots, p$, $G_{ij}$ has mean zero and $\text{E} \{ \exp(t G_{ij}^2) \} \leq K < \infty$ for some constant $K$ and all $\vert t \vert \leq \eta$, where $\eta \in (0, 1/4)$ and $\{ \log (p+q) \} / n \leq \eta$. The same tail conditions hold for $S_i$ and $\epsilon_{i1}$.
\end{assumption}

\begin{assumption}
  \label{a:sparse_matrices_xx}
  For $D$ defined in \eqref{eq:D}, there exist constants $M_X$ and $N_X$ such that $\Vert \Sigma_{XX}^{-1} \Vert_{L_1} \leq M_X$ and $\Vert (D \Sigma_{XX}^{-1})^\top \Vert_{L_1} \leq N_X$. Furthermore if $\omega_{ij}$ denotes the $ij$th entry of $D \Sigma_{XX}^{-1}$, then $\max_i \sum_j \vert \omega_{ij} \vert^\theta < s_0$ for some $s_0$ and $\theta \in [0, 1)$. 
\end{assumption}

The quantity $s_0$ in Assumption \ref{a:sparse_matrices_xx} measures the degree of sparsity of $D \Sigma_{XX}^{-1}$. The condition on $\Vert \Sigma_{XX}^{-1} \Vert_{L_1}$ requires that none of the rows contain too many large entries. This is reasonable, as precision matrices are frequently used to model conditional dependencies between genes in a gene network \citep{danaher2014joint, zhao2014direct}, and gene networks are typically thought to be sparse. The condition on $\Vert D \Sigma_{XX}^{-1} \Vert_{L_1}$ is related to the irrepresentable condition of \citet{zhao2006model}, and is similar to requiring that $S_i$ cannot be completely explained by $G_i$.

\begin{theorem}
  \label{thm:omegaxx}
  Let $\hat\Omega_{I}$ solve~\eqref{eq:M_I} with tuning parameter $\tau_n = (N_X+1) C_1 \{(\log (p + q) ) / n\}^{1/2}$ for $C_1 = 2 \eta^{-2} (2 + \tau + \eta^{-1} e^2 K^2)^2$, where $K$ and $\eta$ are from Assumption~\ref{a:sparse_matrices_xx} and $\tau>0$. Then under Assumptions~\ref{a:tails} and~\ref{a:sparse_matrices_xx}, with probability greater than $(1 - 4 p^{-\tau})$ and $D$ defined in \eqref{eq:D},
  \[
  \Vert\hat\Omega_{I} - D \Sigma_{XX}^{-1} \Vert_{\infty} \leq (4N_X+2) C_1 M_X \{(\log p) / n\}^{1/2}.
  \]
\end{theorem}

Theorem~\ref{thm:omegaxx} shows that our $\hat{\Omega}_{I}$~\eqref{eq:M_I} is a consistent estimate of the population-level matrix $D \Sigma_{XX}^{-1}$. As discussed in Section \ref{sec:incomplete}, in the standard linear regression setting, \citet{javanmard2014confidence} proposed a method for high-dimensional inference that does not require consistent estimation of precision matrices. Section \ref{sec:sparse_est} in the Supplementary Materials discusses whether their approach can be applied here as well, which would avoid the need for the sparsity conditions in Assumption \ref{a:sparse_matrices_xx}.

We can now characterize the asymptotic behavior of our incomplete mediation estimators $(\hat{b}, \hat{a})$ \eqref{eq:incomplete}. We require additional assumptions necessary for the good performance of the scaled lasso of \citet{sun2012scaled}.

\begin{theorem}
  \label{thm:incomplete}
  Let $\hat{b}$ and $\hat{a}$ be calculated such that both tuning parameters $\lambda_n$ and $\tau_n$ are $O\{(n^{-1} \log p)^{1/2}\}$. Assume the model for $Y_i$ \eqref{eq:model} satisfies the conditions of Theorem~2 of \citet{sun2012scaled} and that $\alpha_0$ has at most $s_0 = o(n^{1/2} / \log p)$ non-zero components. Under Assumptions~\ref{a:tails} and~\ref{a:sparse_matrices_xx}, if $(\log p) / n^{1/2} \rightarrow 0$ and $\alpha_0$ and $\Sigma_{SG}$ are not both zero, and if
  $
  \Gamma
  \equiv
  \Sigma_{SS}^{-1} \Sigma_{SG}
  (\Sigma_{GG} - \Sigma_{GS} \Sigma_{SS}^{-1} \Sigma_{SG})^{-1}
  \Sigma_{GS} \Sigma_{SS}^{-1}
  $
  converges to a positive-definite matrix, then
  \[
  n^{1/2}
  \begin{pmatrix}
    \hat{b} - \beta_0\\
    \hat{a} - \alpha_1
  \end{pmatrix}
  \rightarrow
  N(0, V),
  \mbox{ where }
  V = 
  \begin{pmatrix}
    \sigma_1^2 \Gamma + \sigma_2^2 \Sigma_{SS}^{-1} & -\sigma_1^2 \Gamma \\
    -\sigma_1^2 \Gamma & \sigma_1^2 (\Gamma + \Sigma_{SS}^{-1})
  \end{pmatrix}.
  \]
\end{theorem}

The ultra-sparsity assumption on $s_0$ in Theorem \ref{thm:incomplete} is standard in the de-biased lasso literature \citep{javanmard2014confidence, van2014asymptotically, zhang2014confidence}. The choice of $\tau_n$ controls the coherence parameter $\Vert \hat\Omega_{I} \hat\Sigma_{XX} - \hat D\Vert_\infty$ at rate $(n^{-1} \log p)^{1/2}$, which is necessary for showing that the bias of our proposed estimator goes to 0 when $n$ and $p$ go to infinity. The proof of Theorem \ref{thm:incomplete} shows that the asymptotic variance $V$ can be consistently estimated using
\[
\hat{\sigma}_1^2 (I_2 \otimes\hat\Sigma_{SS}^{-1}) \hat\Omega_{I} \hat\Sigma_{XX} \hat\Omega_{I}^\top (I_2 \otimes\hat\Sigma_{SS}^{-1}) 
+
\begin{pmatrix}
  \hat{\sigma}_2^2 \hat\Sigma_{SS}^{-1} & 0\\
  0 & 0
\end{pmatrix}.
\]
Consistency of $\hat{\Sigma}_{XX}$ and $\hat{\Sigma}_{SS}$ is standard, and consistency of $\hat{\Omega}_{I}$ is given by Theorem~\ref{thm:omegaxx}. Estimation of $\hat{\sigma}_1^2$ and $\hat{\sigma}_2^2$ is discussed in Section~\ref{sec:implementation}. 


We caution that Theorem \ref{thm:incomplete} does not cover the setting where both $\Sigma_{SG} = 0$ and $\alpha_0 = 0$. This would cause $n^{1/2}(\hat b - \beta_0)$ to asymptotically equal zero, rather than be normally distributed. A related issue arises even for standard low-dimensional Wald-type tests for the indirect effect, such as Sobel's test \citep{sobel1982asymptotic, hayes2013introduction, barfield2017testing}. In practice, these tests can be conservative when the exposure, the mediator, and the outcome are only weakly associated. In this case, the true finite-sample distribution of the Wald test statistic has higher kurtosis than a normal distribution, so that critical values calculated assuming a normal distribution lead to a conservative test \citep{barfield2017testing}. This setting is different from the weak instrumental variable problem, which we discuss in Section \ref{compare_iv} of the Supplemetary Materials.

\subsection{\label{subsec:complete}Complete Mediation}

We now present the theoretical properties of our indirect effect inference procedure under complete mediation. Similar to Assumption~\ref{a:sparse_matrices_xx}, we first make several sparsity assumptions, under which we can show that $\hat{\Omega}_C$ \eqref{eq:M_C} is a consistent estimate of $\Sigma_{SG} \Sigma_{GG}^{-1}$.

\begin{assumption}
\label{a:sparse_matrices_gg}
There exist constants $M_G$ and $N_G$ such that $\Vert \Sigma_{GG}^{-1} \Vert_{L_1} \leq M_G$ and $\Vert (\Sigma_{SG} \Sigma_{GG}^{-1})^\top \Vert_{L_1} \leq N_G$. Furthermore, if $\omega_{ij}$ denotes the $ij$th entry of $\Sigma_{SG} \Sigma_{GG}^{-1}$, then $\max_i \sum_j \vert \omega_{ij} \vert^\theta < s_0$ for some $s_0$ and $\theta \in [0, 1)$.
\end{assumption}


\begin{theorem}
\label{thm:omegagg}
Let $\hat{\Omega}_{C}$ solves~\eqref{eq:M_C} with tuning parameter $\tau_n = (N_G+1) C_1 \{(\log p) / n\}^{1/2}$. Then under Assumptions~\ref{a:tails} and~\ref{a:sparse_matrices_gg}, with probability greater than $(1 - 4 p^{-\tau})$ and $C_1$ and $\tau$ as in Theorem~\ref{thm:omegaxx},
\[
\Vert \hat{\Omega}_{C} - \Sigma_{SG} \Sigma_{GG}^{-1} \Vert_\infty \leq (4N_G+2) C_1 M_G \{(\log p) / n\}^{1/2}
\]
\end{theorem}

We can now characterize the asymptotic behavior of our complete mediation indirect effect estimator $\tilde{b}$~\eqref{eq:complete}. The proof of Theorem \ref{thm:complete} indicates that the asymptotic variance of $\tilde{b}$ can be consistently estimated by $\hat\sigma_1^2\hat\Sigma_{SS}^{-1}\hat\Omega_{C}\hat\Sigma_{GG}\hat\Omega_{C}\hat\Sigma_{SS}^{-1}+\hat\sigma_2^2\hat\Sigma_{SS}^{-1}$. As with incomplete mediation case, the requirement that $\alpha_0$ and $\Sigma_{SG}$ are not both zero arises here as well.

\begin{theorem}
  \label{thm:complete}
  Let $\tilde{b}$ be calculated such that both tuning parameters $\lambda_n$ and $\tau_n$ are of order $O\{(n^{-1} \log p)^{1/2}\}$. Assume the model for $Y_i$ in mediation model~\eqref{eq:model} has $\alpha_1 = 0$ but otherwise satisfies the conditions of Theorem~2 of \citet{sun2012scaled} and that $\alpha_0$ has at most $s_0 = o(n^{1/2} / \log p)$ non-zero components. Under Assumptions~\ref{a:tails} and~\ref{a:sparse_matrices_gg}, if $(\log p) / n^{1/2} \rightarrow 0$, $\alpha_0$ and $\Sigma_{SG}$ are not both zero, and if $\Sigma_{SG} \Sigma_{GG}^{-1} \Sigma_{GS}$ converges to a positive-definite matrix, then
  \[
  n^{1/2}(\tilde{b} - \beta_0)
  \rightarrow
  N(0, \sigma_1^2\Sigma_{SS}^{-1}\Sigma_{SG} \Sigma_{GG}^{-1} \Sigma_{GS}\Sigma_{SS}^{-1}+\sigma_2^2\Sigma_{SS}^{-1}).
  \]
\end{theorem}

As mentioned in Section~\ref{sec:complete}, under complete mediation the indirect effect $\beta_0$ can also be consistently estimated by directly regressing $Y_i$ on $S_i$. The expression for the asymptotic variance of $\tilde{b}$ from Theorem~\ref{thm:complete} now allows us to analytically compare our estimator with the ordinary least squares estimate of $\beta_0$. 

\begin{proposition}
  \label{p:var}
  In model~\eqref{eq:model}, assume that $\alpha_1 = 0$, so that $\tilde{b}_{OLS} = (S^\top S)^{-1} S^\top Y$ is a consistent estimator of $\beta_0$. Then under the conditions of Theorem~\ref{thm:complete},
  $
  \text{var}\{ n^{1/2}(\tilde{b}_{OLS} - \beta_0) \}
  -
  \text{var}\{ n^{1/2} (\tilde{b} - \beta_0) \}
  $
  converges to a positive semi-definite matrix.
\end{proposition}

Proposition~\ref{p:var} shows that our $\tilde{b}$ always has equal or lower asymptotic variance compared to the ordinary least squares estimator, even when the mediators are high-dimensional. This extends similar findings in low dimensions \citep{kenny2014power, zhao2014more, loeys2015cautionary}. In fact, we show in Section \ref{optimal_variance} of the Supplementary Materials that for any fixed $p$, our estimator $\tilde b$ achieves the minimum asymptotic variance among all asymptotically unbiased estimators of $\beta_0$ with the same convergence rate. Tests based on $\tilde{b}$ thus will have higher power to detect non-zero $\beta_0$ than tests based on $\tilde{b}_{OLS}$, as confirmed by simulations in Section \ref{sec:complete_sims}. In practice, the Wald test based on $\tilde{b}$ can still be conservative when $\alpha_0$ and $\Sigma_{SG}$ are close to zero, for reasons discussed in Section \ref{subsec:incomplete}, but simulations show that our proposed $\tilde b$ can still have significant power gains over the majority of the parameter space.

In a closely related setting, \citet{athey2016estimating} found that when estimating the causal effect of a treatment on a long-term outcome, leveraging intermediate outcomes can increase efficiency. In Section \ref{compare_athey2016a} of the Supplementary materials we provide a detailed comparison. Together, these results converge on a common principle, and provide theoretical justification for recent work in genomics showing that data integration using mediation analysis can increase power to detect important biological signals \citep{wang2012ibag, huang2015igwas}.

The improved efficiency guaranteed by Proposition \ref{p:var} requires strong scientific or expert knowledge to justify the absence of the direct effect. Furthermore, it also depends on the correct specification of both stages of the linear mediation model \eqref{eq:model}. In low dimensions, this was been pointed out by \citet{loeys2015cautionary}. This is in contrast to the usual ordinary least squares estimator, which requires fewer modeling assumptions. We illustrate the effect of model misspecification on our proposed estimator in Section \ref{sec:mis} in the Supplementary Materials.

Our estimator $\hat{b}$~\eqref{eq:incomplete}, proposed in Section \ref{sec:incomplete} under incomplete mediation, could also be used to estimate the indirect effect under complete mediation. Proposition~\ref{p:var'} shows that under complete mediation, $\tilde{b}$ is asymptotically more efficient.

\begin{proposition}
  \label{p:var'}
  In model~\eqref{eq:model}, assume that $\alpha_1 = 0$. Under the conditions of Theorems~\ref{thm:incomplete} and \ref{thm:complete}, $\text{var}\{ n^{1/2}(\hat{b} - \beta_0) \} - \text{var}\{ n^{1/2} (\tilde{b} - \beta_0) \}$ converges to a positive semi-definite matrix.
\end{proposition}

\section{\label{sec:implementation}Implementation}

We first center the $Y_i$, $G_i$, and $S_i$. To apply the scaled lasso, we standardize all covariates to have unit variance and then choose the tuning parameter $\lambda_n$ using the quantile-based penalty procedure in the R package \verb|scalreg|.

To estimate the asymptotic variances of our estimators, given in Theorems~\ref{thm:incomplete} and~\ref{thm:complete}, we need estimates of the residual variances $\sigma_1^2$ and $\sigma_2^2$ from our mediation model~\eqref{eq:model}. \citet{sun2012scaled} showed that the scaled lasso can provide a consistent estimate $\hat{\sigma}_1^2$ for $\sigma_1^2$. Since model~\eqref{eq:model} implies that $Y_i = S_i^\top (\beta_0 + \alpha_1) + \epsilon_i$ where $\epsilon_i \sim N(0, \sigma_1^2 + \sigma_2^2)$, we can estimate $\sigma_2^2$ by first regressing $Y_i$ on $S_i$ to obtain the ordinary least squares residual variance estimator $\hat{\sigma}^2$, and using $\hat\sigma_2^2 = \hat\sigma^2 - \hat\sigma_1^2$. In practice, $\hat\sigma_1$ may sometimes be larger than $\hat\sigma$, in which case we estimate $\hat\sigma_2 = 0$. This is sensible because $\hat\sigma_1 > \hat\sigma$ likely occurs when no mediators are associated with the outcome, i.e., $\alpha_0 = 0$, in which case $\sigma_2$ indeed equals zero.

We construct the matrices $\hat{\Omega}_I$ \eqref{eq:M_I} and $\hat{\Omega}_C$ \eqref{eq:M_C} by setting the tuning parameters $\tau_n = \tau'_n = \{(\log p) / n\}^{1/2} / 3$. This choice is guided by Theorem~\ref{thm:omegaxx} and~\ref{thm:omegagg}. We also tried choosing the tuning parameters by minimizing an \textit{ad hoc} information criterion-type measure, but this resulted in confidence intervals with poor coverage in some cases. Finding a more data-adaptive tuning procedure is an important direction for future research.


The time-consuming part of our method is the constrained $\ell_1$ optimization, in \eqref{eq:M_I} and \eqref{eq:M_C}, which we implement using fast algorithms from the \verb|flare| package. For $n = 300$ subjects, $p = 1,000$ mediators, and $q = 1$ exposure, our procedure with tuning parameter  $\tau_n = \{(\log p) / n\}^{1/2}/3$ takes 66 seconds on a single core of a Intel Xeon X5675 processor at 3.07GHz and 8 GB of RAM, and larger $\tau_n$ results in shorter computation time: our procedure with $\tau_n = \{(\log p) / n\}^{1/2}$ takes 5 seconds. Our procedure is available in the R package \verb|freebird|, available at \verb|github.com/rzhou14/freebird|.

\section{\label{sec:numerical}Numerical Results}

\subsection{\label{sec:methods_compared}Methods Compared}

We compared our methods to a naive non-debiased method, discussed in the begining of Section \ref{sec:intuitions}. This estimates $\beta_0$ using $\hat{\Sigma}_{SS}^{-1} \hat{\Sigma}_{SG} \tilde{\alpha}_0$, where $\tilde{\alpha}_0$ is a standard lasso estimate of $\alpha_0$ implemented using the R package \verb|glmnet|. There is no tractable limiting distribution for this estimate of $\beta_0$, so we used the bootstrap to obtain percentile confidence intervals and obtained average power and coverages based on those intervals. Bootstrapping the lasso is not theoretically justified \citep{dezeure2015high}, but this at least allows us to have a comparable baseline method.

We also compared our procedures to the high-dimensional mediation analysis method of \citet{zhang2016estimating}, using their R package \verb|HIMA|. The method first uses marginal screening on mediators to reduce dimensionality. It then regresses the outcome on the retained mediators using penalized regression with the minimax concave penalty \citep{zhang2010nearly}. Using only the selected mediators, it uses ordinary least squares to compute a pair of $p$-values for each mediator, for its associations with the outcome and the exposure. These $p$-values are Bonferroni-corrected for the number of selected mediators, and \citet{zhang2016estimating} identify a mediator as significant if both of its adjusted $p$-values are less than the desired significance level. However, this testing approach based on the maximum of two $p$-values does not provide confidence intervals.
 
Under complete mediation, the indirect effect is equal to the total effect, which can be tested directly using ordinary least squares. In this setting we therefore also compared our complete mediation method to ordinary least squares.

\subsection{\label{sec:incomplete_sims}Simulations under incomplete mediation}
We first studied our estimators under incomplete mediation. Following model~\eqref{eq:model}, for samples $i = 1, \ldots, n = 300$ we generated $q = 1$ exposure $S_i \sim N(0,1)$ and $p = 500$ potential mediators $G_i$ following $G_i = c \gamma S_i + E_i$, where $c$ was a scalar, $\gamma$ was a $p \times 1$ coefficient vector, and $E_i \sim N(0, \Sigma_E)$. We generated $\Sigma_E$ following procedures in \citet{danaher2014joint} such that $\Sigma_E^{-1}$ was sparse in the sense of Assumption~\ref{a:sparse_matrices_xx} and had diagonal entries equal to 1. Finally, we generated the outcome according to $Y_i = G_i^\top \alpha_0 + S_i^\top \alpha_1 + \epsilon_{1i}$, where $\epsilon_{1i} \sim N(0,5)$. In Section \ref{sec:nongaussian} of the Supplementary Materials, we show that our simulation results were similar even when $\epsilon_{1i}$ was not normally distributed. We let $\gamma$ have 15 non-zero components randomly generated between $[-1, 1]$, fixing $\gamma$ across replications, and let $\alpha_0$ have 15 non-zero components equal to one. We chose either one or five of these non-zero components to correspond to variables whose entries in $\gamma$ were also non-zero; these were the true mediators. Here we set the direct effect $\alpha_1 = 0.1$, and in Section \ref{sec:add_incom} of the Supplementary Materials we present results when $\alpha_1 = 0.5$. In this simulation scheme, the indirect effect $\beta_0 = c \gamma^\top \alpha_0$, and we varied $\beta_0$ by varying the constant $c$.

When there was only one true mediator, we used all three competing methods to test $H_0: \beta_0 = 0$ at the $\alpha = 0.05$ significance level, and used our proposed method and naive method to calculate 95\% confidence intervals for different values of $\beta_0$. When there were five true mediators, we did not apply the method of \citet{zhang2016estimating} because it considers each mediator separately and does not provide inference for the overall indirect effect. The existence of multiple significant mediators does not imply that the indirect effect is nonzero because the effects of the different mediators may cancel each other out, a phenomenon known as inconsistent mediation. 

\begin{figure}[h!]
  \centering
  \includegraphics[width = \textwidth]{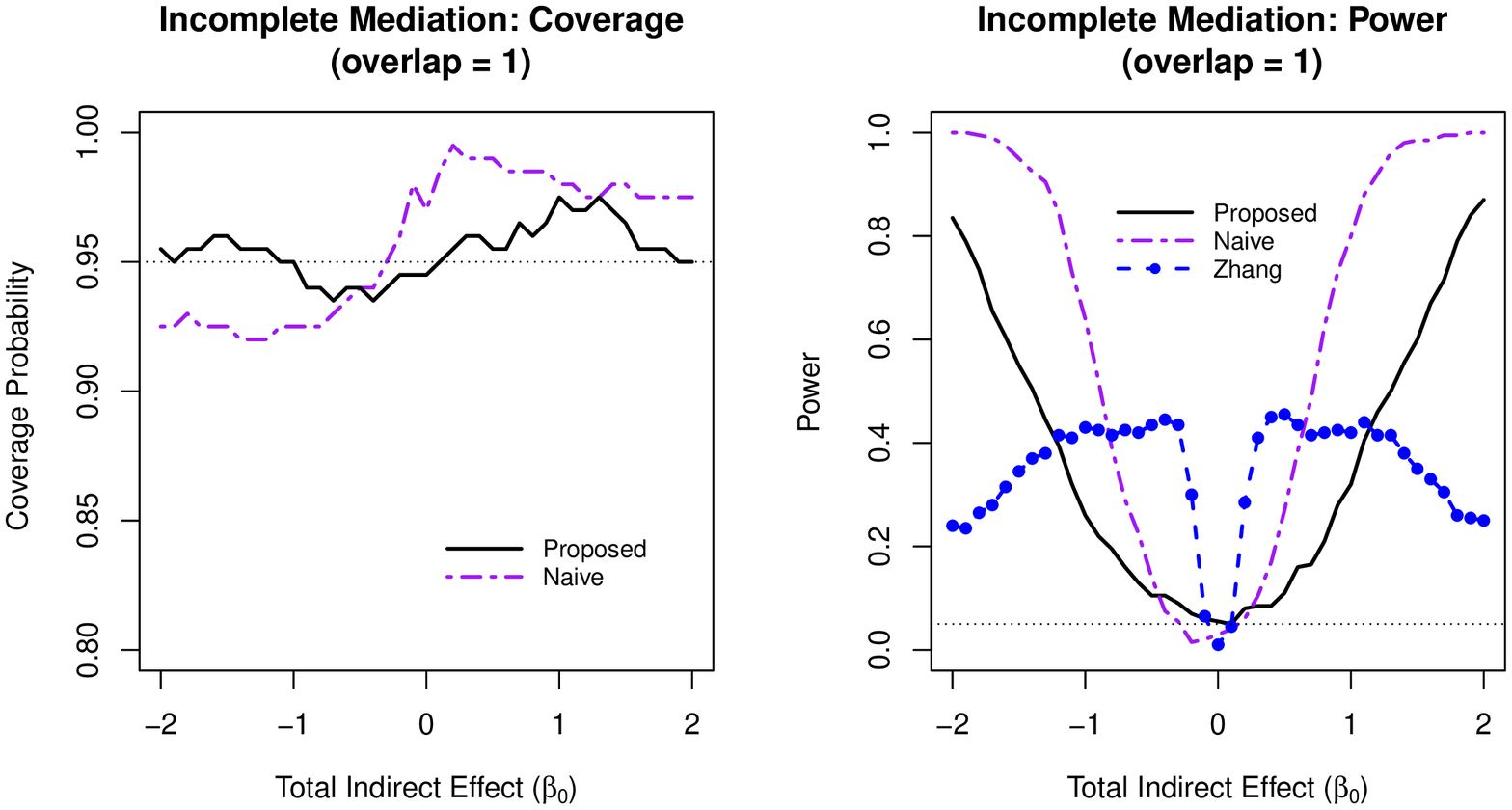}\par
  \includegraphics[width = \textwidth]{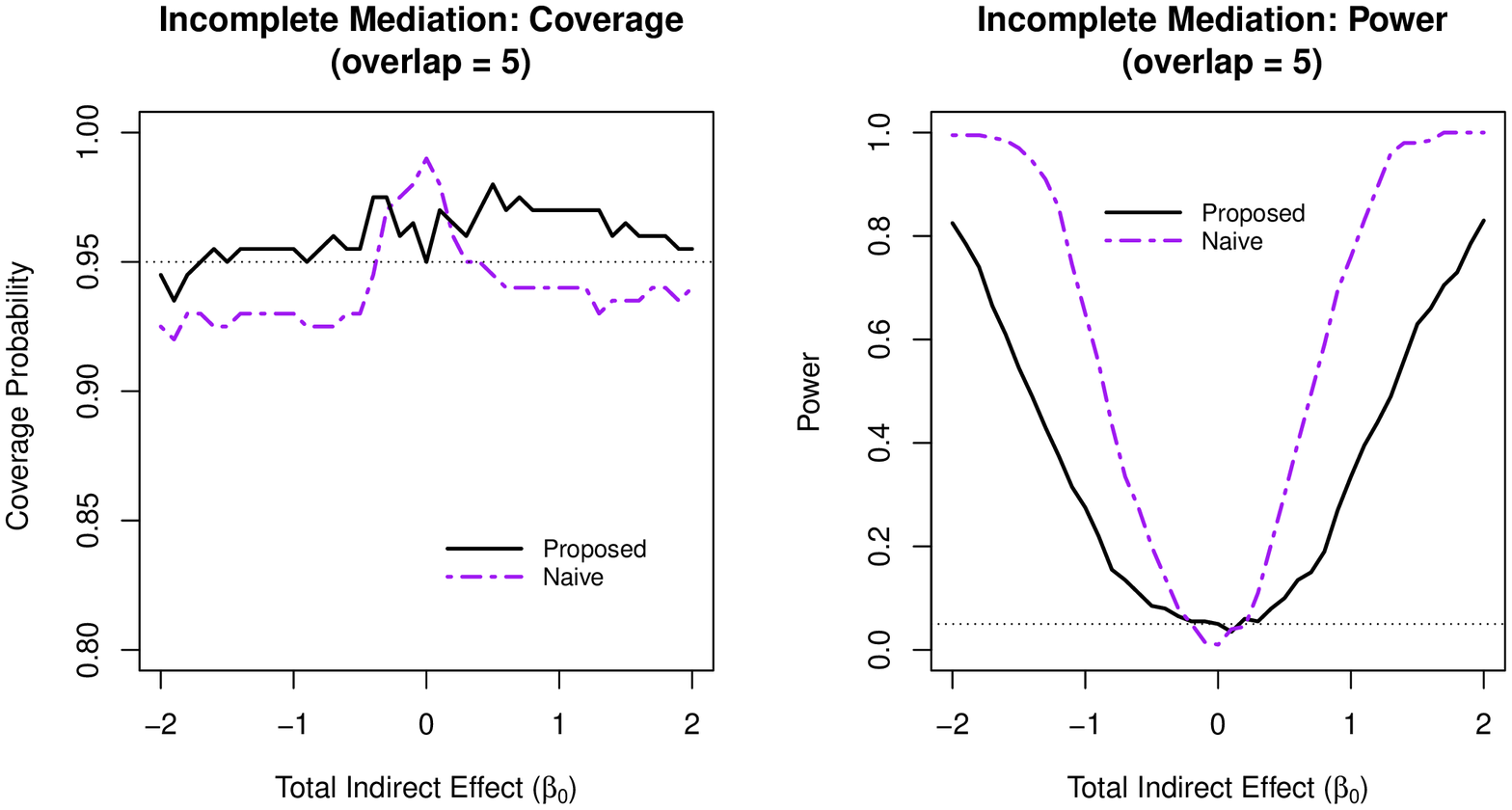}\par
  \caption{\label{fig:incomplete}Average coverage probabilities of 95\% confidence intervals (left panels) and average power curves at significance level $\alpha = 0.05$ (right panels) for estimating and testing the indirect effect under incomplete mediation, over 200 replications. The direct effect was $\alpha_1 = 0.1$. The number of true mediators was 1 in the upper panels and 5 in the lower panels. Proposed: $\hat{b}$ from~\eqref{eq:incomplete}; Naive: the naive method discussed in Section \ref{sec:methods_compared}; Zhang: method of \citet{zhang2016estimating}.}
\end{figure}

Figure~\ref{fig:incomplete} reports average coverage probabilities and power curves over 200 replications. The naive method had worse coverage than our approach, and though it had excellent power, it was not theoretically justified, as mentioned in Section \ref{sec:methods_compared}. The method of \citet{zhang2016estimating} had counterintuitive behavior when $\beta_0$ was large and surpringly high power when $\beta_0$ was small. Its power was poor for large $\beta_0$ because its model selection step performed poorly: larger $\beta_0$ corresponded to larger $c$ and therefore to increased collinearity between $G_i$ and $S_i$ in the regression for $Y_i$, making consistent model selection difficult. Its power was surprisingly high for small $\beta_0$ because it did not appropriately account for the variability of its model selection step. In Section \ref{sec:zhang} of the Supplementary Materials, we describe a slightly modified version of their approach that gives confidence intervals and show that it has poor coverage, and also construct a setting where it fails to maintain type I error because of its improper post-model selection inference.

\subsection{\label{sec:complete_sims}Simulations under complete mediation}

We next studied the performance of our indirect effect estimator under complete mediation. We considered four simulation settings based on the same data generation scheme used above, but with $\alpha_1 = 0$. We generated $p = 500$  potential mediators with either 1 or 5 true mediators, and in Section \ref{sec:add_com} of the Supplementary Materials we present results for $p = 1000$.

\begin{figure}[h!]
  \centering
  \includegraphics[width = \textwidth]{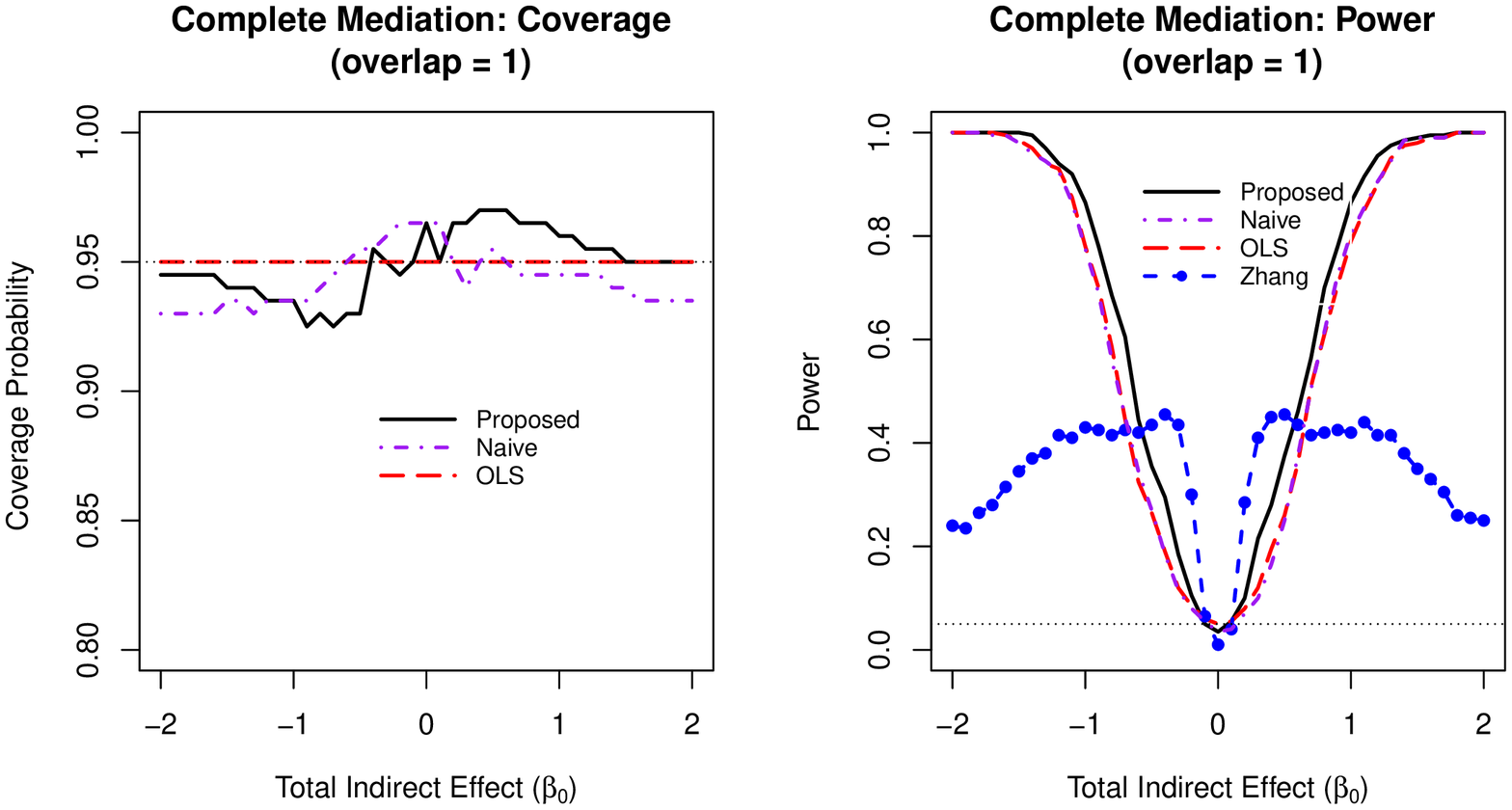}\par
  \includegraphics[width = \textwidth]{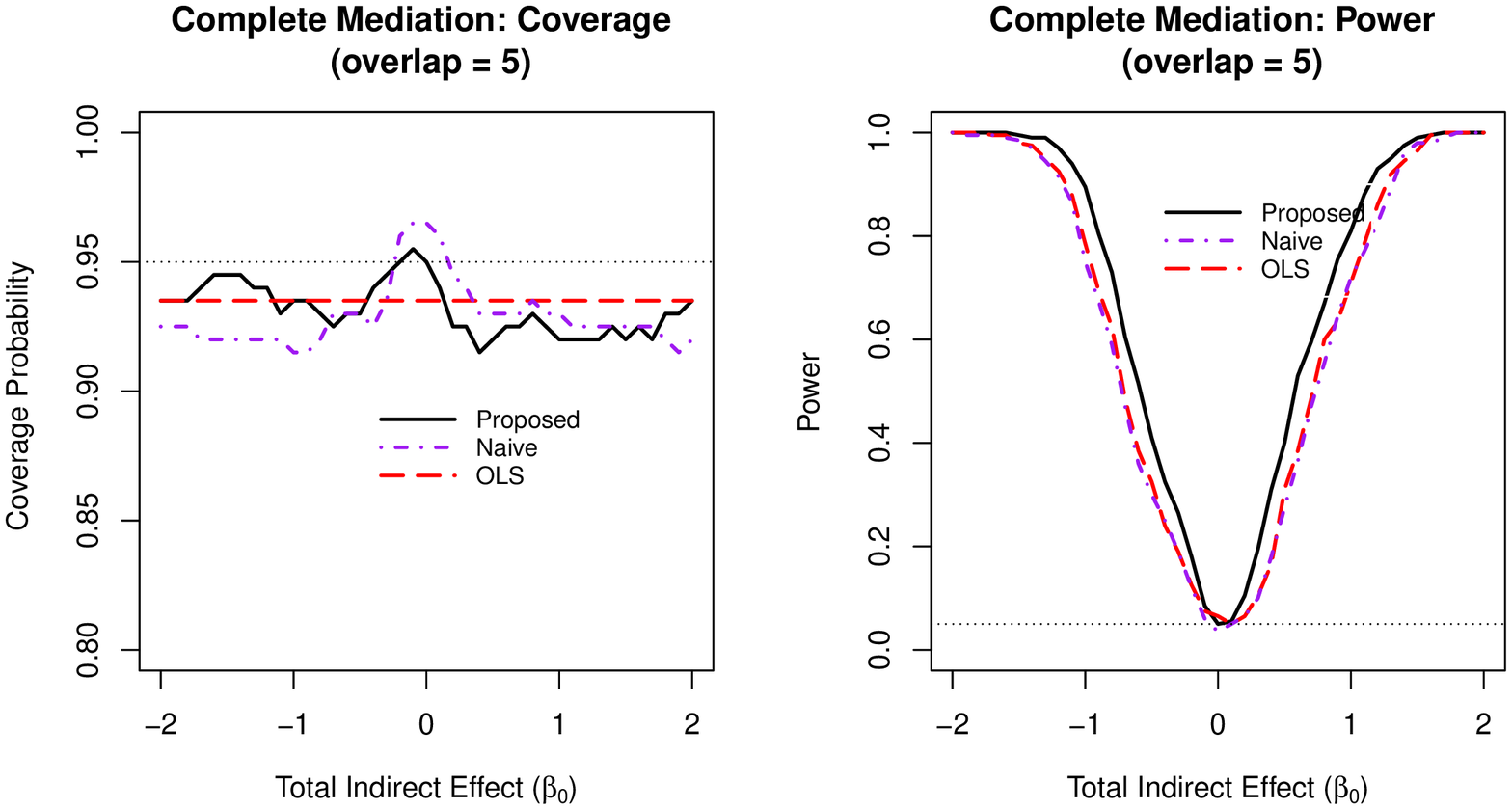}\par
  \caption{\label{fig:complete}Average coverage probabilities of 95\% confidence intervals (left panels) and average power curves at significance level $\alpha = 0.05$ (right panels) for estimating and testing the indirect effect under complete mediation, over 200 replications. The number of true mediators was 1 in the upper panels and 5 in the lower panels, with 500 potential mediators; Proposed: $\tilde{b}$ from~\eqref{eq:complete}; Naive: the naive method discussed in Section \ref{sec:methods_compared}; OLS: ordinary least squares estimate; Zhang: the method of \citet{zhang2016estimating}.}
\end{figure}

Figure~\ref{fig:complete} reports average coverage probabilities and power curves over 200 replications. Our method was always able to maintain the nominal coverage probability and significance level, and in every case had higher power than ordinary least squares and the naive method for sufficiently large $\beta_0$, consistent with Proposition~\ref{p:var}. The average lengths of 95\% confidence intervals were also smaller for our method compared to ordinary least squares; see Section \ref{sec:length} in the Supplementary Materials. Similar to the incomplete mediation setting, the method of \citet{zhang2016estimating} had high power when $\beta_0$ was small and counterintuitive behavior when $\beta_0$ was large. Our test was slightly conservative for $\beta_0$ close to zero because the normal approximation to the distribution of our Wald-type test statistic is poor under weak mediation, as discussed in Section~\ref{subsec:incomplete}.

\section{\label{sec:data}Data analysis}
\subsection{\label{sec:data_description}Data description}

Understanding the mechanisms behind individual variation in drug response is an important step in the development of personalized medicine. We applied our proposed methods to pharmacogenetic studies of the response to the cancer drug docetaxel in human lymphoblastoid cell line \citep{hanson2016computational, niu2012genetic}. The data consists of genotype data on 1,362,849 single nucleotide polymorphisms and expression data on 54,613 probes, after preprocessing, from cell lines from 95 Han-Chinese, 96 Caucasian, and 93 African-American individuals. These data are available from the Gene Expression Omnibus under accession number GSE24277. \citet{niu2012genetic} exposed these cells to docetaxel and quantified their responses using $\text{EC}_{50}$, the concentration at which a drug reduces the population of cells by half \citep{hanson2016computational}.

\subsection{\label{sec:gene}Gene set analysis}

It is common in gene expression profiling experiments to identify genes that are significantly associated with the phenotype being studied. A natural next step is to identify gene sets, representing biological pathways, through which these significant genes may act. This is a difficult analysis problem because the intervening pathways may contain a large number of genes, resulting in a high-dimensional mediation analysis problem. Note that this is a different from standard gene set enrichment analysis \citep{subramanian2005gene}, as the latter does not allow for direct testing of mediation by the gene set. 

We applied our proposed procedures to test whether a candidate gene set mediates the indirect effect of a given gene of interest on the phenotype. We used our incomplete mediation estimator $\hat{b}$~\eqref{eq:incomplete}, because the gene of interest may have a direct effect on the phenotype that does not proceed through the candidate gene set. As in our simulations, we set the tuning parameter $\tau_n = \{(\log p) / n\}^{1/2} / 3$ when estimating $\hat{\Omega}_I$~\eqref{eq:M_I}. As an illustration, we studied the indirect effects of TMED10, a transmembrane trafficking protein whose corresponding gene was the most significantly associated with docetaxel response in our data. We retrieved biological process Gene Ontology gene sets with at least 50 genes from Molecular Signatures Database \citep{subramanian2005gene, liberzon2011molecular}, then applied our proposed approach to test the indirect effect of TMED10 through each of the 4,436 candidates. Of these, 420 gene sets contained more genes than there were samples, making our high-dimensional approach indispensable. 

\begin{table}
  \centering
  \begin{tabular}{lcc}
    \hline
    Gene set & $95\%$ CI & $p$-value \\
    \hline 
Regulation of heart rate& -0.83$\pm$0.22 & $6.3 \times 10^{-14}$\\
Synaptic vesicle cycle& -0.60$\pm$0.17 & $ 2.6 \times 10^{-12}$ \\
Regulation of vasoconstriction& -1.08$\pm$0.30 & $3.1 \times 10^{-12}$\\
Negative regulation of transporter activity& -0.64$\pm$0.18 & $ 5.3\times 10^{-12}$ \\
Negative regulation of cation transmembrane transport& -0.70$\pm$0.20 & $ 5.3\times 10^{-12}$ \\
Positive regulation of blood circulation& -1.07$\pm$0.31 & $ 1.1\times 10^{-11}$ \\
Negative regulation of transmembrane transport& -0.73$\pm$0.21 & $ 2.2\times 10^{-11}$ \\
Regulation of cardiac muscle contraction    & -0.60$\pm$0.18 & $ 6.8\times 10^{-11}$ \\
Neurotransmitter transport    & -0.58$\pm$0.18 & $ 1.4\times 10^{-10}$ \\
Regulation of oxidoreductase activity  & -0.73$\pm$0.22 & $ 1.5\times 10^{-10}$ \\
    \hline
  \end{tabular}
  \caption{\label{tab:GO}Top 10 most significant gene sets through which the TMED10 gene may act on drug response. $95\%$ CI: confidence intervals obtained from the proposed method under incomplete mediation~\eqref{eq:incomplete}; $p$-value: raw $p$-values obtained from the proposed procedure.}
\end{table}

Our procedure found 257 gene sets with significant indirect effects that passed Bonferroni correction. One reason for the large number of significant findings is that many gene sets are subgroups of larger sets. Table~\ref{tab:GO} reports the top 10 most significant ones, as ranked by their indirect effect $p$-values. Many of these are involved in transmembrane transport, which suggests that the role of TMED10 in the response to docetaxel may be to move small molecules into and out of cells. Our proposed method can thus generate useful exploratory results for further downstream analysis. We also implemented the method of \citet{zhang2016estimating}, which found no significant gene sets.

\subsection{\label{sec:snp}Non-coding Variants Analysis}

We next studied the effects of non-coding genetic variants on the response to docetaxel. We first performed a standard genome-wide association study and regressed docetaxel $\text{EC}_{50}$ on each variant separately, controlling for the first five principal components of the genotype data in order to control for population stratification \citep{price2006principal}. This approach did not identify any significant variants after multiple testing correction. We were then interested in whether a high-dimensional mediation analysis method could provide more power. We chose the top 1,000 expression probes with the largest variances as potential mediators and controlled for the first five principal components. We first applied the method of \citet{zhang2016estimating}, but it did not detect any significant variants that passed Bonferroni correction.

It is known that non-coding variants likely do not have a direct effect on the phenotype. This justifies application of our complete mediation estimator $\tilde{b}$~\eqref{eq:complete} to test for non-coding variances associated with $\text{EC}_{50}$. We use $\tau_n' = \{(\log p) / n\}^{1/2} / 3$ when estimating $\hat{\Omega}_C$~\eqref{eq:M_C}, and controlled for the first five principal components in all of our analyses. Our new procedure was indeed able to identify one significant variant that passed Bonferroni correction for all non-coding variants: the single nucleotide polymorphism rs11578000, with an estimated indirect effect of $\hat{b} = -0.0777 \pm 0.0186$ and a $p$-value of $2.8 \times 10^{-16}$. Interestingly, the Genotype Tissue Expression Project \citep{lonsdale2013genotype} found that in heart and muscle tissue, rs11578000 regulated the expression of the gene SUSD4, which has been found to inhibit the complement system \citep{holmquist2013sushi}, a system of proteins involved in innate immunity that may be involved in the response to epirubicin/docetaxel treatment in breast cancer patients \citep{michlmayr2010modulation}. Our $\tilde{b}$ provides novel findings that could not have been detected using standard approaches.

\section{\label{sec:dis}Discussion}

Our methods require that the directions of causality in mediation model~\eqref{eq:model} be correctly specified. In practice this causal pathway may be complex, as some genes react to the outcome, rather than cause the outcome. Our method's findings should thus be further analyzed to verify that the causal directions are indeed of interest. One potential solution to this issue is to use recently developed methods for high-dimensional causal inference \citep{buhlmann2014high} to first screen out reactive genes before applying our proposed procedures.

Though we focused on testing the indirect effect in this paper, our incomplete mediation method also provides $\hat{a}$, a natural estimate of the direct effect, as discussed in Section~\ref{sec:incomplete}. We explored using $\hat{a}$ to test for the presence of a direct effect, and similar to \citet{kenny2014power}, we found that the power was relatively low. This may be because when calculating our estimators, we penalize the direct effect parameter $\alpha_1$ when we fit the scaled lasso. This makes sense if the direct effect is expected to be zero, which is sensible in our integrative genomics applications, but an alternative is to leave $\alpha_1$ unpenalized. This may give a more powerful test for the direct effect, and more work is required to derive the asymptotic distribution of the resulting estimator. Inference for the direct effect in high dimensions could also be achieved by applying debiased lasso methods to test $\alpha_1$ in the regression of $Y_i$ on $G_i$ and $S_i$ in our model~\eqref{eq:model}. Based on some simulations we found that our estimator $\hat{a}$ is always smaller in absolute value, and usually had smaller variance, compared to the debiased estimator of \citet{van2014asymptotically}.

Finally, we have so far only considered linear mediation models for continuous outcomes. It is possible to extend our methods to generalized linear models for the outcome $Y_i$ in mediation model~\eqref{eq:model}. However, the causal interpretation of these nonlinear models requires special care \citep{vanderweele2010odds, vanderweele2015explanation}. Also, we have so far assumed that the residual errors $\epsilon_{1i}$ and $E_i$ are independent of the exposure $S_i$ and mediator $G_i$ in model \eqref{eq:model}. Under heteroskedasticity, if the errors are dependent on either $S_i$ or $G_i$, our theoretical results will likely not hold, and extending our approach to this setting is an important research direction.

\section*{Acknowledgements}

The authors would like to thank the reviewers and the Associate Editor for their extremely useful comments, as well as Casey Hanson for his help with data processing. This work was funded in part by the Mayo Clinic-UIUC Alliance and by a grant awarded by NIGMS through funds provided by the trans-NIH Big Data to Knowledge (BD2K) initiative (www.bd2k.nih.gov). The content is solely the responsibility of the authors and does not necessarily represent the official views of the National Institutes of Health. The work of Dave Zhao was funded in part by the National Science Foundation.

\bibliographystyle{abbrvnat}
\bibliography{egbib}

\appendix

\section{\label{sec:causal_inter}Causal interperation of direct and indirect effect}
Throughout the Supplementary Materials, let $\Omega_I = D \Sigma_{XX}^{-1}$ and $\Omega_C = \Sigma_{SG} \Sigma_{GG}^{-1}$ be the population level matricesas introduced in Section \ref{sec:incomplete} and \ref{sec:complete} in the main text. Let $\epsilon_2$ be an $n \times 1$ vector of the $\epsilon_{2i} = E_i^\top \alpha_0$.

Suppose in a mediation model we have scalar exposure $S$, scalar outcome $Y$, and $p$ potential mediators denoted by $G = (G_1, \dots, G_p)^\top$. Let $G_s$ be the counterfactual value of $G$ if exposure $S$ were set to the value $s$ and let $Y_{sg}$ denote the value of outcome $Y$ if $S$ were set to the value $s$ and $G$ were set to $g$. The controlled direct effect on $Y$ of changing exposure $S$ from $s$ to $s^*$ is defined to be $Y_{sg} - Y_{s^*g}$, where the value of $G$ is kept at $g$. This measures the portion of the effect of $S$ on $Y$ that does not proceed through $G$. The natural direct effect of exposure $S$ on outcome $Y$ is defined as $Y_{s G_{s^*}} - Y_{s^*G_{s^*}}$. Finally, the natural indirect effect is defined as $Y_{s G_s} - Y_{s G_{s^*}}$. The natural indirect effect measures the effect of setting the exposure to equal $s$, and then changing the mediator from what it would equal if the exposure were equal to $s$, to what would it would equal if the exposure were equal to $s^*$. 

To identify these direct and indirect effects, we will need the following no unmeasured confounding assumptions, following \citet{vanderweele2014mediation}:
\begin{assumption}
  \label{a:confounding1}
There is no unmeasured confounding of the exposure-outcome relationship.
\end{assumption}

\begin{assumption}
  \label{a:confounding2}
There is no unmeasured confounding of the mediator-outcome relationship.
\end{assumption}
\begin{assumption}
  \label{a:confounding3}
There is no unmeasured confounding of the exposure-mediator relationship.
\end{assumption}

\begin{assumption}
  \label{a:confounding4}
There is no confounders between mediators and outcome that is affected by the exposure.
\end{assumption}

Assumptions \ref{a:confounding1}--\ref{a:confounding4} can also be interpreted using several conditional independence relationships and details can be found in Section 2 of \citet{vanderweele2014mediation}. Under Assumptions \ref{a:confounding1}--\ref{a:confounding4} and model \eqref{eq:model} in the main text, the same derivation from Section 3.2 of \citet{vanderweele2014mediation} can be used to show that
\begin{align*}
  &\text{E} ( Y_{sg} - Y_{s^*g} ) = \alpha_1 (s - s^*),\\
  &\text{E} ( Y_{sG_s^*} - Y_{s^*G_{s^*}} ) = \alpha_1 (s - s^*),\\
  &\text{E} ( Y_{s G_s} - Y_{s G_s^*} ) = \gamma_0 \alpha_0^\top (s - s^*) = \beta_0^\top (s - s^*).
\end{align*}
Thus our $\alpha_1$ can be interpreted as the average controlled and natural direct effects, and $\beta_0$ can be interpreted as the average natural indirect effect, of a one-unit change in the exposure $S$. The detailed proof can be found in the Appendix of \citet{vanderweele2014mediation}.

\section{\label{sec:covs}Incorporating covariates}

If $Z_i$ is a low-dimensional vector of potential confounders, we can write
\[
Y_i = G_i^\top \alpha_0 + S_i^\top \alpha_1 + Z_i^\top \alpha_z + \epsilon_{1i},
\quad
G_i = \gamma S_i + \gamma_z Z_i + E_{i}.
\]
We can extend our proposed methods to perform high-dimensional inference on the indirect effect $\beta_0 = \gamma^\top \alpha_0$ under this model. We do not consider more complicated models, such as interactions between the $Z_i$ and $G_i$ or between $S_i$ and $G_i$. These may require additional conceptual and methodological development, which we leave for future work.

Define $\tilde{S}_i = (S_i^\top, Z_i^\top)^\top$ and $\tilde{X}_i = (G_i^\top, \tilde{S}_i^\top)^\top$. Following the development of our proposed incomplete mediation method in Section \ref{sec:incomplete} in the main text, we also define
\begin{equation}
  \label{eq:D}
  D =
  \begin{pmatrix}
    \Sigma_{\tilde{S}G} & 0 \\
    0 & \Sigma_{\tilde{S}\tilde{S}}
  \end{pmatrix},
  \quad
  \hat{D} =
  \begin{pmatrix}
    \hat{\Sigma}_{\tilde{S}G} & 0 \\
    0 & \hat{\Sigma}_{\tilde{S}\tilde{S}} \\
  \end{pmatrix}.
\end{equation}
Define
\[
\hat{\Omega}
=
\argmin_\Omega \Vert \Omega \Vert
\text{ subject to }
\Vert \Omega \hat\Sigma_{\tilde{X}\tilde{X}} - \hat{D} \Vert_\infty \leq \tau_n,
\]
where $\tau_n$ is a tuning parameter. We consider estimators of the form $\hat{\Omega} \hat{\Sigma}_{\tilde{X}Y}$, which satisfy
\[
\hat{\Omega} \hat{\Sigma}_{\tilde{X}Y} - 
\begin{pmatrix}
  \Sigma_{\tilde{S}G} \alpha_0\\
  \Sigma_{\tilde{S}S} \alpha_{1z}
\end{pmatrix}
=
(\hat{\Omega} \hat{\Sigma}_{\tilde{X}\tilde{X}} - D) \alpha  
+
\frac{1}{n} \hat{\Omega} \tilde{X}^\top \epsilon_1,
\]
where $\alpha_{1z} = (\alpha_1^\top, \alpha_z^\top)^\top$ and $\alpha = (\alpha_0^\top, \alpha_{1z}^\top)^\top$. We estimate the bias term using $(\hat{\Omega} \hat{\Sigma}_{\tilde{X}\tilde{X}} - \hat{D}) \tilde{\alpha}$, where $\tilde{\alpha}$ is a lasso estimate of $\alpha$, and construct a debiased estimate of $\Sigma_{\tilde{S}G} \alpha_0$ by subtracting the estimated bias from $\hat{\Omega} \hat{\Sigma}_{\tilde{X}Y}$. We then pre-multiply the debiased estimator by $I_2 \otimes \hat{\Sigma}_{\tilde{S}\tilde{S}}^{-1}$, where $I_2$ denotes the $2 \times 2$ identity matrix and $\otimes$ denotes the Kronecker product, to obtain
\begin{align*}
  \begin{pmatrix}
    \hat{b} \\
    \hat{a}_{1z}
  \end{pmatrix}
   =\,&
   (I_2 \otimes \hat{\Sigma}_{\tilde{S}\tilde{S}}^{-1}) \{ \hat{\Omega} \hat{\Sigma}_{\tilde{X}Y} - (\hat{\Omega} \hat{\Sigma}_{\tilde{X}\tilde{X}} - \hat{D}) \tilde{\alpha} \}
   \\
   =\,&
  \begin{pmatrix}
    \hat{\Sigma}_{\tilde{S}\tilde{S}}^{-1} \hat{\Sigma}_{\tilde{S}G} \tilde{\alpha}_0 \\
    \tilde{\alpha}_{1z}
  \end{pmatrix}
  +
  (I_2 \otimes \hat{\Sigma}_{\tilde{S}\tilde{S}}^{-1}) \frac{1}{n} \hat{\Omega} \tilde{X}^\top (Y - \tilde{X} \tilde{\alpha}),
\end{align*}
where $\tilde{\alpha}_{1z}$ is the component of $\tilde{\alpha}$ that estimates $\alpha_{1z}$. The first $q$ components of $\hat{b}$ will be an asymptotically normal estimate of $\beta_0$.

\section{\label{sec:sparse_Sigma_SG}Conditions for directly using a debiased estimator of $\alpha_0$}

In Section \ref{sec:model} of the main text we argued that the reason we do not opt for plugging in a de-biased lasso estimator in $(\sum_i S_i S_i^\top)^{-1} (\sum_i S_i G_i^\top \alpha_0)$ is that then it will not be gauranteed that the bias of the estimator is $o_P(1)$. This may in fact be possible, assuming either that $p \log(p) / n^{1/2} \rightarrow 0$ or that the correlations between $S_i$ and $G_i$ are sparse, and here we give a detailed reasoning by look at the de-biased lasso estimator $\alpha_0$ in detail.

Recall that our model implies 

\begin{equation} 
\label{eq:model_variance}
Y = G \alpha_0 + \epsilon_1,
\quad
G \alpha_0 = S \beta_0 + \epsilon_2.
\end{equation}

If we directly used a de-biased estimator $\hat \alpha$ for $\alpha$, the estimator for the indirect effect $\beta_0$ would be $\tilde \beta = (S^\top S)^{-1} S^\top G \hat \alpha$, and we would have
\begin{eqnarray*}
\tilde \beta - \beta_0 
&=&
 (S^\top S)^{-1} S^\top G (\hat \alpha - \alpha_0) + (S^\top S)^{-1} S^\top G \alpha_0 - \beta_0 \\
& = &
\hat \Sigma_{SS}^{-1} \hat \Sigma_{SG} (\hat \alpha - \alpha_0) + (S^\top S)^{-1} S^\top (S \beta_0 + \epsilon_2) - \beta_0,
\end{eqnarray*}
so
\begin{eqnarray*}
n^{1/2} ( \tilde \beta - \beta_0 )
&=&
n^{1/2} \hat \Sigma_{SS}^{-1} \hat \Sigma_{SG} (\hat \alpha - \alpha_0) + n^{-1} \hat \Sigma_{SS}^{-1} S^\top \epsilon_2.
\end{eqnarray*}

Based on properties of the debiased lasso estimator \citep{javanmard2014confidence}, 
$$
n^{1/2} ( \hat \alpha - \alpha_0 ) = Z + \Delta, 
$$
where $Z = n ^ {1/2} M G^\top \epsilon_1$ and $\Delta = n^{1/2} (M \hat\Sigma_{GG} - I) (\hat\alpha - \alpha_0)$, and $M$ is chosen so that $ \Vert M \hat\Sigma_{GG} - I \Vert_\infty \asymp (\log(p) / n)^{1/2}$. Since the debaised $\hat \alpha$ satisfies $\Vert \hat\alpha - \alpha_0 \Vert_1 = O_P\{  (\log(p) / n)^{1/2} \}$,
\begin{eqnarray*}
\Vert \Delta \Vert_\infty 
&=&
 \Vert n^{1/2} (M \hat\Sigma_{GG} - I) (\hat\alpha - \alpha_0) \Vert_\infty \\
&\leq &
n^{1/2} \Vert M \hat\Sigma_{GG} - I \Vert_\infty \Vert \hat\alpha - \alpha_0 \Vert_1 \\
&=&
O_P\{ \log(p) / n^{1/2} \}.
\end{eqnarray*}
Then
\begin{eqnarray*}
n^{1/2} ( \tilde \beta - \beta_0 )
&=&
\hat \Sigma_{SS}^{-1} \hat \Sigma_{SG} Z + n^{-1} \hat \Sigma_{SS}^{-1} S^\top \epsilon_2 
+ \hat \Sigma_{SS}^{-1} \hat \Sigma_{SG} \Delta,
\end{eqnarray*}
and the first two terms can be shown to be asymptotically normal.

The difficulty arises when trying to show that the last term is $o_P(1)$. The most straightforward analysis gives
\[
\Vert \hat \Sigma_{SS}^{-1} \hat \Sigma_{SG} \Delta \Vert_\infty
\leq
\Vert \hat \Sigma_{SS}^{-1} \Vert_{L_1}
(\Vert \Sigma_{SG}^\top \Vert_{L_1} \Vert \Delta \Vert_\infty
+
\Vert \hat \Sigma_{SG}^\top -  \Sigma_{SG}^\top \Vert_{L_1} \Vert \Delta \Vert_\infty).
\]
Without any additional assumptions, $\Vert \Sigma_{SG}^\top \Vert_{L_1} = O(p)$, so $\Vert \hat \Sigma_{SS}^{-1} \hat \Sigma_{SG} \Delta \Vert_\infty = O_P\{ p \log(p) / n^{1/2} \}$. Therefore, if $p$ and $n$ are such that $p \log(p) / n^{1/2} \rightarrow 0$, this term will be $o_P(1)$. Alternatively, we might assume that each row of $\Sigma_{SG}$ has only at most $o\{n^{1/2} / \log(p) \}$ non-zero components, in which case the term will be $o_P(1)$ as well.

\section{Relationship with other methods}

\subsection{\label{compare_athey2018}Comparing with \citet{athey2018approximate}}

Here we compare our method and the method of \citet{athey2018approximate}, with both attempt to debias a linear combination of a sparse high-dimensional vector. First, \citet{athey2018approximate} adopt the potential outcomes framework where they observe a treatment $W_i$, the covariates $X_i$, and outcome 
\[
Y_i^{\text{obs}} = Y_i(W_i) = \begin{cases}
Y_i(1) \quad \text{if  } W_i = 1,\\
Y_i(0) \quad \text{if  } W_i = 0
\end{cases}
\]
for each experimental unit $i$. We adopt a mediation framework with exposure $S_i$, mediators $G_i$, and outcome $Y_i$. Their treament variable $W_i$ is analogous to our exposure $S_i$, though $W_i$ can only take values 0 or 1 while $S_i$ is under no restriction.

Second, \citet{athey2018approximate} assume that $Y_i^{\text{obs}}$ and $X_i$ are related by the linear model $\text{E}\{Y_i(0) \mid X=x\} = x^\top \beta_c$ in the control group and $\text{E}\{Y_i(1) \mid X=x\} = x^\top \beta_t$ in the greatment group. This is analogous to the first regression model $Y_i = G_i^\top \alpha_0 + S_i^\top \alpha_1 + \epsilon_{1i}$ in our mediation model \eqref{eq:model} in the main text. However, our second model $G_i = \gamma S_i + E_i$ has no analog in the framework of \citet{athey2018approximate}, as their $X_i$ are not related to their $W_i$.

Third, the objective of \citet{athey2018approximate} is to estimate the average treatment effect
\[
\tau = \frac{1}{n_t} \sum_{ \{ i: W_i = 1 \} } \text{E} ( Y_i(1) - Y_i(0) \mid X_i),
\]
where $n_t$ is the size of the treatment group. In contrast, we have a very different goal, which is to estimate an indirect effect.

Fourth, \citet{athey2018approximate} note that the average treatment effect can be written as 
\[
\tau = \overline{X}_t \beta_t - \overline{X}_t \beta_c,
\]
where $\overline{X}_t = n_t^{-1}\sum_{i:W_i=1} X_i$. A major focus of their work is to estimate and do inference on the second term, which is a linear combination of the unknown high-dimensional vector $\beta_c$. A similar challenge lies at the heart of our work, which seeks to estimate an indirect effect $\beta_0$ that under \eqref{eq:model} can be written as $\beta_0 = \gamma_0^\top \alpha_0$. This is also a linear combination of an unknown high-dimensional vector, namely $\alpha_0$. Directly using a penalized regression estimate of $\beta_c$ would result in a large bias for $\tau$, which is the same problem that we face with directly using a penalized regression estimate of $\alpha_0$; this was discussed in Section \ref{sec:model} in the main text. An additional difficulty that we have, due to our mediation framework, is that the coefficients $\gamma_o$ of our linear combination are unknown and must be estimated as well, in contrast to the known $\overline{X}_t$ in \citet{athey2018approximate}.

Finally, our solutions to this bias problem are similar in principle. \citet{athey2018approximate} proposed to estimate $\overline{X}_t \beta_c$ using $\overline{X}_t \hat\beta_c + \sum_{\{i: W_i = 0\}} \gamma_i (Y_i^{\text{obs}} - X_i \hat \beta_c)$, where $\hat\beta_c$ is the lasso estimate of $\beta_c$ and $\gamma_i$ are weights, chosen to minimize the estimation error. The term $\overline{X}_t\hat\beta_c$ corresponds to the first part $\hat \Sigma_{SS}^{-1}\hat{\Sigma}_{SG} \tilde{\alpha}_0$ of our estimator \eqref{eq:complete}, and the term $\sum_{\{i: W_i = 0\}} \gamma_i (Y_i^{\text{obs}} - X_i \hat \beta_c)$ is analgous to the second part $\hat \Sigma_{SS}^{-1} n^{-1} \hat{\Omega}_{C} G^\top (Y - G \tilde\alpha)$. In other words, our estimators have the same form and differ in how the weights of the penalized regression residuals are defined.

\subsection{\label{compare_athey2016a}Comparing with \citet{athey2016estimating}}

Similar to our Proposition \ref{p:var} in the main text, \citet{athey2016estimating} also found that leveraging surrogate outcomes can increase the efficiency of estimating average treatment effects. Here we discuss the relationship between our results and Section 7 of \citet{athey2016estimating}.

First, there are certain similarities between the surrogate outcome framework of \citet{athey2016estimating} and our mediation framework. In their single sample design, \citet{athey2016estimating} consider a binary treatment $W$, a long-term outcome $Y$, and intermediate outcome vector $S$, and assume that $S$ is a surrogate for $Y$ in the sense that $W$ and $Y$ are independent conditional on $S$. In the context of our mediation model \eqref{eq:model} in the main text, their treatment $W$ is analogous to our exposure $S$, their surrogate $S$ is analogous to our mediator $G$, and their outcome $Y$ is our outcome $Y$ as well. Their surrogacy assumption is analogous to our assumption of complete mediation in Section \ref{sec:complete} of the main text, because without a direct effect, $Y$ and $S$ are independent conditional in $G$. On the other hand, \citet{athey2016estimating} consider a model-free approach while we posit a parametric linear mediation model.

The key to both our procedure and the procedure of \citet{athey2016estimating} is to assume that the surrogates/mediators are the only pathway through which the treatment/exposure can affect the outcome. When their surrogacy assumption holds, \citet{athey2016estimating} show that surrogates can be used to pool treated and untreated units to estimate a causal treatment effect. This increases the effective sample size and therefore the estimation efficiency. In contrast, our complete mediation method works by denoising our outcome $Y_i$ using our mediators $G_i$. We regress our denoised outcome on our exposure $S_i$, which is more efficient than the regression of $Y_i$ on $S_i$ directly.

Both we and \citet{athey2016estimating} also characterize the theoretical efficiency gain provided by the surrogates/mediators, and our results share many qualitative similarities.
Assuming homoskedasticity and no pretreatment variables or other covariates, \citet{athey2016estimating} show that the decrease in variance afforded by using surrogates can be expressed as
\begin{equation}
\label{eq:athey_eff_gain}
\text{E} \left [ \frac{2 \sigma^2}{p (1-p)} \cdot \left \{ p(1-p) - (r(S_i) - p)^2 \right \} \right ], 
\end{equation}
where $p$ is the probability of receiving the treament and $r(S_i)$ is the conditional probability of having received the treament given the value of the surrogates. The $\sigma^2$ in~\eqref{eq:athey_eff_gain} denotes the variance of the outcome conditional on the surrogates and is thus analogus to our $\sigma_1^2$ from model \eqref{eq:model}. In our work, Proposition~\ref{p:var} shows that the difference of the asymptotic variances of the ordinary least squares estimator and our estimator can be represented as
\begin{equation}
\label{eq:our_var_diff}
\sigma_1^2\Sigma_{SS}^{-1}(\Sigma_{SS}-\Sigma_{SG}\Sigma_{GG}^{-1}\Sigma_{GS})\Sigma_{SS}^{-1}. 
\end{equation}

First, both \eqref{eq:athey_eff_gain} and \eqref{eq:our_var_diff} indicate that that the efficiency gain is a linear function of the variance of the outcome conditional on the surrogates/mediators. This makes sense because both of our methods gain efficiency essentially by denoising the outcome using the surrogates/mediators. In Section~\ref{sec:effectnoise} below we confirmed this claim in simulations by varying $\sigma_1^2$.

Second, both analyses show that the relative efficiency gain is weakest when the treatment/exposure is perfectly correlated with the surrogates/mediators, as in the case the latter offer no additional information. Indeed, when treatment is perfectly correlated with the surrogates, $r(S_i) = W_i$ in \eqref{eq:our_var_diff} and the efficiency gain is zero. When our exposure is perfectly correlated with our mediators, $\Sigma_{SG} \Sigma_{GG}^{-1} \Sigma_{GS} = \Sigma_{SS}$, so \eqref{eq:our_var_diff} equals zero as well. Similarly, the efficiency gain is the greatest when the treatment/exposure is independent of the surrogates/mediators, when $r(S_i) = p$ in \eqref{eq:athey_eff_gain} and $\Sigma_{SG} = 0$ in \eqref{eq:our_var_diff}. Of course, in this case the true treatment/indirect effect is also equal to zero, so this scenario is not interesting from the testing perspective. These results imply that when testing for a treatment/indirect effect, both methods are most useful when the treatment/exposure is weakly correlated with the surrogates/mediators.

A major difference is that our method assumes a parametric two-stage linear mediation model, whereas \citet{athey2016estimating} operate under a semi-parametric framework. Though the semi-parametric setting is more flexible, it also necessarily imposes restrictions on the support of the treatment/exposure: \citet{athey2016estimating} only allow for binary treatments, and extensions to more than two treatments groups or to continuous treatments are difficult. Furthermore, \citet{athey2016estimating} only proposed estimation methods, whereas we also provide distributional results that allow for valid inference. Nevertheless, relaxing our framework to semiparametric is an important next step and in the future, we hope to explore more in this direction.

\subsection{\label{compare_incomplete}Alternative way to estimate the direct effect under incomplete mediation}

Under incomplete mediation, we proposed estimators
$
\begin{pmatrix}
  \hat b \\
  \hat a
\end{pmatrix}
$ for 
$
\begin{pmatrix}
 \beta_0 \\
 \alpha_1
\end{pmatrix}
$
where we have 
\[
n^{1/2}
\begin{pmatrix}
  \hat{b} - \beta_0\\
  \hat{a} - \alpha_1
\end{pmatrix}
\rightarrow
N(0, V),
\mbox{ where }
V = 
\begin{pmatrix}
  \sigma_1^2 \Gamma + \sigma_2^2 \Sigma_{SS}^{-1} & -\sigma_1^2 \Gamma \\
  -\sigma_1^2 \Gamma & \sigma_1^2 (\Gamma + \Sigma_{SS}^{-1})
\end{pmatrix},
\]
and  $
\Gamma
\equiv
\Sigma_{SS}^{-1} \Sigma_{SG}
(\Sigma_{GG} - \Sigma_{GS} \Sigma_{SS}^{-1} \Sigma_{SG})^{-1}
\Sigma_{GS} \Sigma_{SS}^{-1}
$. At the same time, our model implies that
$$Y = S(\alpha_1 + \beta_0) + \epsilon_1 + \epsilon_2 = S \theta + \epsilon_1 + \epsilon_2,$$
where $\theta = \alpha_1 + \beta_0$ is the total effect. Define $\hat \theta_{OLS}$ to be the ordinary least squares estimate of $\theta$. As pointed out by a referee, an alternative way to estimate the direct effect $\alpha_1$ is $\tilde\alpha_1 = \hat\theta_{OLS} - \hat b$. We show here that this alternative is asymptotically equivalent to $\hat a$.

First, it's easy to see that $\tilde\alpha_1$ is asymptotically normal and unbiased, because $\text{E}(\hat\theta_{OLS} - \hat b) = \theta - \beta_0 = \alpha_1$. To calculate its asmptotic variance, write
\[
\tilde\alpha_1 - \alpha_1= (\hat\theta_{OLS} - \hat b) - (\theta - \beta_0) = (\hat{\theta}_{OLS} - \theta) - (\hat{b} - \beta_0),
\]
so that
\begin{align*}
  n^{1/2} (\tilde\alpha_1 - \alpha_1)
  =\,&
  (1, -1) \{ n^{1/2}
  \begin{pmatrix}
    \hat{\theta}_{OLS} - \theta\\
    \hat{b} - \beta_0
  \end{pmatrix} \} \\
  =\,&
  (1, -1)
  \begin{pmatrix}
    n^{-1/2}\hat\Sigma_{SS}^{-1} S^\top (\epsilon_1 + \epsilon_2)\\
    n^{-1/2} (\hat\Sigma_{SS}^{-1} \quad 0) \hat\Omega_{I} X^\top \epsilon_1 + n^{-1/2}\hat\Sigma_{SS}^{-1} S^\top \epsilon_2 + \Delta_2
  \end{pmatrix} \\
  =\,&
  (1, -1)
  \begin{pmatrix}
    n^{-1/2}\Sigma_{SS}^{-1} S^\top (\epsilon_1 + \epsilon_2) + \Delta_3\\
    n^{-1/2} (\Sigma_{SS}^{-1} \quad 0) \Omega_{I} X^\top \epsilon_1 + n^{-1/2}\Sigma_{SS}^{-1} S^\top \epsilon_2 + \Delta_{0,1} + \Delta_{1,1} + \Delta_2
  \end{pmatrix},
\end{align*}
where $\Delta_2$ is same $\Delta_2$ as in the proof of Theorem~\ref{thm:incomplete}, $\Delta_{0,1}$ and $\Delta_{1,1}$ are the first q components of $\Delta_{0}$ and $\Delta_{1}$ as in the proof of Theorem~\ref{thm:incomplete}, and $\Delta_3 =  n^{-1/2} (\hat\Sigma_{SS}^{-1}-\Sigma_{SS}^{-1})S^\top(\epsilon_1 + \epsilon_2)$ has the same structure with $\Delta_2 = n^{-1/2} (\hat\Sigma_{SS}^{-1}-\Sigma_{SS}^{-1})S^\top \epsilon_2$. So following the proof of Theorem~\ref{thm:incomplete}, it can be easily shown that $\Vert \Delta_3 \Vert_\infty = o_P(1)$ and $\Vert \Delta_{0,1} + \Delta_{1,1} + \Delta_2 \Vert_\infty = o_P(1)$.

Let
\[
W = n^{-1/2} \sum_{i=1}^{n}
\begin{pmatrix}
    \Sigma_{SS}^{-1} S_i^\top \epsilon_{1i }+ \Sigma_{SS}^{-1} S_i^\top \epsilon_{2i}\\
     (\Sigma_{SS}^{-1} \quad 0) \Omega_{I} X_i^\top \epsilon_{1i} + \Sigma_{SS}^{-1} S_i^\top \epsilon_{2_i}
\end{pmatrix}.
\]
For each $i$ let
\[
W_i = W_{1i} + W_{2i} = 
\begin{pmatrix}
    \Sigma_{SS}^{-1} S_i^\top \epsilon_{1i }\\
     (\Sigma_{SS}^{-1} \quad 0) \Omega_{I} X_i^\top \epsilon_{1i}
  \end{pmatrix}
+
\begin{pmatrix}
   \Sigma_{SS}^{-1} S_i^\top \epsilon_{2i}\\
  \Sigma_{SS}^{-1} S_i^\top \epsilon_{2_i}
  \end{pmatrix}.
\]
Then $\text{E}(W_{1i}) = \text{E}(W_{2i}) = 0$, and
\begin{align*}
  &\text{var}(W_{1i}) = \text{E}(W_{1i} W_{1i}^\top) \\
  =\,& \text{E}
  \begin{pmatrix}
    \Sigma_{SS}^{-1} S_i^\top \epsilon_{1i }\epsilon_{1i }^\top S_i  \Sigma_{SS}^{-1} 
    & \Sigma_{SS}^{-1} S_i^\top \epsilon_{1i }\epsilon_{1i }^\top X_i  \Omega_{I}^\top \begin{pmatrix}\Sigma_{SS}^{-1} \\ 0\end{pmatrix} \\
      (\Sigma_{SS}^{-1} \quad 0) \Omega_{I} X_i^\top \epsilon_{1i} \epsilon_{1i }^\top S_i  \Sigma_{SS}^{-1} 
      &  (\Sigma_{SS}^{-1} \quad 0) \Omega_{I} X_i^\top \epsilon_{1i} \epsilon_{1i }^\top X_i  \Omega_{I}^\top \begin{pmatrix}\Sigma_{SS}^{-1} \\ 0\end{pmatrix}
  \end{pmatrix} \\
  =\,&
  \begin{pmatrix}
    \sigma_1^2 \Sigma_{SS}^{-1} 
    & \sigma_1^2 \Sigma_{SS}^{-1} \Sigma_{SX} \Omega_{I}^\top \begin{pmatrix}\Sigma_{SS}^{-1} \\ 0\end{pmatrix} \\
      \sigma_1^2  (\Sigma_{SS}^{-1} \quad 0) \Omega_{I} \Sigma_{XS} \Sigma_{SS}^{-1} 
      & \sigma_1^2 (\Sigma_{SS}^{-1} \quad 0) \Omega_{I} \Sigma_{XX} \Omega_{I}^\top \begin{pmatrix}\Sigma_{SS}^{-1} \\ 0\end{pmatrix}
  \end{pmatrix},
\end{align*}
while
\[
\text{var}(W_{2i}) = \text{E}(W_{2i} W_{2i}^\top) \\
=
\text{E}
\begin{pmatrix}
  \Sigma_{SS}^{-1} S_i^\top \epsilon_{2i }\epsilon_{2i }^\top S_i  \Sigma_{SS}^{-1} 
  &  \Sigma_{SS}^{-1} S_i^\top \epsilon_{2i }\epsilon_{2i }^\top S_i  \Sigma_{SS}^{-1}  \\
  \Sigma_{SS}^{-1} S_i^\top \epsilon_{2i }\epsilon_{2i }^\top S_i  \Sigma_{SS}^{-1} 
  &   \Sigma_{SS}^{-1} S_i^\top \epsilon_{2i }\epsilon_{2i }^\top S_i  \Sigma_{SS}^{-1} 
\end{pmatrix} 
= 
\begin{pmatrix}
  \sigma_2^2 \Sigma_{SS}^{-1} 
  & \sigma_2^2 \Sigma_{SS}^{-1}\\
  \sigma_2^2 \Sigma_{SS}^{-1}
  & \sigma_2^2 \Sigma_{SS}^{-1}
\end{pmatrix}
\]
and $\text{cov}(W_{1i}, W_{2i}) = \text{E}(W_{1i} W_{2i}^\top) = 0$. Therefore $\text{var}(W_i)$ is given by
\begin{align*}
\Sigma_W
=\,&
\begin{pmatrix}
   \sigma_1^2 \Sigma_{SS}^{-1} 
& \sigma_1^2 \Sigma_{SS}^{-1} \Sigma_{SX} \Omega_{I}^\top \begin{pmatrix}\Sigma_{SS}^{-1} \\ 0\end{pmatrix} \\
   \sigma_1^2  (\Sigma_{SS}^{-1} \quad 0) \Omega_{I} \Sigma_{XS} \Sigma_{SS}^{-1} \quad\quad
& \sigma_1^2 (\Sigma_{SS}^{-1} \quad 0) \Omega_{I} \Sigma_{XX} \Omega_{I}^\top \begin{pmatrix}\Sigma_{SS}^{-1} \\ 0\end{pmatrix}
  \end{pmatrix} \,+\\
&
\begin{pmatrix}
   \sigma_2^2 \Sigma_{SS}^{-1} 
& \sigma_2^2 \Sigma_{SS}^{-1}\\
   \sigma_2^2 \Sigma_{SS}^{-1}
& \sigma_2^2 \Sigma_{SS}^{-1}
\end{pmatrix}.
\end{align*}

Since $W = n^{-1/2} \sum_{i=1}^{n} W_i$, by the central limit theorem $W \rightarrow N (0, \Sigma_W)$ as $n \rightarrow \infty$, so $(1, -1) W \rightarrow N (0, \Sigma_a)$ where
\begin{align*}
  \Sigma_a
  =\,&
  (1, -1) \Sigma_W \begin{pmatrix}1 \\ -1 \end{pmatrix} \\
  =\,&
  \sigma_1^2 \left\{ \Sigma_{SS}^{-1}  + (\Sigma_{SS}^{-1} \quad 0) \Omega_{I} \Sigma_{XX} \Omega_{I}^\top \begin{pmatrix}\Sigma_{SS}^{-1} \\ 0\end{pmatrix}
    - 2 \Sigma_{SS}^{-1} \Sigma_{SX} \Omega_{I}^\top \begin{pmatrix}\Sigma_{SS}^{-1} \\ 0\end{pmatrix} \right\}.
\end{align*}
Since 
\[
\Omega_{I}^\top = \Sigma_{XX}^{-1} D^\top = \Sigma_{XX}^{-1} \begin{pmatrix} \Sigma_{GS} & 0 \\ 0 & \Sigma_{SS}\end{pmatrix}
\]
and 
\begin{eqnarray*}
\Sigma_{XX}^{-1}=
\begin{pmatrix}
\Sigma_{GG} &\Sigma_{GS}\\
\Sigma_{SG} & \Sigma_{SS}
\end{pmatrix}^{-1}
=
\begin{pmatrix}
J^{-1} & -J^{-1}\Sigma_{GS}\Sigma_{SS}^{-1}\\
-\Sigma_{SS}^{-1}\Sigma_{SG}J^{-1} \quad & \Sigma_{SS}^{-1} + \Sigma_{SS}^{-1}\Sigma_{SG}J^{-1}\Sigma_{GS}\Sigma_{SS}^{-1}
\end{pmatrix},
\end{eqnarray*}
where $J=\Sigma_{GG}-\Sigma_{GS}\Sigma_{SS}^{-1}\Sigma_{SG}$, we have 
\begin{eqnarray*}
\Omega_{I}^\top \begin{pmatrix}\Sigma_{SS}^{-1} \\ 0 \end{pmatrix} 
&=&
\begin{pmatrix}
J^{-1} & -J^{-1}\Sigma_{GS}\Sigma_{SS}^{-1}\\
-\Sigma_{SS}^{-1}\Sigma_{SG}J^{-1} \quad & \Sigma_{SS}^{-1} + \Sigma_{SS}^{-1}\Sigma_{SG}J^{-1}\Sigma_{GS}\Sigma_{SS}^{-1}
\end{pmatrix}
 \begin{pmatrix} \Sigma_{GS} & \\ & \Sigma_{SS}\end{pmatrix} \begin{pmatrix}\Sigma_{SS}^{-1} \\ 0 \end{pmatrix} \\
&=&
\begin{pmatrix} 
J^{-1} \Sigma_{GS} \Sigma_{SS}^{-1} \\
- \Sigma_{SS}^{-1} \Sigma_{SG} J^{-1} \Sigma_{GS} \Sigma_{SS}^{-1} 
\end{pmatrix}
\end{eqnarray*}
Then we have 
\[
\Sigma_{SX} \Omega_{I}^\top \begin{pmatrix}\Sigma_{SS}^{-1} \\ 0 \end{pmatrix} 
= \begin{pmatrix} \Sigma_{SG} & \Sigma_{SS}\end{pmatrix} 
\begin{pmatrix} 
J^{-1} \Sigma_{GS} \Sigma_{SS}^{-1} \\
- \Sigma_{SS}^{-1} \Sigma_{SG} J^{-1} \Sigma_{GS} \Sigma_{SS}^{-1} 
\end{pmatrix}
= 0.
\]
Similarly, it can be shown that
\begin{eqnarray*}
  &&
(\Sigma_{SS}^{-1} \quad 0) \Omega_{I} \Sigma_{XX} \Omega_{I}^\top \begin{pmatrix}\Sigma_{SS}^{-1} \\ 0 \end{pmatrix}\\
&=&
( \Sigma_{SS}^{-1} \Sigma_{SG}  \quad 0) \Sigma_{XX}^{-1} \begin{pmatrix}\Sigma_{GS}\Sigma_{SS}^{-1} \\ 0\end{pmatrix} \\
&=&
( \Sigma_{SS}^{-1} \Sigma_{SG}  \quad 0) 
\begin{pmatrix}
J^{-1} & -J^{-1}\Sigma_{GS}\Sigma_{SS}^{-1}\\
-\Sigma_{SS}^{-1}\Sigma_{SG}J^{-1}  & \Sigma_{SS}^{-1} + \Sigma_{SS}^{-1}\Sigma_{SG}J^{-1}\Sigma_{GS}\Sigma_{SS}^{-1}
\end{pmatrix}
\begin{pmatrix}\Sigma_{GS}\Sigma_{SS}^{-1} \\ 0\end{pmatrix}  \\
&=&
\Sigma_{SS}^{-1} \Sigma_{SG} J^{-1}\Sigma_{GS}\Sigma_{SS}^{-1} \\
&=&
\Sigma_{SS}^{-1} \Sigma_{SG} 
( \Sigma_{GG}-\Sigma_{GS}\Sigma_{SS}^{-1}\Sigma_{SG} )^{-1}
\Sigma_{GS}\Sigma_{SS}^{-1}.
\end{eqnarray*} 
Therefore $\tilde\alpha_1$ and $\hat a$ are asymptotically equivalent.

\subsection{\label{compare_iv}Comparing with instrumental variable problems}

On its surface, our mediation setting is similar to instrumental variable models. Our model \eqref{eq:model}
\[
Y_i = G_i^\top \alpha_0 + S_i^\top \alpha_1 + \epsilon_{1i},
\quad
G_i = \gamma S_i + E_{i},
\]
is analogous to an instrumental variable model where the $S_i$ are analogous to the instrumental variables. However, we study a different problem than the instrumental variable literature. The latter is interested in estimating $\alpha_0$ when ordinary regression fails because the variable of interest $G_i$ is correlated with the error term $\epsilon_{1i}$, using instrumental variables $S_i$. On the contrary, we are interested in estimating the indirect effect of $S_i$ on $Y_i$, and we do not have an endogeneity issue. Ordinary least squares would work fine in our complete mediation setting.

In our Theorems \ref{thm:incomplete} and \ref{thm:complete}, we require that $\Sigma_{SG}$ and $\alpha_0$ are both nonzero, and at the end of Section \ref{subsec:incomplete} in the main text we discuss why our test can be conservative when the exposure, the mediator, and the outcome are only weakly associated. This seems similar to the weak instrumental variable problem, where $\gamma$, related to our $\Sigma_{SG}$, is close to zero. If the instrumental variable is weak, then the limiting distribution of the proposed estimator will be distorted, leading to biased estimatoin and invalid inference. However, the case where both $\Sigma_{SG}$ and $\alpha_0$ are close to zero in our mediation problem is actually not as bad as the weak IV problem, because although conservative, our proposed method is remains valid.

\section{Optimal variance}
\label{optimal_variance}

Under certain conditions, the asymptotic variance of our proposed estimator for the indirect effect under complete mediation achieves the asymptotic Cram\'{e}r-Rao lower bound. Under complete mediation, model \eqref{eq:model} in the main text becomes 
\begin{equation}
  \label{eq:model2}
  Y_i = G_i^\top \alpha_0 + \epsilon_{1i},
  \quad
  G_i = \gamma S_i + E_{i}.
\end{equation}
We assume here that $p$ is fixed and that $\epsilon_1$ and $E$ are both normally distributed.

Let $\theta = (\gamma_0^\top, \alpha_0^\top)^\top$, so that $\theta$ is a $2p \times 1$ vector containing all parameters in model \eqref{eq:model2}. Then we can view $\beta_0$ as a function of $\theta$. To see this, $\gamma_0 = V \theta$ where
  \[
  V = 
  \begin{pmatrix}
    1 & 0 & \dots & 0  & 0 & 0 & \dots & 0 \\
    0 & 1 & \dots & 0  & 0 & 0 & \dots & 0\\
    &  & \ddots & & &  & \ddots \\
    0  & 0 & \dots & 1 & 0 & 0  & \dots & 0
  \end{pmatrix}
  =
  \begin{pmatrix}
    I_p & O_p
  \end{pmatrix},
  \]
  where $I_p$ is a $p \times p$ identity matrix and $O_p$ is a $p \times p$ matrix of zeros. Similarly, $\alpha_0 = U \theta$, where $U = (O_P \, I_p)$. Thus we have $g(\theta) = \theta^\top V^\top U \theta = \beta_0$.

\begin{theorem}
\label{sup:thm:optimal_var}
Suppose model \eqref{eq:model2} holds, and assume that $\epsilon_{1i}$ and $E_i$ are both normally distributed. Let $T_n$ be any estimator sequence of $g(\theta)$ such that for any fixed $\theta$ and $h \in \mathbb{R}^{k}$ ,
\begin{equation}
\label{condition:limit}
\sqrt{n} (T_n - g(\theta + h / \sqrt{n}))  \rightarrow L_{\theta}, \text{ every } h,
\end{equation}
where $L_\theta$ has mean zero and covariance matrix $\Sigma_\theta$. Our estimator $\tilde{b}$ \eqref{eq:complete} has the smallest variance among all estimators that satisfy \eqref{condition:limit}.
\end{theorem}

\begin{proof}
  By Theorem 8.5 in \citep{van2000asymptotic}, $\Sigma_\theta - g'(\theta) I(\theta)^{-1} g'(\theta)$ is non-negative definite, where $I(\theta)$ is the Fisher information matrix at point $\theta$. It remains to showthat $\tilde{b}$ has variance $g'(\theta) I(\theta)^{-1} g'(\theta)$.
  
Based on model \eqref{eq:model2} in the Supplementary Materials, the joint density function of a sample of $n$ observations is
\begin{eqnarray*}
f(Y, G, S) &\propto& f(Y \mid G, S) f(G \mid S) \\
&\propto&
\exp \left\{ -\frac{1}{2 \sigma_1^2} \sum_{i=1}^{n} (Y_i - G_i \alpha_0)^2 \right\} \exp \left\{ -\frac{1}{2} \sum_{i=1}^{n} (G_i - S_i \gamma_0^\top) \Sigma_E^{-1} (G_i - S_i \gamma_0^\top)^\top \right\}.
\end{eqnarray*}
Therefore the log-likelihood function equals
$$
l(\theta ; Y, G, S) \propto - \frac{1}{2 \sigma_1^2} \sum_{i=1}^{n} G_i U \theta \theta^\top U^\top G_i^\top 
- \frac{1}{2} \sum_{i=1}^{n} S_i \theta^\top V^\top \Sigma_E^{-1} V \theta S_i
+ h(\theta)
$$
where $h(\theta)$ only involves linear functions of $\theta$.

Then we have
\begin{eqnarray*}
\frac{\partial ^2 l(\theta ; Y, G, S)}{\partial \theta^2} 
&=&
- \sum_{i=1}^{n} \left( \frac{1}{\sigma_1^2} U^\top G_i^\top G_i U + S_i^2 V^\top \Sigma_E^{-1} V \right),
\end{eqnarray*}
so the expected Fisher information matrix $I_n(\theta)$ of $\theta$ is 
\begin{eqnarray*}
  - \text{E} \frac{\partial ^2 l(\theta ; Y, G, S)}{\partial \theta^2}
  =
  \frac{n}{\sigma_1^2} U^\top \text{E} ( G_i^\top G_i) U + n \text{E}(S_i^\top S_i)  V^\top \Sigma_E^{-1} V
  =
  n
  \begin{pmatrix}
    \Sigma_{SS} \Sigma_E^{-1} & 0 \\
    0 & \sigma_1^{-2} \Sigma_{GG}.
  \end{pmatrix}
\end{eqnarray*}

Then we have
\begin{eqnarray*}
g'(\theta) I_n(\theta)^{-1} g'(\theta)
&=&
\frac{1}{n} \theta^\top (V^\top U + U^\top V) 
\begin{pmatrix}
\Sigma_{SS}^{-1} \Sigma_E & 0 \\
0 & \sigma_1^2 \Sigma_{GG}^{-1}
\end{pmatrix}
(V^\top U + U^\top V)  \theta \\
&=&
\frac{1}{n} (\alpha_0^\top, \gamma_0^\top) 
\begin{pmatrix}
\Sigma_{SS}^{-1} \Sigma_E^{-1} & 0 \\
0 & \sigma_1^2 \Sigma_{GG}^{-1}
\end{pmatrix}
\begin{pmatrix}
\alpha_0 \\
\gamma_0
\end{pmatrix} \\
&=&
\frac{1}{n}  ( \Sigma_{SS}^{-1} \alpha_0^\top \Sigma_E \alpha_0 + \sigma_1^2 \gamma_0 \Sigma_{GG}^{-1} \gamma_0^\top )
\end{eqnarray*}
The residual error $\epsilon_{2i}$ equals $E_i^\top \alpha_0$ in the notation of model~\eqref{eq:model2} from these Supplementary Materials, so $\sigma_2^2 = \alpha_0^\top \Sigma_E \alpha_0$ and $\Sigma_{SG} = \text{cov}(S_i, G_i) = \Sigma_{SS} \gamma_0^\top$. Therefore 
$
g'(\theta) I_n(\theta)^{-1} g'(\theta) = 
n^{-1} (\sigma_2^2 \Sigma_{SS}^{-1} + \sigma_1^2 \Sigma_{SS}^{-1} \Sigma_{SG} \Sigma_{GG}^{-1} \Sigma_{GS} \Sigma_{SS}^{-1}),
$ 
which is the variance of $\tilde{b}$ by Theorem \ref{thm:complete} in the main text.
\end{proof}

\section{Additional simulation results}

\subsection{\label{sec:add_incom}Incomplete mediation with $\alpha_1 = 0.5$}

\begin{figure}[h!]
  \centering
  \includegraphics[width = \textwidth]{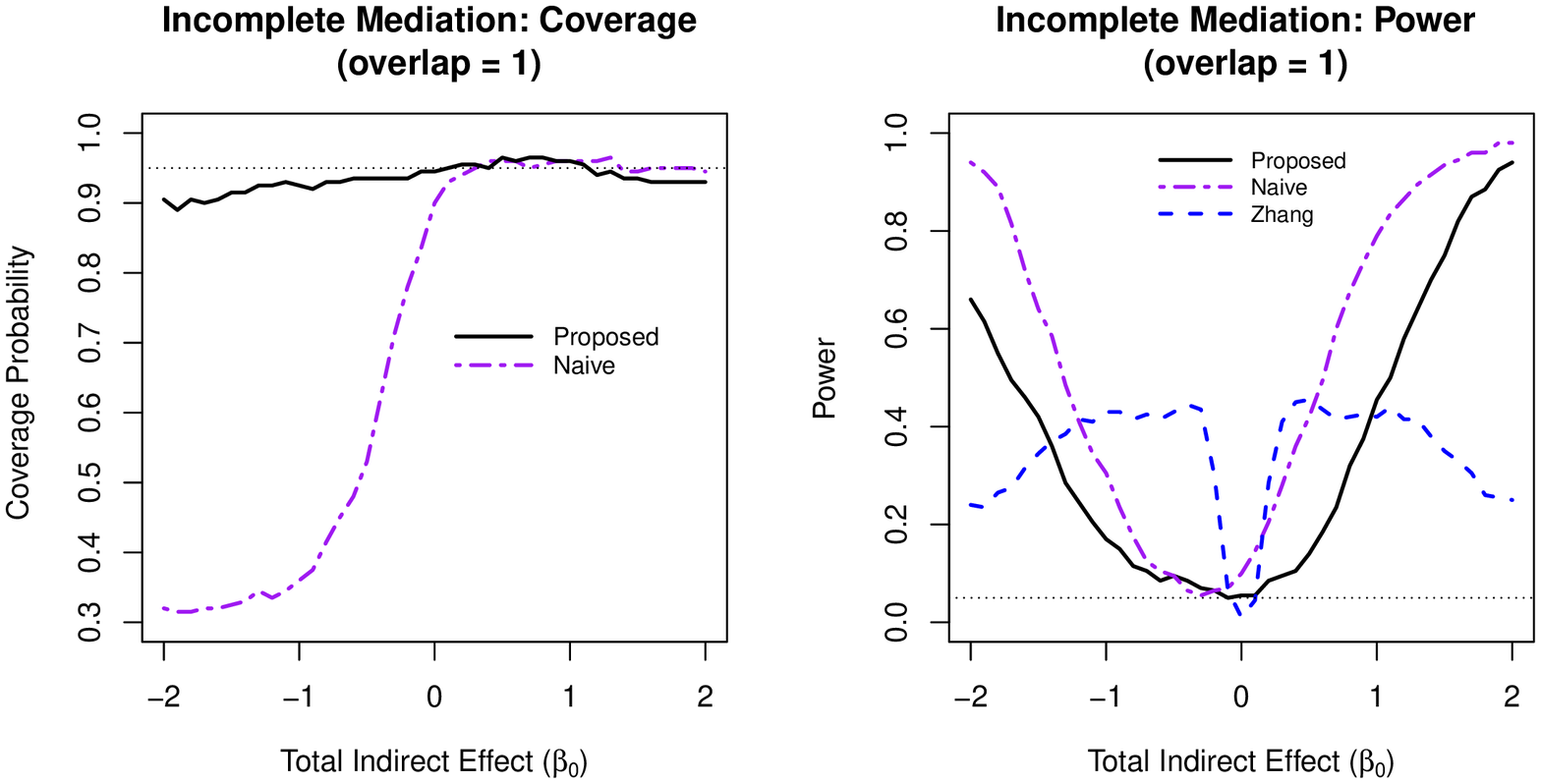}\par
  \includegraphics[width = \textwidth]{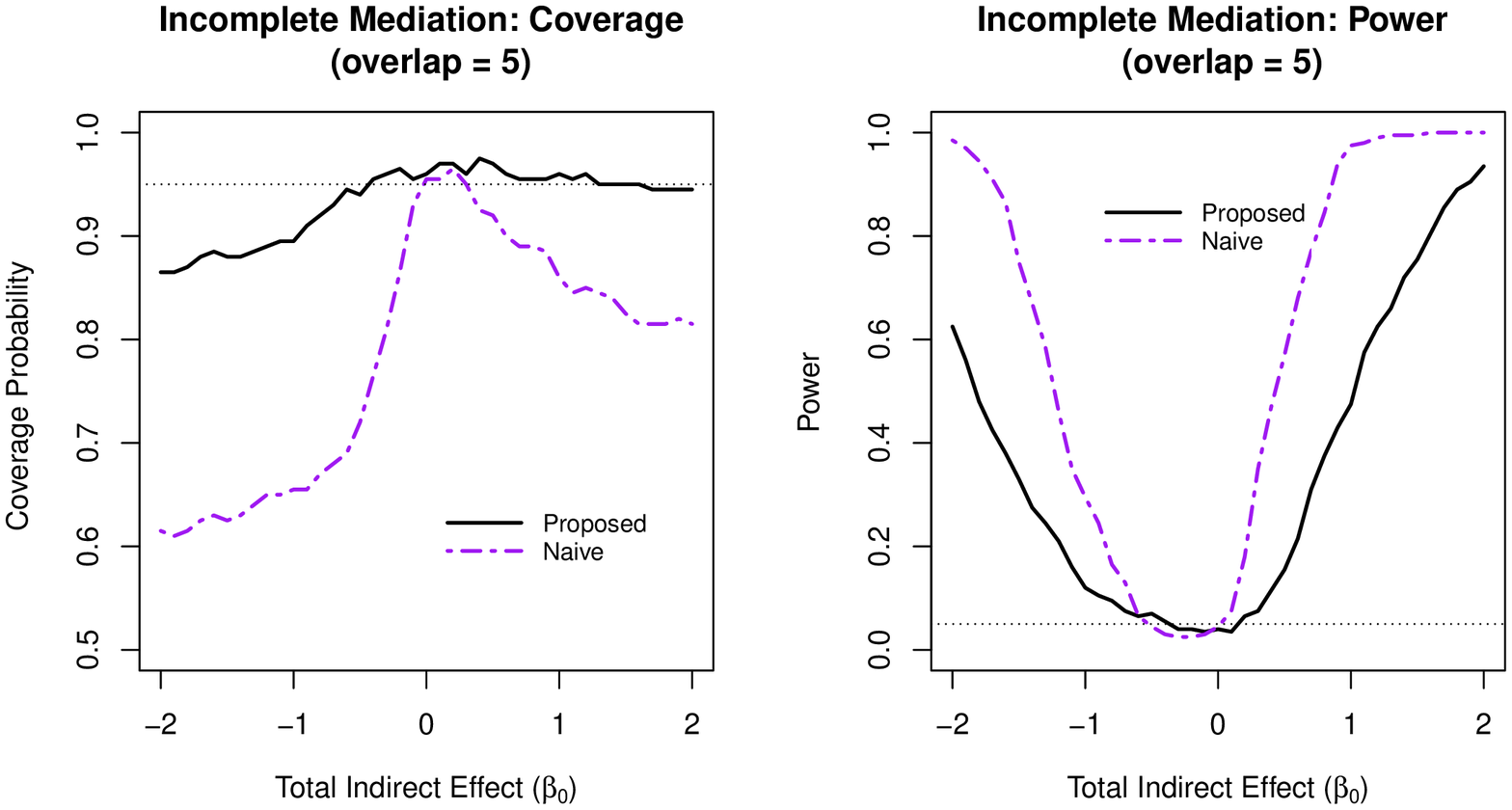}\par
\caption{\label{fig:incomplete_2} Average coverage probabilities of 95\% confidence intervals (left panels) and average power curves at significance level $\alpha = 0.05$ (right panels) for estimating and testing the indirect effect under incomplete mediation, over 200 replications. The direct effect was $\alpha_1 = 0.5$. The number of true mediators was 1 in the upper panels and 5 in the lower panels. Proposed: $\hat{b}$ from~\eqref{eq:incomplete}; Naive: the naive method discussed in Section \ref{sec:methods_compared}; Zhang: method of \citet{zhang2016estimating}.}
\end{figure}

Here we present simulation results for incomplete mediation under the same settings as Section~\ref{sec:incomplete_sims} of the main text, but with the direct effect $\alpha_1 = 0.5$. Figure~\ref{fig:incomplete_2} reports the coverage probabilities of the 95\% confidence intervals as well as the average power curves at nominal 0.05 significance level, over 200 replications. The results show that the coverage of our method deteriorates somewhat compared to when $\alpha_1 = 0.1$, while the coverage of the naive method dropped dramatically. The power curves show that all tests were able to maintain the nominal significance level, though the naive test cannot be implemented in practice because it is difficult to estimate the standard error of the naive indirect effect estimator. The coverage probability and the power curves were not always symmetric because changing $\beta_0$ also changed the degree of multicollinearity between $G_i$ and $S_i$ in the regression model for $Y_i$ in model~\eqref{eq:model}. The drops in the coverage probabilities may be due to this increased collinearity between $G_i$ and $S_i$, and is worth further investigation in the future.

\subsection{\label{sec:add_com}Complete mediation with $p = 1000$}

Here we present simulation results for complete mediation under the same simulation settings as Section~\ref{sec:complete_sims} of the main text, but with $p=1000$ potential mediators. Figure~\ref{fig:complete_2} reports the coverage probabilities of the 95\% confidence intervals as well as the average power curves at nominal 0.05 significance level, over 200 replications. The power curves shows that all tests were able to maintain the nominal significance level, and our method had higher power than ordinary least squares for sufficiently large $\beta_0$, consistent with Proposition~\ref{p:var}. However, the coverage curves shows that the coverage of our method suffers for large $p$. This is not unexpected, as even in the standard linear regression setting, low coverage has been observed for debiased lasso estimators in high dimensions, for example in the simulations of \citet{van2014asymptotically}. We suspect that our coverage suffers for similar reasons, but leave a careful analysis of this problem for future work.

\begin{figure}[h!]
  \centering
  \includegraphics[width = \textwidth]{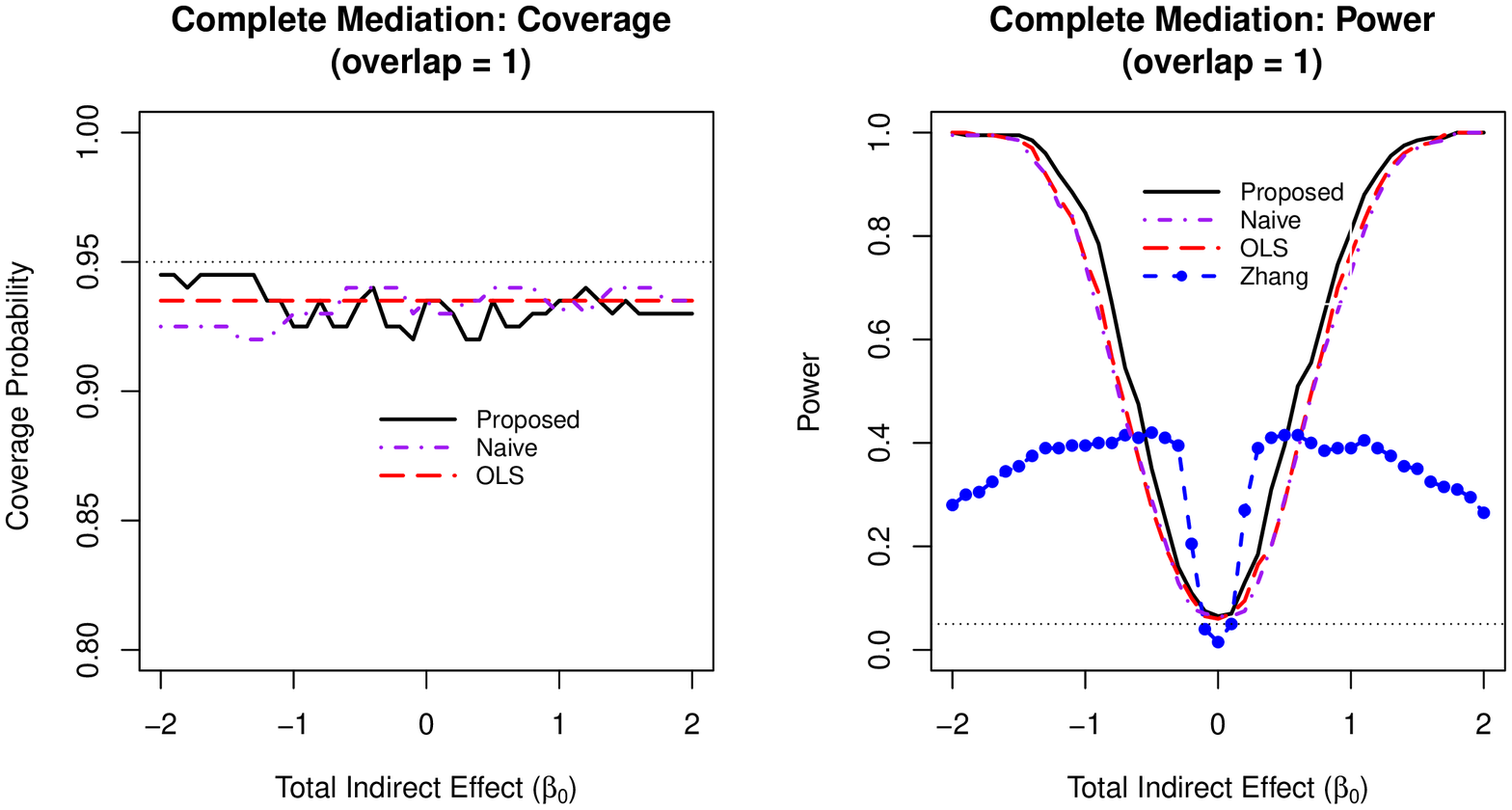}\par
  \includegraphics[width = \textwidth]{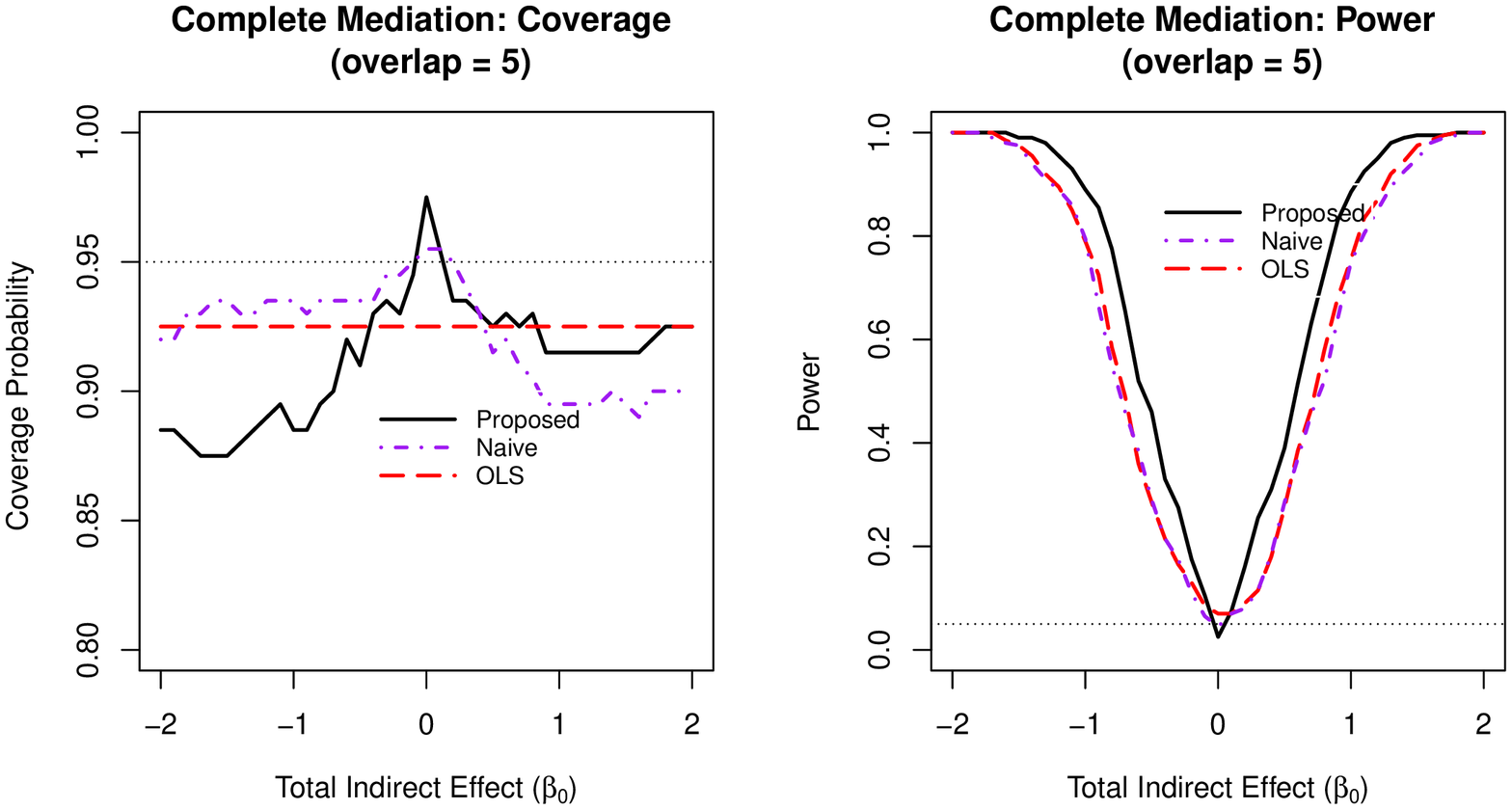}\par
\caption{\label{fig:complete_2}Average coverage probabilities of 95\% confidence intervals (left panels) and average power curves at significance level $\alpha = 0.05$ (right panels) for estimating and testing the indirect effect under complete mediation, over 200 replications. The number of true mediators was 1 in the upper panels and 5 in the lower panels, with 1000 potential mediators; Proposed: $\tilde{b}$ from~\eqref{eq:complete}; Naive: the naive method discussed in Section \ref{sec:methods_compared}; OLS: ordinary least squares estimate; Zhang: the method of \citet{zhang2016estimating}.}
\end{figure}

\subsection{\label{sec:length}Length of interval estimates}

Proposition~\ref{p:var} in the paper shows that our proposed estimator $\tilde{b}$ under complete mediation always has equal or lower asymptotic variance compared to the ordinary least squares estimator. Since both our estimator and ordinary least squares estimator are asymptotically normal and unbiased, this likely translates to smaller expected lengths of confidence intervals of our proposed estimator $\tilde{b}$ compared to the ordinary least squares under the same nominal level. Figure~\ref{fig:length} reports the average lengths of 95\% confidence intervals of our proposed $\tilde b$ and the ordinary least squares estimator $\tilde b_{OLS}$ over 200 replications, under the same complete mediation simulation settings considered in Section \ref{sec:complete_sims} in the main text. Indeed, our proposed method has shorter lengths of intervals than ordinary least squares in every case.

\begin{figure}[h!]
  \centering
  \includegraphics[width = \textwidth]{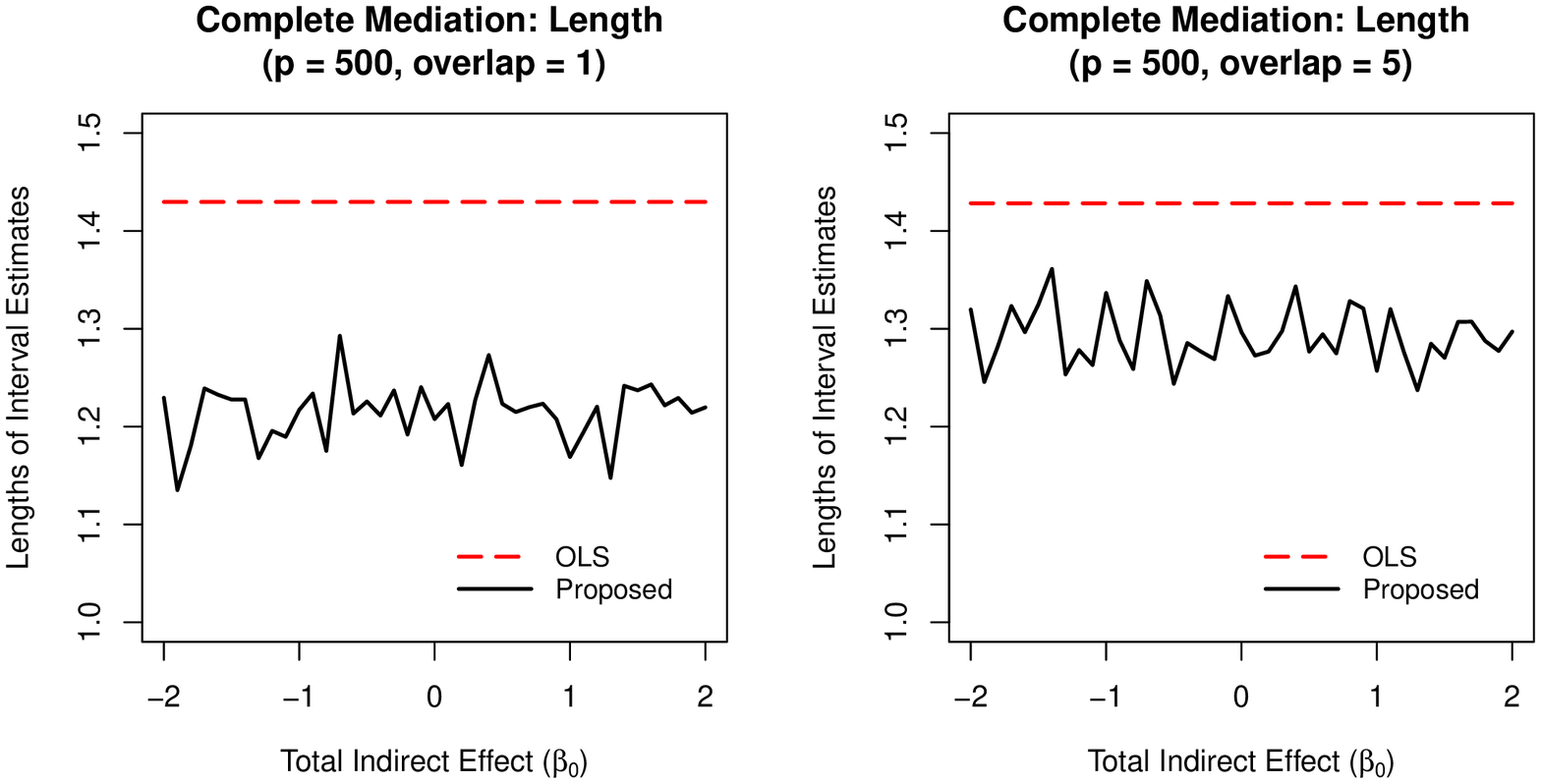}\par
  \includegraphics[width = \textwidth]{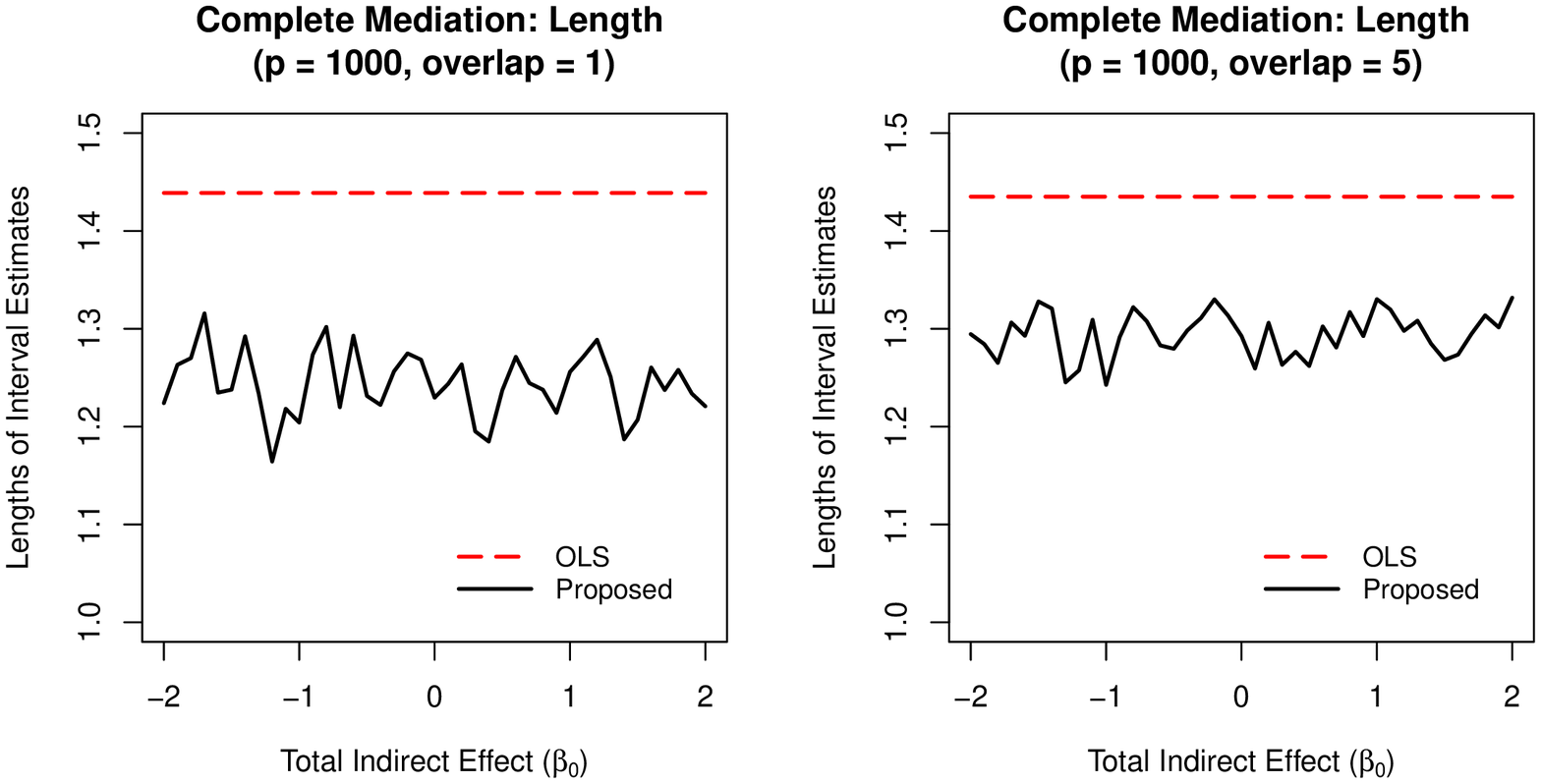}\par
\caption{\label{fig:length} Average lengths of 95\% confidence intervals for testing the indirect effect under complete mediation for estimating and testing the indirect effect under complete mediation, over 200 replications. Upper panel: 500 potential mediators with 1 and 5 true mediators, respectively; Lower panel: $1,000$ potential mediators with 1 and 5 true mediators, respectively. Proposed: $\tilde{b}$ from~\eqref{eq:complete}; OLS: ordinary least squares estimate.}
\end{figure}

\subsection{\label{risk}Risk}

In the paper we mainly focused on the power of our method compared to ordinary least squares, as well as the naive method and the \citet{zhang2016estimating} method, under the complete mediation setting. Here, we also compare the risks of those methods under squared error loss. The simulation setting we are using here are the same as those in Section~\ref{sec:complete_sims} in the main text. Figure~\ref{fig:risk} compares average risk curves of our $\tilde{b}$ and ordinary least squares, in our four simulation examples, over 200 replications. For most values of $\beta_0$ in most of the settings, our method has a smaller risk than ordinary least squares. The risk of our method increased with the number of true mediators, and exceeded the risk of ordinary least squares for certain $\beta_0$, but the difference was not large. The risk behavior in Example 4 of Figure \ref{fig:risk} is interesting and may be worth future study.

\begin{figure}[h!]
  \centering
  \includegraphics[width = \textwidth]{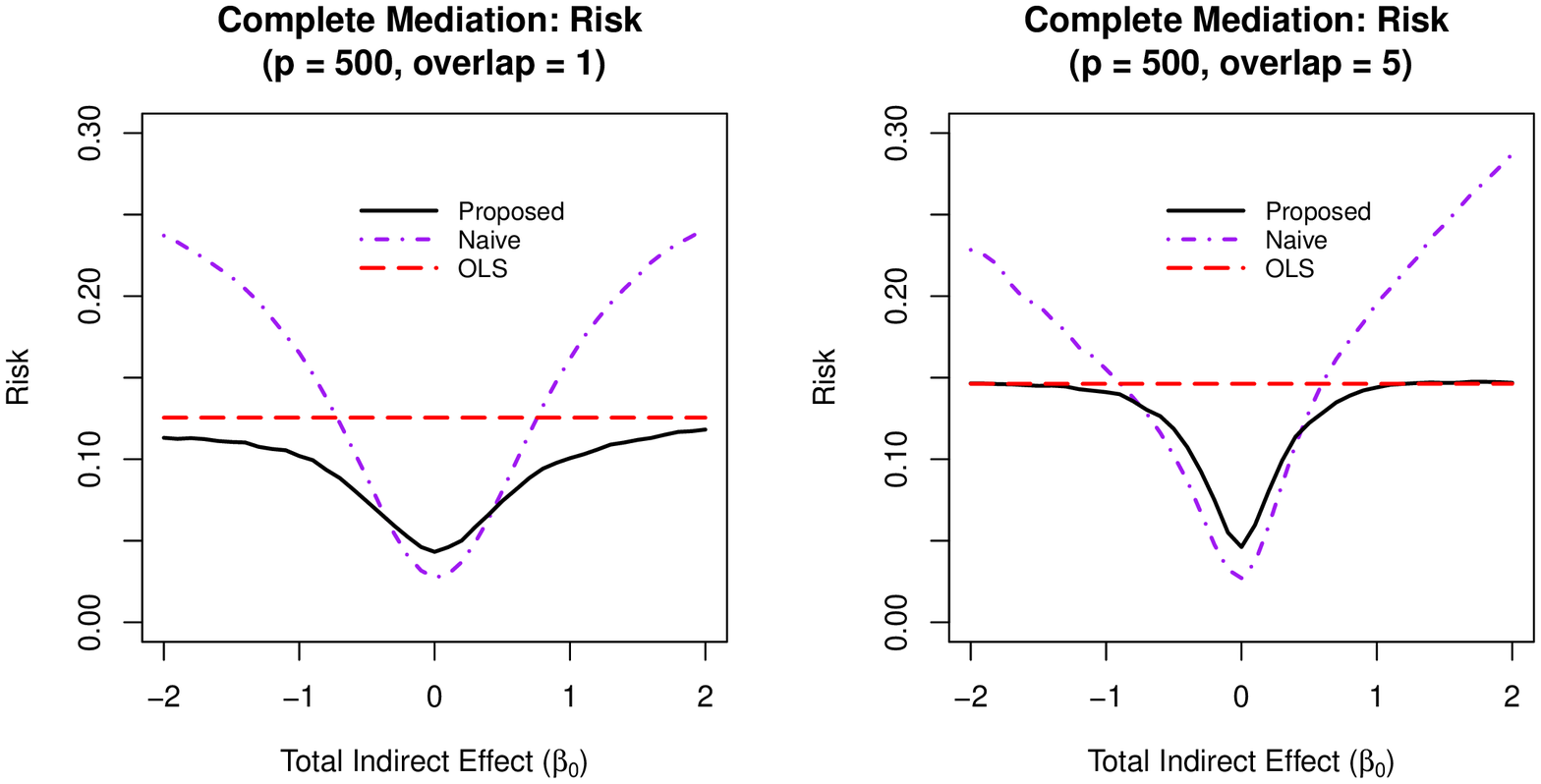}\par
  \includegraphics[width = \textwidth]{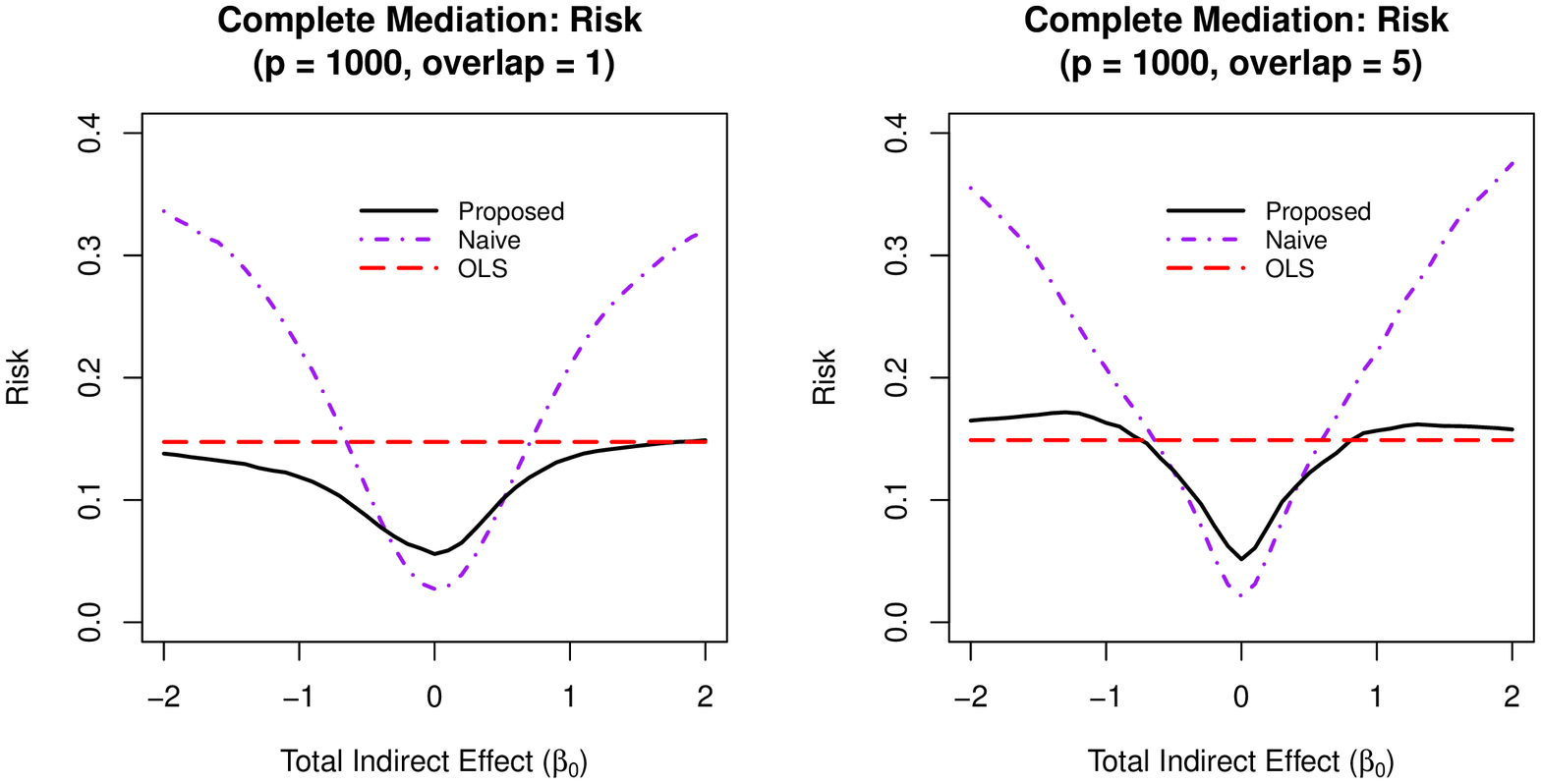}\par
\caption{\label{fig:risk}Average risk curves for testing the indirect effect under complete mediation at significance level $\alpha = 0.05$ for estimating and testing the indirect effect under complete mediation, , over 200 replications. Upper panel: 500 potential mediators with 1 and 5 true mediators, respectively; Lower panel: $1,000$ potential mediators with 1 and 5 true mediators, respectively. Proposed: $\tilde{b}$ from~\eqref{eq:complete}; OLS: ordinary least squares estimate.}
\end{figure}

In Figure \ref{fig:risk}, we did not compare to the risk of the procedure of \citet{zhang2016estimating} because the method does not provide a point estimate for the total indirect effect, and only estimates the separate indirect effects through each of the mediators. We therefore added these individual estimates to give an \textit{ad hoc} point estimate for $\beta_0$. Specifically, we add the estimates from mediators with significance $p$-values smaller than 0.05. Figure~\ref{fig:zhang_risk_ex1} compares the risk curves between the calculated estimates from \citet{zhang2016estimating} method, our $\tilde{b}$, the naive method and ordinary least squares, in simulations with 500 potential mediators and one true mediator, over 200 replications. We can see that the risk of the \citet{zhang2016estimating} method is much larger than the risk of all other methods. Another way to estimate the total indirect effect for the procedure of \citet{zhang2016estimating} is to add up the indirect effects from all selected mediators, and we found that the result is similar.

\begin{figure}
\label{fig:zhang_risk_ex1}
  \centering
  \includegraphics[width = 10cm]{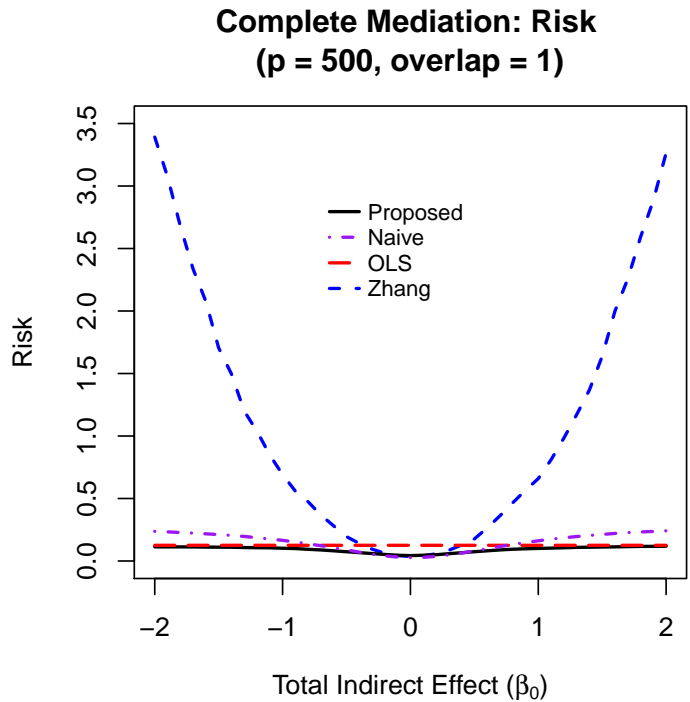}
\caption{Average risk curves for testing the indirect effect under complete mediation at significance level $\alpha = 0.05$ for estimating and testing the indirect effect under complete mediation, over 200 replications. Simultions contained 500 potential mediators and one true mediator. Proposed: $\tilde{b}$ from~\eqref{eq:complete}; Naive: the naive method discussed in Section \ref{sec:methods_compared}; OLS: ordinary least squares estimate; Zhang: the method of \citet{zhang2016estimating}.}
\end{figure}

\subsection{\label{sec:effectnoise}Effect of noise}

Because we have a two-stage linear mediation model, we have two sources of noise: $\epsilon_{1i}$ and $\epsilon_{2i}$, from \eqref{eq:model} in the main text. In our simulations we generated data following $G_i = c \gamma S_i + E_i$, where $E_i \sim N(0, \Sigma_E)$, so $\text{var}(\epsilon_{2i}) = \sigma_2^2 = \alpha_0^\top \Sigma_E \alpha_0$. In this section we study the effect of varying $\Sigma_E$ and $\text{var}(\epsilon_{1i}) = \sigma_1^2$ on the power of our proposed approach in the complete mediation setting. We used the same simulation settings as in Section~\ref{sec:complete_sims}.

\begin{figure}[h!]
 \centering
  \includegraphics[width = \textwidth]{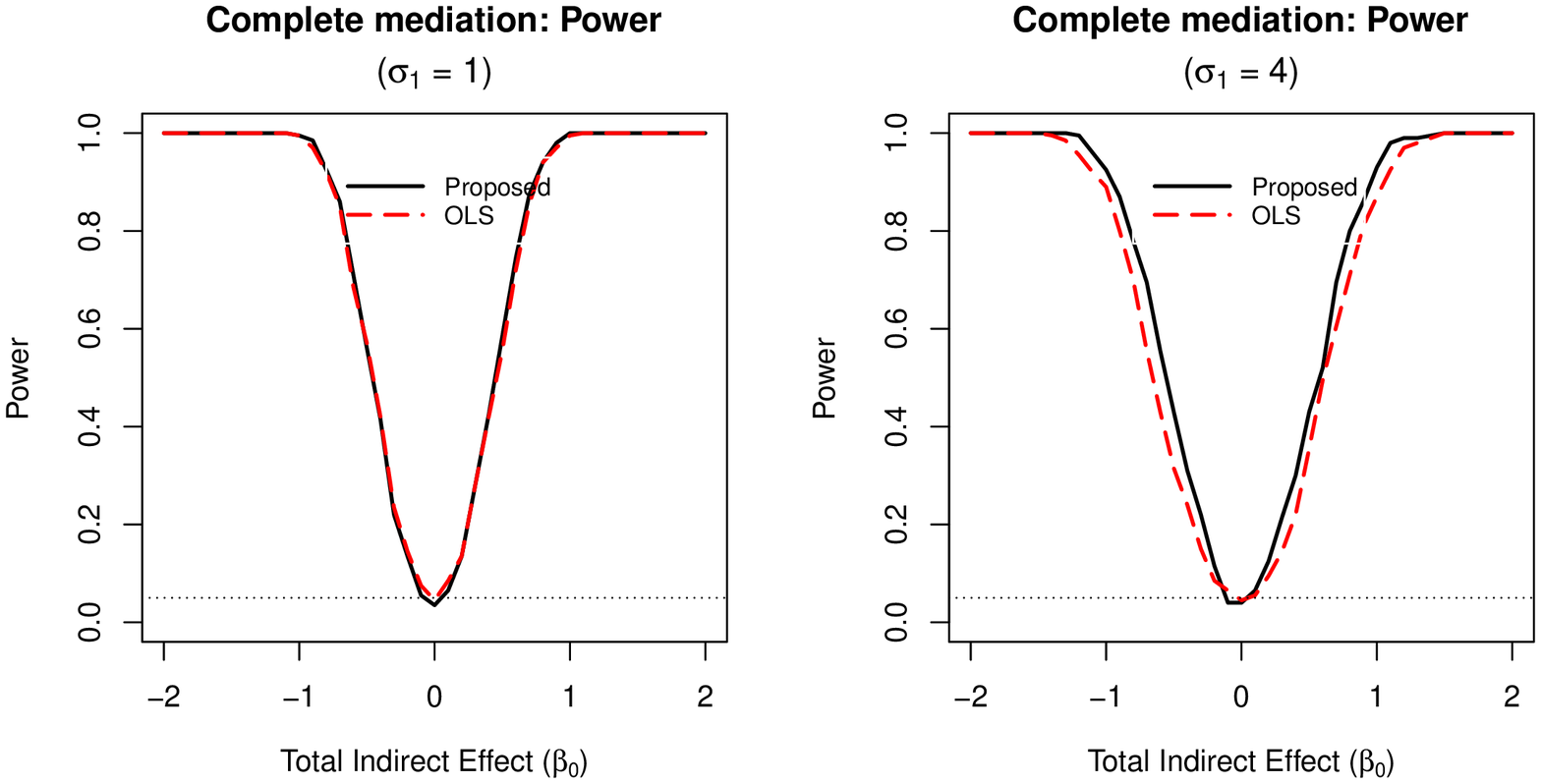}\par
  \includegraphics[width = \textwidth]{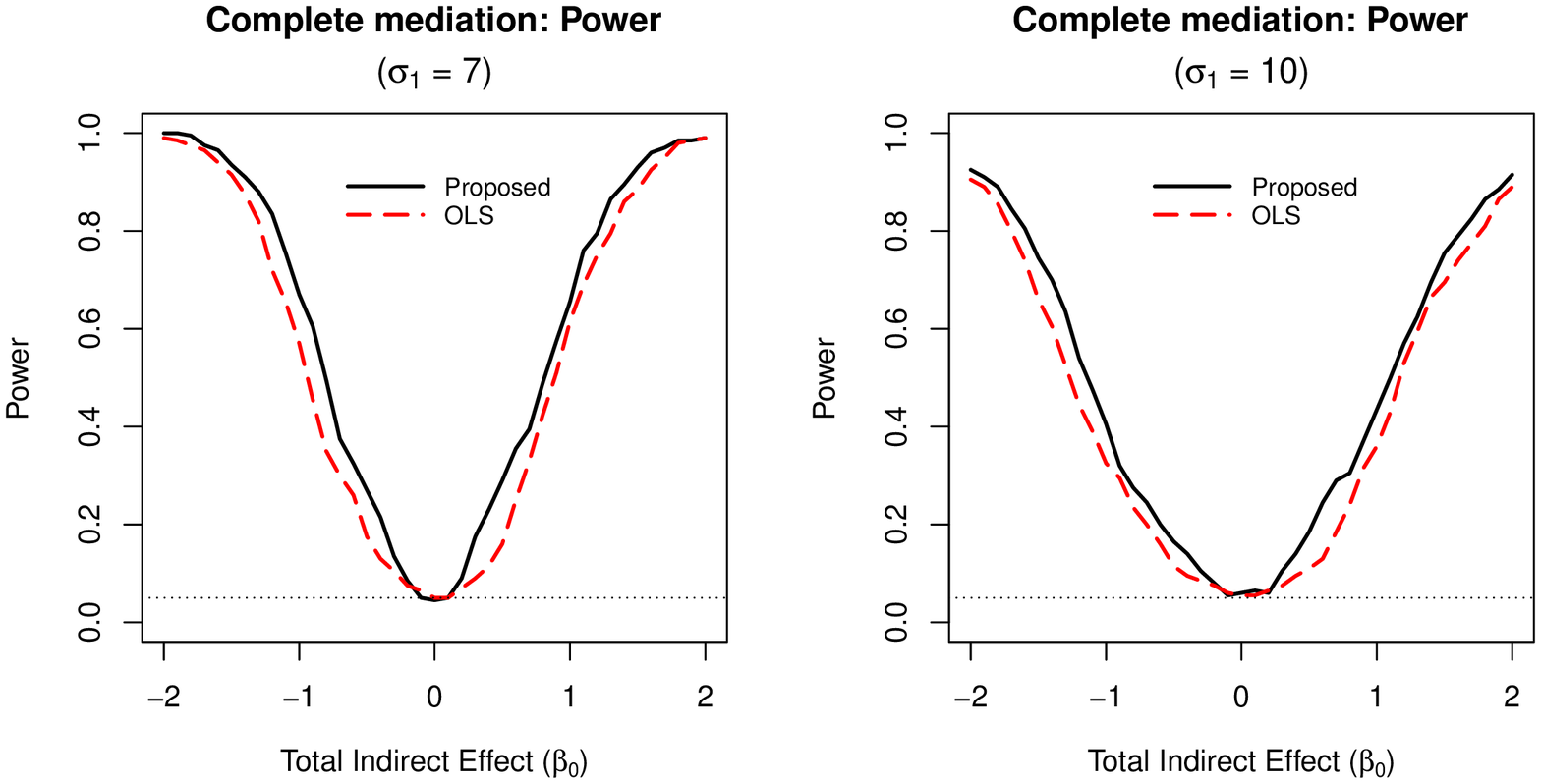}\par 
\caption{\label{fig:sig1}Average power curves for testing the indirect effect under complete mediation at significance level $\alpha = 0.05$ for estimating and testing the indirect effect under complete mediation, over 200 replications, for different values of $\sigma_{1}$. All examples had 500 potential mediators with 1 true mediator and the diagonal entries of $\Sigma_E$ were fixed at 1. Proposed: $\tilde{b}$ from~\eqref{eq:complete}; OLS: ordinary least squares estimate.}
\end{figure}

\begin{figure}[h!]
  \includegraphics[width = \textwidth]{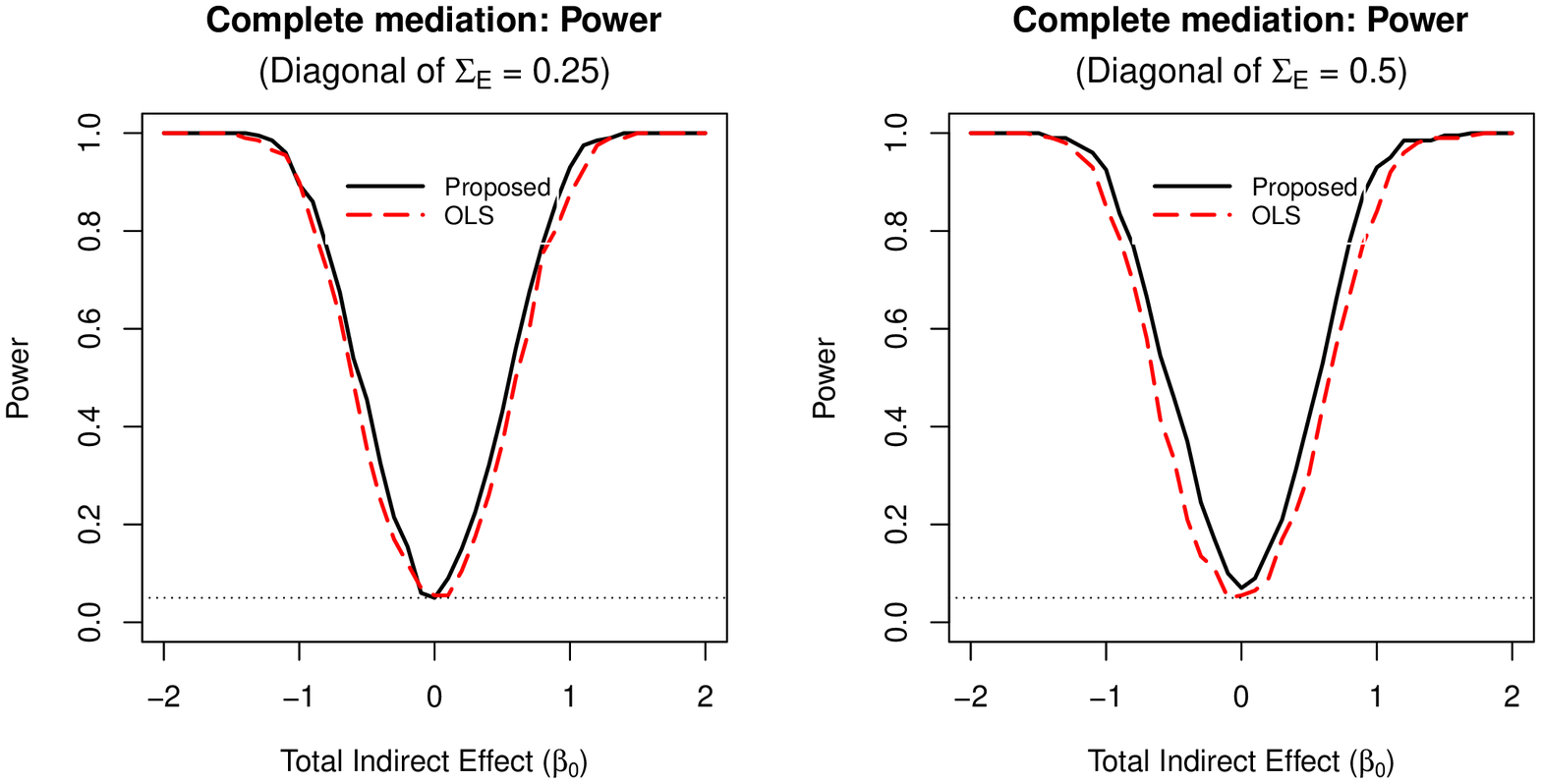}\par
  \includegraphics[width = \textwidth]{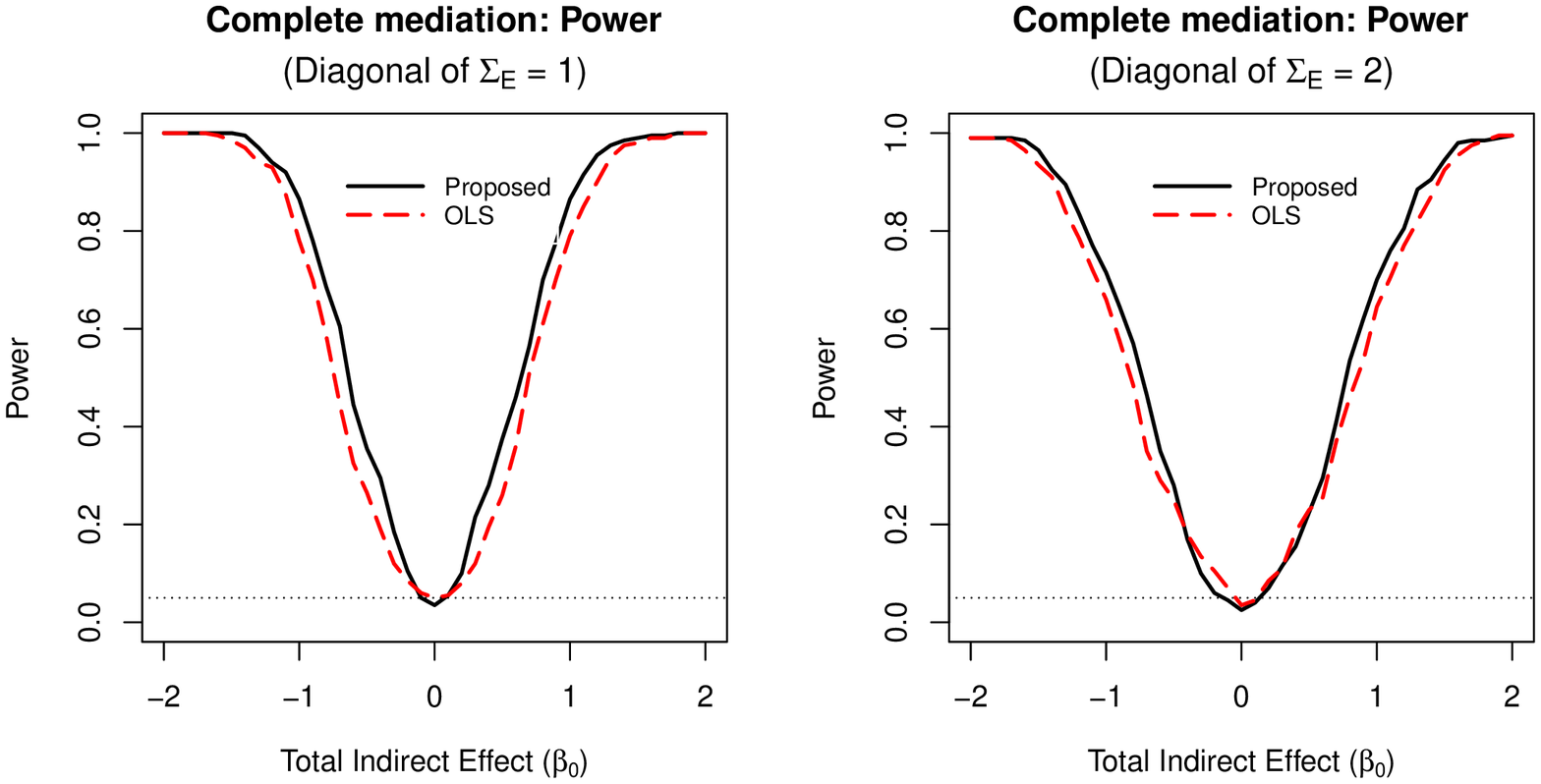}\par
\caption{\label{fig:sig2}Average power curves for testing the indirect effect under complete mediation at significance level $\alpha = 0.05$ for estimating and testing the indirect effect under complete mediation, over 200 replications, for different values of diagonal of $\Sigma_E$. All examples had 500 potential mediators with 1 true mediator and $\sigma_1 = 5$. Proposed: $\tilde{b}$ from~\eqref{eq:complete}; OLS: ordinary least squares estimate.}
\end{figure}

Figure~\ref{fig:sig1} reports power curves for different values of $\sigma_1$ when $\Sigma_E$ was fixed with diagonal entries equal to 1. As $\sigma_1$ increased, our method consistently outperformed ordinary least squares, even for very large $\sigma_1$. This is sensible because our method denoises the outcome using the observed mediators, and the advantage of doing so becomes more apparent the larger the noise. From another point of view, Proposition~\ref{p:var} in the main text states that the the difference of asymptotic variance between ordinary least squares and our method converges to a positive semi-definite matrix $\sigma_1^2\Sigma_{SS}^{-1}(\Sigma_{SS}-\Sigma_{SG}\Sigma_{GG}^{-1}\Sigma_{SG}^\top)\Sigma_{SS}^{-1}$, which increases in norm as $\sigma_1$ increases. Both of methods became less powerful as $\sigma_1$ increased, but the relative power advantage of our method compared to ordinary least squares remains large.

Figure \ref{fig:sig2} shows that when $\sigma_1$ was fixed and $\Sigma_E$ was varied, our method had the greatest advantage relative to ordinary least squares when $\Sigma_E$ was moderate. When $\Sigma_E$ was too large or too small, our procedure behaved more like ordinary least squares. This is consistent with Proposition \ref{p:var} in the main text and makes intuitive sense as well. Our method should have the greatest relative advantage when $S$ and $G$ are correlated but only weakly so, or else $G$ would offer no additional information. This occurs when $\Sigma_E$ is not too large. When $\Sigma_E = 0$, the indirect effect becomes zero as well, so there is no effect to detect.

\subsection{\label{sec:nongaussian}Non-Gaussian errors}

In our simulations in Section \ref{sec:numerical} we generated $\epsilon_{1i}$ and $E_i$ from Gaussian distributions for simplicity, but as indicated in Theorem~\ref{thm:incomplete} and ~\ref{thm:complete}, we do not need the Gaussian assumptions. Here we report simulation results when the errors were $t$-distributed. Specifically, we adopt the same simulation settings from Example 1 in Section~\ref{sec:complete_sims}, but instead of generating $\epsilon_{1i} \sim N(0, 25)$, we now generate $\epsilon_{1i}$ from $t_3 \cdot (25/3)^{1/2}$, where $t_3$ denotes a $t$-distributed random variable with three degrees of freedom. The normalizing constant $(25/3)^{1/2}$ is such that $\text{var}(\epsilon_{1i})$ still equals $25$.

\begin{figure}
\centering
  \includegraphics[width = \textwidth]{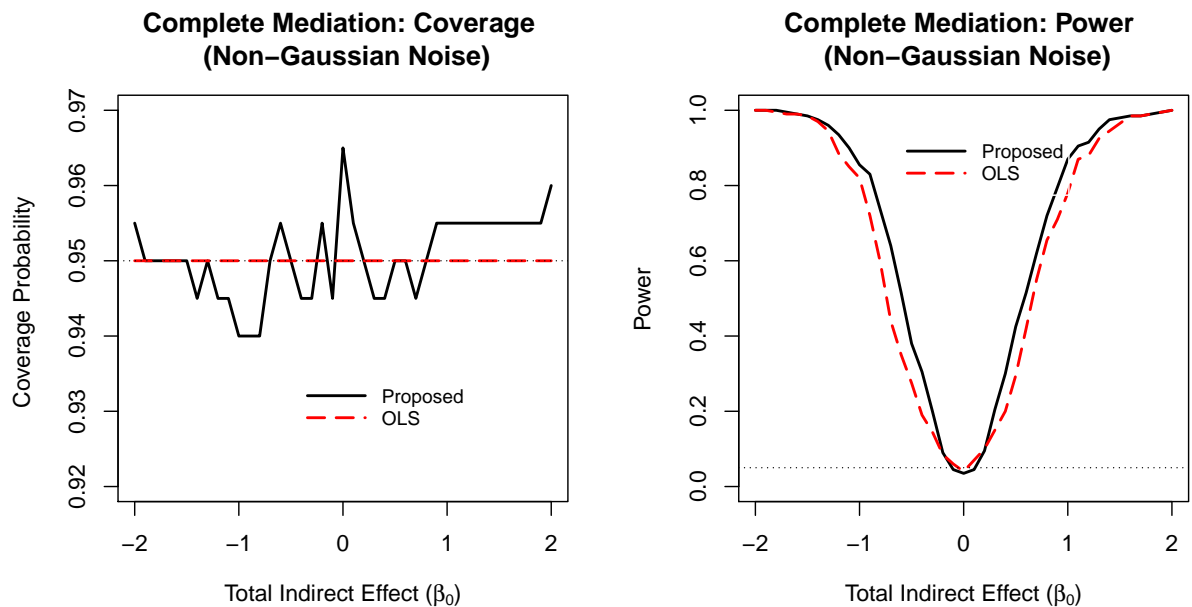}
\caption{\label{fig:non_normal}Average coverage probabilities of 95\% confidence intervals (left panels) and average power curves at significance level $\alpha = 0.05$ (right panels) for estimating and testing the indirect effect under complete mediation, over 200 replications, for $t_3$-distributed $\epsilon_{1i}$. The example have 500 potential mediators with 1 true mediator, with $\sigma_1$ generated from $t_3 \cdot (25/3)^{1/2}$. Proposed: $\tilde{b}$ from~\eqref{eq:complete}; OLS: ordinary least squares estimate.}
\end{figure}

Figure \ref{fig:non_normal} reports the coverage probabilities of the $1-\alpha$ confidence intervals and average power curves for our method and for ordinary least squares. With non-Gaussian noise, our proposed method was still able to maintain the nominal coverage probability and significance level, and had higher power than ordinary least squares for sufficiently large $\beta_0$, which is consistent with Proposition~\ref{p:var}.

\subsection{\label{sec:sparse_est}Robustness to matrix estimation accuracy}

Currently, our proposed method uses a constrained $\ell_1$-minimization approach to estimate the matrices $\Omega_{I}$ and $\Omega_{C}$, defined in Theorems \ref{thm:omegaxx} and \ref{thm:omegagg} in the main text. Under sparsity assumptions on various matrices, described in Assumptions \ref{a:sparse_matrices_xx} and \ref{a:sparse_matrices_gg}, these estimates are consistent. On the other hand, \citet{javanmard2014confidence} showed that in the standard high-dimensional linear regression setting, asymptotically optimal inference is possible without consistent estimation or sparsity of the precision matrix of the covariates. In this section we explore whether consistent estimation is indeed necessary for our method. We considered the same complete mediation simulation settings as in Example 1 of Section \ref{sec:complete_sims} in the main text.

\begin{figure}
\centering
  \includegraphics[width = \textwidth]{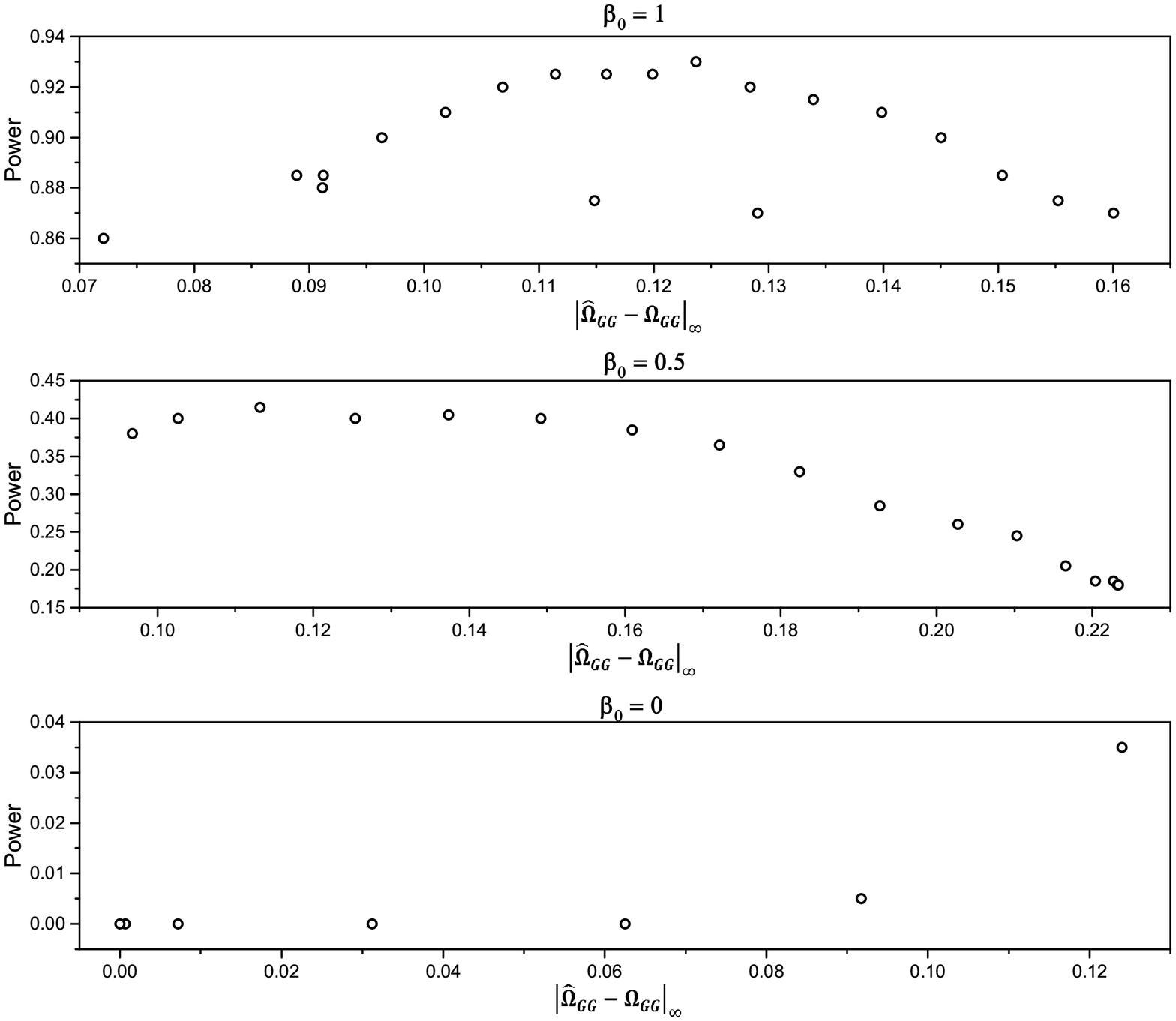}
\caption{\label{fig:omega_sparse}Average power of our proposed method versus average estimation accuracy of $\Omega_{C}$ with 500 potential mediators and 1 true mediator, at significance level $\alpha = 0.05$ over 200 replications.}
\end{figure}

We first varied the accuracy with which $\Omega_{C}$ was estimated by fixing the value of the tuning parameter $\tau_n$ of the estimation procedure \eqref{eq:M_C} to 20 different values between 0.05 and 0.7, instead of choosing an optimal value. We measured the resulting estimation error using $\Vert \hat\Omega_{C} - \Omega_{C} \Vert_\infty$, and we evaluated the power of our procedure when $\beta_0$ equaled 1, 0.5 or 0. Figure~\ref{fig:omega_sparse} shows that for $\beta_0 \ne 0$, both very accurate and very inaccurate estimation resulted in diminished power. An estimation error of $\Vert \hat\Omega_{C} - \Omega_{C} \Vert_\infty \approx 0.11$ seemed to give the highest power.

\begin{figure}
\centering
  \includegraphics[width = \textwidth]{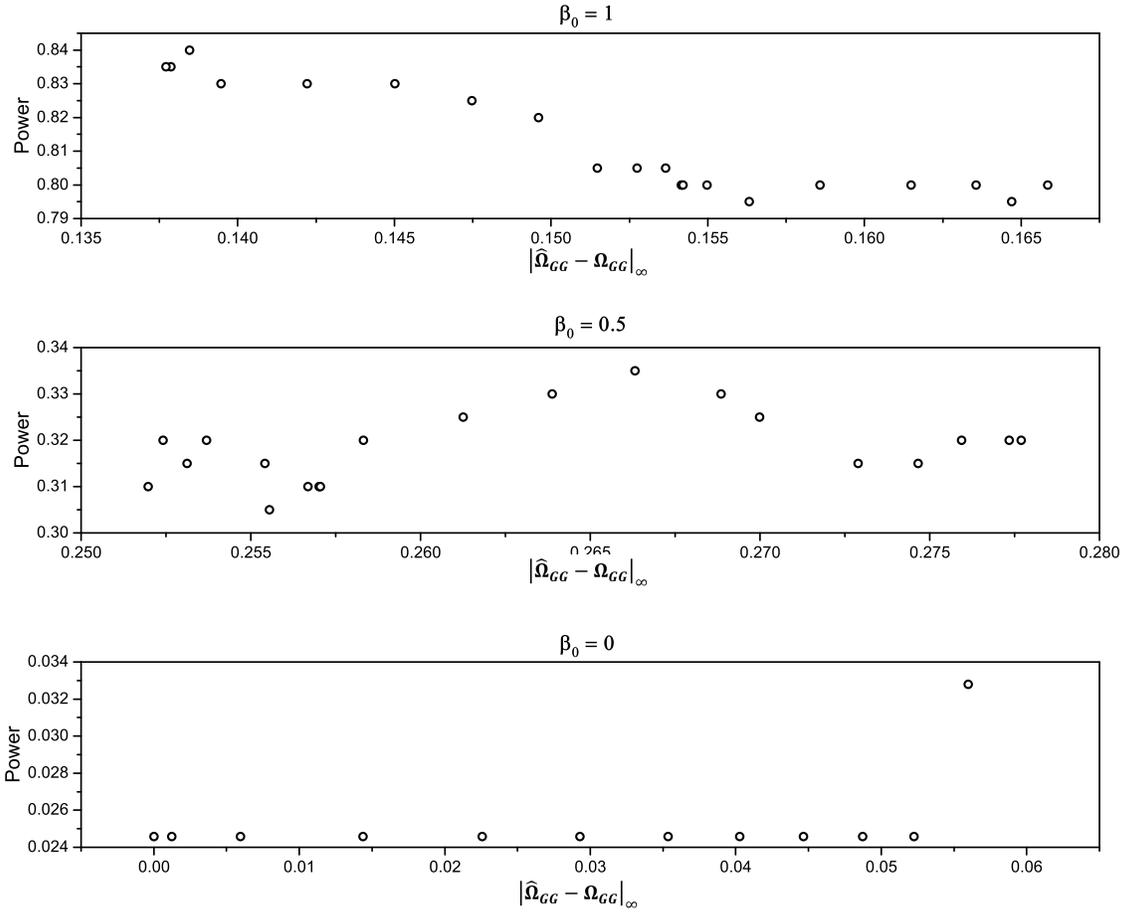}
\caption{\label{fig:omega_dense}Average power of our proposed method versus average estimation accuracy of $\Omega_{C}$ for dense $\Sigma_E^{-1}$, with 500 potential mediators and 1 true mediator, at significance level $\alpha = 0.05$ over 200 replications.}
\end{figure}

We next examined whether the sparsity conditions in Assumptions \ref{a:sparse_matrices_xx} and \ref{a:sparse_matrices_gg} were necessary. We used the same simulations as above, but changed $\Sigma_E$ such that $\Sigma_E^{-1}$ was no longer sparse. Specifically, $\Sigma_E^{-1}$ was a compound symmetric matrix
 with diagonal entries equal to 1 and off-diagonal entries equal to 0.8. Figure~\ref{fig:omega_dense} shows that our method could maintain type I error and have high power, for example when $\beta_0 = 1$, even when the sparsity assumptions were violated.

Together, these results suggest that our procedure is indeed robust to sparsity and consistent estimation of $\Omega_{I}$ and $\Omega_{C}$. On the other hand, these properties are necessary for our current approach to establishing the asymptotic distributions of our estimators. Our general strategy is to apply the central limit theorem to the terms on the right-hand sides of \eqref{eq:incomplete} and \eqref{eq:complete} in the main text, which involves $\hat \Omega_{I}$ and $\hat \Omega_{C}$. The behaviors of these random variables are difficult to characterize in general. We could consider conditioning on $G$ and $S$ so that $\hat \Omega_{I}$ and $\hat \Omega_{C}$ essentially become constants, but under our mediation model \eqref{eq:model} this fixes the values of the residual errors $\epsilon_{2i}$, which leads to difficulties. We cannot condition on $S$ alone, because doing so induces dependencies between the terms in \eqref{eq:incomplete}, as well as those in \eqref{eq:complete}, and makes applying the central limit theorem difficult. We therefore marginalize over $G$ and $S$, but since $\hat \Omega_{I}$ and $\hat \Omega_{C}$ are functions of $G$ and $S$, we need to establish that they are close to well-defined population-level quantities so that we can apply an unconditional central limit theorem. The simplest way to do this is to show that $\hat \Omega_{I}$ and $\hat \Omega_{C}$ converge to $\Omega_{I}$ and $\Omega_{C}$, which requires sparsity assumptions.

However, there may be other ways to show that $\hat \Omega_{I}$ and $\hat \Omega_{C}$ converge to well-defined population-level quantities. This would allow us to consider alternative estimators that are not based on constrained $\ell_1$-minimization, which may give improved finite-sample performance. For example, we can consider an extension of the approach of \citet{javanmard2014confidence} to our mediation analysis setting, as described in the main text in Section \ref{sec:incomplete}. In the future we hope to improve our proof strategies so that our theoretical results reflect the robustness demonstrated by our procedure in simulations.

\subsection{\label{sec:mis}Misspecified model}

If the mediation model \eqref{eq:model} is not correctly specified, there is no gaurantee that our proposed estimation and inference procedures will remain valid. We illustrate this by evaluaing our method, the naive method, and the method of \citet{zhang2016estimating} under a misspecified model simulation. We set the direct effect $\alpha_1 = 0$ so that we can also compare with ordinary least squares.

We generated data following simulations from Section \ref{sec:complete_sims} in the main text with 500 potential mediators and five true mediators. When we implemented the methods, we supposed that we only observed the first 350 potential mediators among 500 potential mediators, and that only four true mediators were contained in the observed sample. We estimated $\beta_0$ and tested $H_0: \beta_0 = 0$ at the $\alpha = 0.05$ significance level.

Figure~\ref{fig:mis} reports the coverage probabilities of the 95\% confidence intervals as well as the average power curves over 200 replications. Neither our method or the naive method were able to maintain nominal coverage, and all mediation-based methods had lower power compared to ordinary least squares. This makes sense because the indirect effect in our fitted model is no longer equal to the true indirect effect. Ordinary least squares was not affected because it does not require correctly specifying a mediation model.

\begin{figure}
\centering
  \includegraphics[width = \textwidth]{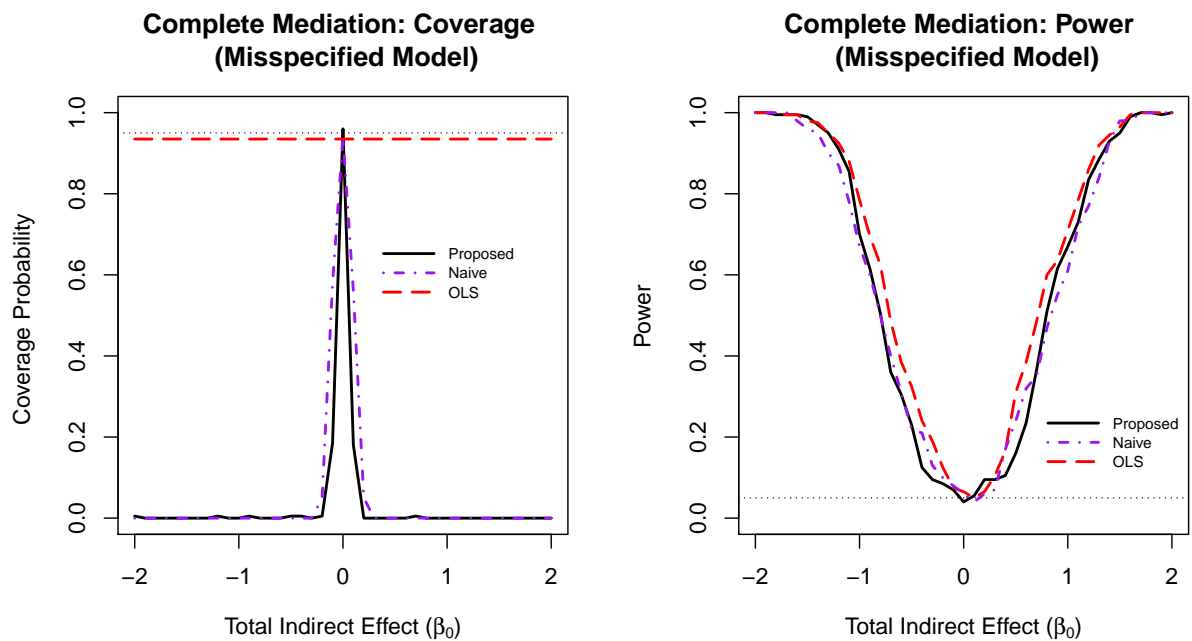}
\caption{\label{fig:mis}Average coverage probabilities of 95\% confidence intervals (left panels) and average power curves at significance level $\alpha = 0.05$ (right panels) for estimating and testing the indirect effect under a misspecified complete mediation model, at significance level $\alpha = 0.05$, over 200 replications. We observed 350 potential mediators with 4 true mediators, while the number of true mediators was 5 with 500 potential mediators; Proposed: $\tilde{b}$ from~\eqref{eq:complete}; Naive: the naive method discussed in Section \ref{sec:methods_compared}; OLS: ordinary least squares estimate; Zhang: the method of \citet{zhang2016estimating}.}
\end{figure}

\section{\label{sec:zhang}Discussion of \citet{zhang2016estimating} and additional simulations}

As the simulations in Section \ref{sec:numerical} of the main text revealed, the procedure of \citet{zhang2016estimating} exhibited counterintuitive behavior: the power was unexpectedly high when $\beta_0$ was small and decreased as $\beta_0$ increased. This section shows that this behavior is a consequence of the method's model selection step.

To help illustrate this, we introduced a slightly modified version of the method that applies the low-dimensional mediation analysis procedure of \citet{zhao2014more} to the mediators identified as being significant by \citet{zhang2010nearly}. This modified method, which we denote as $\text{Zhang}^*$ in the text, can provide interval estimates of the indirect effect, unlike the original procedure of \citet{zhang2016estimating}. If no mediator was selected to be significant, then we do not reject the null hypothesis and set the confidence interval at any significance level to equal the singleton $\{0\}$.

We also conducted additional simulations. Following the same data generating process in Section \ref{sec:complete_sims} in the main text, for samples $i = 1, \ldots, n = 300$ we generated $q = 1$ exposure $S_i \sim N(0,1)$ and $p = 500$ potential mediators $G_i$ following $G_i = c \gamma S_i + E_i$, we then generated the outcome according to $Y_i = G_i^\top \alpha_0 + S_i^\top \alpha_1 + \epsilon_{1i}$, where $\epsilon_{1i} \sim N(0,5)$. We let $\gamma$ have 15 non-zero components randomly generated between $[-1, 1]$, fixing $\gamma$ across replications, and let $\alpha_0$ have 15 non-zero components equal to $a = 1$. We chose one of these non-zero components to correspond to variables whose entries in $\gamma$ were also non-zero, which was the true mediator. Then we also let $\alpha_0$ have 15 non-zero components equal to $a = 0.8$, to see how signal stength in the model selection step affects the perfomance of the \citet{zhang2016estimating} procedure.

Figure \ref{fig:zhang_3} reports average model selection accuracies, coverage probabilities, and power curves of the \citet{zhang2016estimating} and $\text{Zhang}^*$ procedures over 200 replications. Since there was only one true mediator, model selection accuracy was measured as the proportion of replications where the selected model contained the true mediator when $\beta_0 \ne 0$, and as 1 when $\beta_0 = 0$. The $\text{Zhang}^*$ and \citet{zhang2016estimating} methods performed nearly identically, and the coverage plot shows that the $\text{Zhang}^*$ procedure had very poor coverage except when $\beta_0 = 0$. This suggests that the method of \citet{zhang2016estimating} had surprisingly high power for small $\beta_0$ was because it did not appropriately account for the variability of its model selection step. Indeed, the performances of the \citet{zhang2016estimating} procedures were in large part determined by the success of their model selection steps: when the components of $\alpha_0$ were changed from 1 to 0.8, the power of the \citet{zhang2016estimating} procedure for small $\beta_0$ dropped dramatically and no longer dominated our proposed procedure or ordinary least squares.

Figure \ref{fig:zhang_3} also showed that the model selection accuracy deterioriated as the magnitude of $\beta_0$ increased. This is because larger $\beta_0$ corresponded to larger $c$ and therefore to increased collinearity between $G_i$ and $S_i$ in the regression for $Y_i$, which makes consistent model selection difficult. This drop in model selection accuracy likely explains the drop in power of the \citet{zhang2016estimating} procedure for large $\beta_0$.

\begin{figure}
\centering
  \includegraphics[width = 12cm, height=6cm]{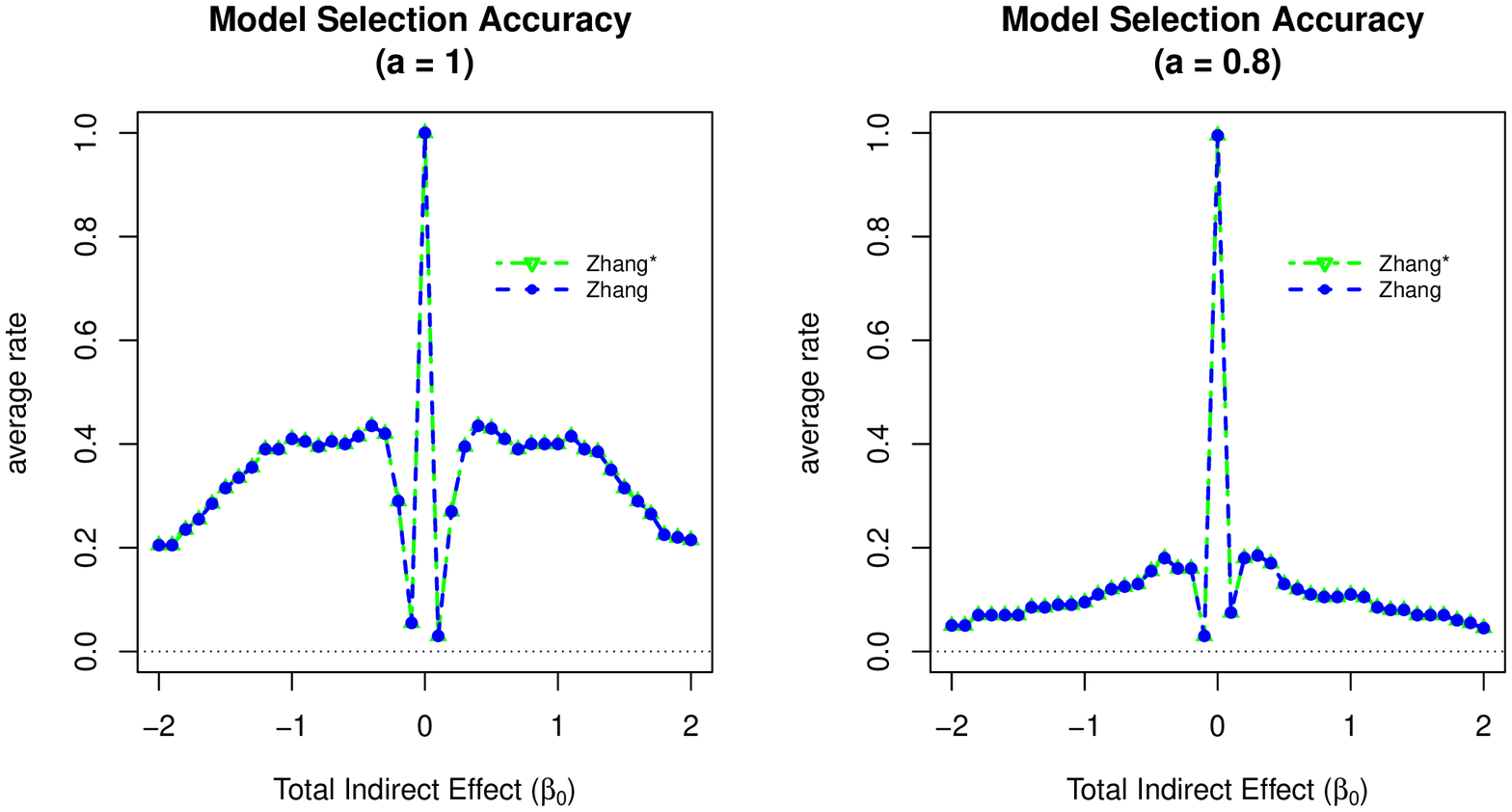}\par
  \includegraphics[width = 12cm, height=6cm]{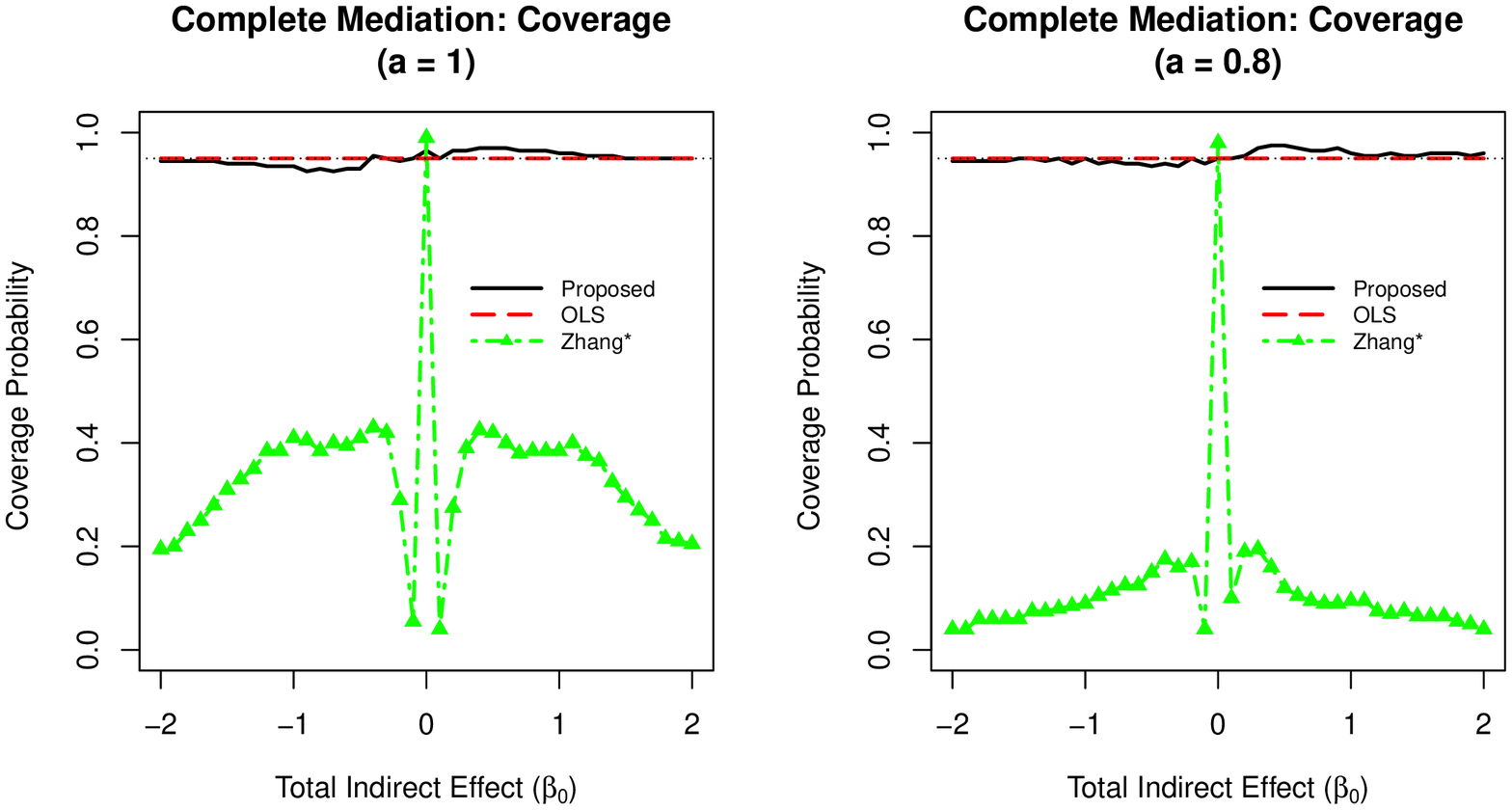}\par
  \includegraphics[width = 12cm, height=6cm]{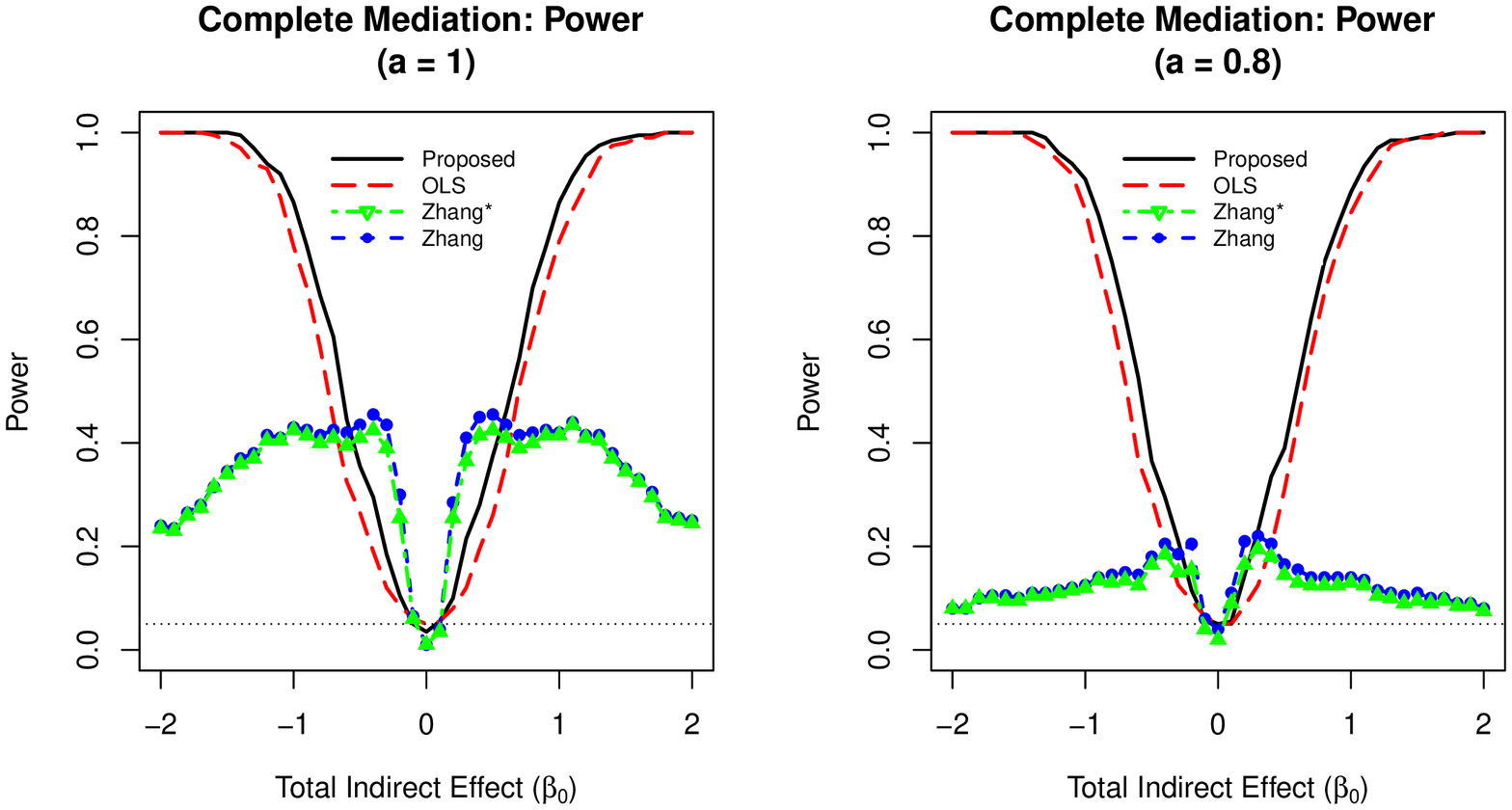}\par
\caption{\label{fig:zhang_3} Average coverage probabilities of 95\% confidence intervals (upper panels), average power curves at significance level $\alpha = 0.05$ (middle panels) and average model selection accuracy (of the \citet{zhang2016estimating} method, lower panels), for estimating and testing the indirect effect under complete mediation, over 200 replications. The number of true mediators was 1 with 500 potential mediators, and the outcome-mediator signal strength constant $a$ was 1 in the left panels and 0.8 in the right panels; Proposed: $\tilde{b}$ from~\eqref{eq:complete}; OLS: ordinary least squares estimate; Zhang: the method of \citet{zhang2016estimating}; $\text{Zhang}^\ast$, the modified method of \citet{zhang2016estimating} discussed in Section \ref{sec:zhang} in Supplementary Materials.}
\end{figure}

As we mentioned in Section \ref{sec:incomplete_sims} in the main text, improperly accounting for the variability of the model selection step can also inflate the type I error of the method of \citet{zhang2016estimating}. Following the data generating process in Section \ref{sec:complete_sims} in the main text, for samples $i = 1, \ldots, n = 300$ we generated $q = 1$ exposure $S_i \sim N(0,1)$ and $p = 500$ potential mediators $G_i$ following $G_i = \gamma S_i + E_i$. We then generated the outcome according to $Y_i = G_i^\top \alpha_0 + S_i^\top \alpha_1 + \epsilon_{1i}$, where $\epsilon_{1i} \sim N(0,5)$. We let $\alpha_0$ have one non-zero component equal to $a = 1$, and set the corresponding entry of $\gamma$ to be $c$, which was the true mediator. The indirect effect $\beta_0$ then equals $c$, which we varied in our simulations. We then set the remaining 499 components of $\gamma$ equal to 1, which causes the components of $G_i$ to be highly correlated.

Figure~\ref{fig:zhang_type1} reports average power curves over 200 replications. The method of \citet{zhang2016estimating} had superior power, but it also failed to maintain type I error when $\beta_0=0$. This is because its model selection performance was degraded due to the multicollinearity between the components of $G_i$, and it erroneously selected components of $G_i$ that were not true mediators. Because it does not take into account the fact that the selected model is random, the method of \citet{zhang2016estimating} incorrectly identifies these other components as mediators.

\begin{figure}
 \centering
  \includegraphics[width = 10cm]{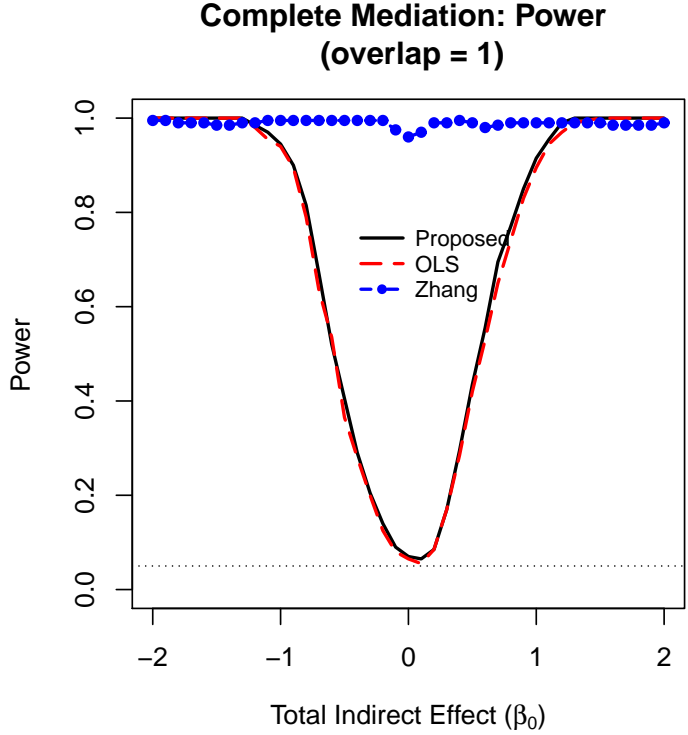}
\caption{\label{fig:zhang_type1} Average power curves at significance level $\alpha = 0.05$ for estimating and testing the indirect effect under complete mediation, over 200 replications. The number of true mediators was 1 with 500 potential mediators; Proposed: $\tilde{b}$ from~\eqref{eq:complete}; OLS: ordinary least squares estimate; Zhang: the method of \citet{zhang2016estimating};}
\end{figure}

\section{\label{sec:proofs}Proofs of theorems}

\begin{proof}[of Theorem~\ref{thm:omegaxx} and ~\ref{thm:omegagg}]
We first prove the result for the $\hat\Omega_{C}$ that solves~\eqref{eq:M_C}. 
First we show that $\Vert\Omega_{C}\hat\Sigma_{GG}-\hat\Sigma_{SG}\Vert_{\infty}\leq(N_G+1)\{(\log p) / n\}^{1/2}$ with probability going to 1. First we have
\begin{align*}
\Vert\Omega_{C} \hat\Sigma_{GG}-\hat\Sigma_{SG}\Vert_{\infty}\leq\Vert\Omega_{C}(\hat\Sigma_{GG}-\Sigma_{GG})\Vert_\infty+\Vert\Sigma_{SG}-\hat\Sigma_{SG}\Vert_\infty.
\end{align*}
Following the proofs of Theorems~1 and~4 from \citet{cai2011constrained}, we have 
\[
\Vert\hat\Sigma_{XX}-\Sigma_{XX}\Vert_{\infty}\leq C_1 [\{\log (p+q)\} / n]^{1/2}
\]with probability greater than $p_1 = 1-4p^{-\tau}$ under Assumption~\ref{a:tails}, where $C_1=2\eta^{-2}(2+\tau+\eta^{-1}e^2K^2)^2$ and $\tau>0$. Since $\Vert \hat\Sigma_{GG}-\Sigma_{GG}\Vert_\infty$ and $\Vert \hat\Sigma_{SG}-\Sigma_{SG}\Vert_\infty$ are both less than or equal to $\Vert\hat\Sigma_{XX}-\Sigma_{XX}\Vert_{\infty}$, and by Assumption~\ref{a:sparse_matrices_gg} $\Vert \Omega_{C}^\top \Vert_{L_1} \leq N_G$, we have
\begin{align*}
  \Vert\Omega_{C} \hat\Sigma_{GG} - \hat\Sigma_{SG}\Vert_{\infty}
  \leq\,&
  \Vert\Omega_{C} ^\top \Vert_{L_1} \Vert \hat\Sigma_{GG}-\Sigma_{GG} \Vert_\infty
  +
  \Vert\Sigma_{SG}-\hat\Sigma_{SG}\Vert_\infty \\
  \leq\,&
  (N_G + 1) C_1 [\{\log (p+q)\} / n]^{1/2}
\end{align*}
with probability greather than $p_1$. Next, for $1\leq i \leq q$, let ${\hat b}_i$ be the solution of the optimization problem:
\begin{align*}
\min \Vert b \Vert_1\mbox{ subject to } \Vert{b}^\top\hat\Sigma_{GG}-\hat\Sigma_{SG,i}\Vert_\infty\leq\tau_n,
\end{align*}
where ${b}$ is a vector in ${{R}}^{p}$ and $\Sigma_{SG,i}$ denote the i-th row of $\Sigma_{SG}$.
Let $\hat{B} = (b_1^\top, \ldots, b_q^\top)$ and write $\Omega = (\omega_1^\top, \ldots, \omega_q^\top)$. Then the constraint $\Vert\Omega\hat\Sigma_{GG}-\hat\Sigma_{SG}\Vert_\infty\leq\tau_n$ from~\eqref{eq:M_C} is equivalent to
$\Vert{\omega}_i^\top\hat\Sigma_{GG}-\hat\Sigma_{SG,i}\Vert_\infty\leq\tau_n\quad\textrm{for all}\quad 1\leq i \leq q$.
Thus we have
$\Vert{\omega}_i\Vert_1\geq\Vert{\beta}_i\Vert_1$ for all $1\leq i \leq q$.
Since $\Vert{B}^\top\hat\Sigma_{GG}-\hat\Sigma_{SG}\Vert_\infty\leq\tau_n$, we have $\Vert \hat\Omega \Vert_1\leq \Vert\hat{{B}} \Vert_1$,
thus $\Omega_1 = \hat{B}$.
By the definition of ${b}_i^\top$, we can see that $\Vert{b}_i^\top\Vert_1\leq\Vert\Omega_{C}\Vert_{L_1}$ for $1\leq i\leq q$. Then we have $\Vert\hat\Omega_{C}^\top\Vert_{L_1}\leq\Vert\Omega_{C}^\top\Vert_{L_1}$.
Finally we have $\Vert(\hat{\Omega}_{C}-\Omega_{C})\Sigma_{GG}\Vert_\infty$
\begin{align*}
  =\,&
  \Vert\hat{\Omega}_{C}(\Sigma_{GG}-\hat{\Sigma}_{GG})+\hat{\Omega}_{C}\hat{\Sigma}_{GG}-\hat{\Sigma}_{SG}+\hat{\Sigma}_{SG}-\Sigma_{SG}\Vert_\infty\\
  =\,&
  \Vert(\hat{\Omega}_{C}-\Omega_{C})(\Sigma_{GG}-\hat{\Sigma}_{GG})+\Omega_{C}(\Sigma_{GG}-\hat{\Sigma}_{GG})+
  \hat{\Sigma}_{SG}-\Sigma_{SG}\Vert_\infty+\Vert\hat{\Omega}_{C}\hat{\Sigma}_{GG}-\hat{\Sigma}_{SG}\Vert_\infty
  \\
  =\,&
  \Vert(\hat{\Omega}_{C}-\Omega_{C})(\Sigma_{GG}-\hat{\Sigma}_{GG})\Vert_\infty+\Vert\Omega_{C}\hat\Sigma_{GG}-\hat{\Sigma}_{SG}\Vert_\infty+\Vert\hat{\Omega}_{C}\hat{\Sigma}_{GG}-\hat{\Sigma}_{SG}\Vert_\infty\\
  =\,&
  \Vert(\hat{\Omega}_{C}-\Omega_{C})(\Sigma_{GG}-\hat{\Sigma}_{GG})\Vert_\infty+2\tau_n
\end{align*}
and
\[
\Vert(\hat{\Omega}_{C}-\Omega_{C})(\Sigma_{GG}-\hat{\Sigma}_{GG})\Vert_\infty
\leq
2\Vert\Omega_{C}^\top\Vert_{L_1}\Vert\Sigma_{GG}-\hat{\Sigma}_{GG}\Vert_\infty
\leq
2N_GC_1 \{ \log{(p+q)}/n \} ^{1/2}
\]
under the same event with probability greater than $p_1$.
This gives us 
$$\Vert(\hat{\Omega}_{C}-\Omega_{C})\Sigma_{GG}\Vert_\infty\leq2N_GC_1(\log{p}/n)^{1/2}+2\tau_n=(4N_G+2) \{ \log{(p+q)}/n \} ^{1/2}$$
Since we assume $\Vert\Sigma_{GG}^{-1}\Vert_{L_1}\leq{M_G}$, we have
$ \Vert\hat{\Omega}_{C}-\Omega_{C}\Vert_\infty\leq(4N_G+2)M_G \{ \log{(p+q)}/n \} ^{1/2}$ 
with probability greater than $p_1.$

The proof for $\hat\Omega_{I}$ can then be obtained by complete analogy.
\end{proof}

\begin{proof}[of Theorem~\ref{thm:complete}]
Based on Section~\ref{sec:complete} and Theorem~\ref{thm:complete}, we have
$$
  n^{1/2} ( \tilde{b} - \beta_0 )
= W + \Delta_0 + \Delta_1 + \Delta_2,
$$
where $W
=
n^{-1/2}\Sigma_{SS}^{-1}S^\top\epsilon_{2}+n^{-1/2}\Sigma_{SS}^{-1}\Omega_C{G}^\top\epsilon_{1}
=
n^{-1/2}\Sigma_{SS}^{-1} \sum_{i=1}^{n}(S_i^\top\epsilon_{2i}+\Omega_C G_i^\top\epsilon_{1i}),$ and 
\begin{eqnarray*}
\Delta_0 &=& n^{1/2} \Sigma_{SS}^{-1}(\hat\Omega_{C}\hat\Sigma_{GG} - \hat \Sigma_{SG})(\hat\alpha-\alpha_0), \\
\Delta_1 &=& n^{-1/2} ({\hat\Sigma_{SS}^{-1}}\hat\Omega_{C}-{\Sigma_{SS}^{-1}}\Omega_{C})G^\top\epsilon_1, \\
\Delta_2 &=& n^{-1/2} (\hat\Sigma_{SS}^{-1}-\Sigma_{SS}^{-1})S^\top\epsilon_2.
\end{eqnarray*}
First we prove that $W\rightarrow N(0,\sigma_1^2\Sigma_{SS}^{-1}\Sigma_{SG}\Sigma_{GG}^{-1}\Sigma_{SG}\Sigma_{SS}^{-1}+\sigma_2^2\Sigma_{SS}^{-1})$:
For $i = 1, \ldots, n$, let 
\[
W_i=S_i^\top\epsilon_{2i}+\Omega_C G_i^\top\epsilon_{1i}.
\]
Then we have 
\[
\text{E}(W_i)=\text{E} \{ \text{E}(S_i^\top\epsilon_{2i}\mid S_i) \} +\text{E} \{ \text{E}(\Omega_C G_i^\top\epsilon_{1i}\mid G_i) \} =0,
\]
and
\[
\text{var}(S_i^\top\epsilon_{2i})=\sigma_2^2\text{E}(S_i^\top S_i),
\]
\[
\text{var}(\Omega_C G_i^\top\epsilon_{1i})=\sigma_1^2\Omega_C\text{E}(G_i^\top G_i)\Omega_C^\top
\]
\[
\text{cov}(S_i^\top\epsilon_{2i},\Omega_C G_i^\top\epsilon_{1i})=\text{E} \{ \text{E}(S_i^\top\epsilon_{2i}\epsilon_{1i}^\top G_i\Omega_C ^\top\mid S_i,G_i) \} =0,
\] which means that 
\[
\text{var}(W_i)=\sigma_2^2\Sigma_{SS}+\sigma_1^2\Omega_C\Sigma_{GG}\Omega_C^\top.\]

Since $W = n^{-1/2}\Sigma_{SS}^{-1}\sum_{i=1}^{n}W_i$, by CLT, when $n\rightarrow\infty$,
$$W\rightarrow{N}_q(0,\sigma_1^2\Sigma_{SS}^{-1}\Sigma_{SG}\Sigma_{GG}^{-1}\Sigma_{SG}\Sigma_{SS}^{-1}+\sigma_2^2\Sigma_{SS}^{-1}).$$
Next we prove that $\Vert\Delta_0\Vert_\infty=o_P(1)$, $\Vert \Delta_1 \Vert_\infty=o_P(1)$ and $\Vert \Delta_2\Vert_\infty=o_P(1)$.

We start by showing 
\[
\Vert \Delta_1\Vert_\infty = \Vert n^{-1/2}(\hat\Sigma_{SS}^{-1}\hat\Omega_C-\Sigma_{SS}^{-1}\Omega_C){G}^\top\epsilon_{1}\Vert_\infty=o_P(1).
\]
Since $\Vert n^{-1/2}(\hat\Sigma_{SS}^{-1}\hat\Omega_C-\Sigma_{SS}^{-1}\Omega_C){G}^\top\epsilon_{1}\Vert_{\infty}$
\begin{eqnarray*}
&\leq&
\Vert n^{-1/2}(\hat\Sigma_{SS}^{-1}\hat\Omega_C-\hat\Sigma_{SS}^{-1}\Omega_C){G}^\top\epsilon_{1}\Vert_{\infty}
+\Vert n^{-1/2}(\hat\Sigma_{SS}^{-1}\Omega_C-\Sigma_{SS}^{-1}\Omega_C){G}^\top\epsilon_{1}\Vert_{\infty}\\
&\leq&
\Vert\hat\Sigma_{SS}^{-1}(\hat\Omega^\top_C-\Omega^\top_C)\Vert_{L_1}\Vert n^{-1/2}{G}^\top\epsilon_{1}\Vert_{\infty}
+\Vert\Omega^\top_C(\hat\Sigma_{SS}^{-1}-\Sigma_{SS}^{-1})\Vert_{L_1}\Vert n^{-1/2}{G}^\top\epsilon_{1}\Vert_{\infty}.
\end{eqnarray*}
Let $M_s = \Vert\hat\Sigma_{SS}^{-1}\Vert_{\infty}$, then 
\[
\Vert\hat\Sigma_{SS}^{-1}(\hat\Omega_C^\top-\Omega^\top_C)\Vert_{L_1}
\leq
q\Vert\hat\Sigma_{SS}^{-1}\Vert_{\infty}\Vert\hat\Omega^\top_C-\Omega^\top_C\Vert_{L_1}
\leq
qM_sC_3s_0(\log{p}/n)^{1/2}
\]
and 
\[
\Vert\Omega^\top_C(\hat\Sigma_{SS}^{-1}-\Sigma_{SS}^{-1})\Vert_{L_1}
\leq
q\Vert\hat\Sigma_{SS}^{-1}-\Sigma_{SS}^{-1}\Vert_{\infty}\Vert\Omega^\top_C\Vert_{L_1}
\leq
cqM(\log{q}/n)^{1/2}.
\]

We have $\Vert{G}^\top\epsilon_{1}\Vert_{\infty}=\max_{j\in\{1,\ldots,p\}}\{\sum_{i=1}^{n}G_{ij}\epsilon_{1i}\}$,
and $G_{ij}\epsilon_{1i}, i=1,\ldots,n$ are i.i.d random variables given each j.

Following the proof in \citet{cai2011constrained} for Theorem 1(a), let $t=\eta \sqrt{\log{p}/n}$ and $C_{K_1} = 1+\tau+K^2$, using the inequality $\vert e^{s} - 1 - s\vert \leq s^2 e^{\vert s \vert}$ and $s^2 \leq e^{\vert s \vert}$ for any $s \in R$, by basic calculations, we can get 
\begin{align*}
& P \{ \sum_{i=1}^{n}G_{ij}\epsilon_{1i}\geq \eta^{-1} C_{K_1} (n\log{p})^{1/2} \} &\\
&\leq e^{-C_{K_1}\log{p}} (\text{E}\exp{(t G_{ij}\epsilon_{1i})})^n &\\
&\leq \exp (-C_{K_1}\log{p} + n \log{\text{E}\exp{(t G_{ij}\epsilon_{1i})}} )&\\
&\leq \exp (-C_{K_1}\log{p} + n ( \text{E}e^{\vert t G_{ij}\epsilon_{1i} \vert} -1 - \text{E} \vert t G_{ij}\epsilon_{1i} \vert )) 
\quad (\text{since } \text{E}(G_{ij}\epsilon_{1i})=0)& \\
&\leq \exp ( -C_{K_1}\log{p} + nt^2 \text{E} (G_{ij}^2\epsilon_{1i}^2 e^{t \vert G_{ij}\epsilon_{1i}\vert}))& \\
&\leq \exp ( -C_{K_1}\log{p} + nt^2 \eta^{-2} \text{E} ( \eta^2 G_{ij}^2\epsilon_{1i}^2 e^{t \vert G_{ij}\epsilon_{1i}\vert}))& \\
&\leq \exp ( -C_{K_1}\log{p} + \log{p} \text{E} (e^{ \eta \vert G_{ij}\epsilon_{1i} \vert } e^{t \vert G_{ij}\epsilon_{1i}\vert}))& \\
&\leq \exp ( -C_{K_1}\log{p} + \text{E} e^{ \frac{\eta + t}{2} (G_{ij}^2 + \epsilon_{1i}^2) } \log{p})& \\
& = \exp (-C_{K_1}\log{p} + \text{E} e^{ \frac{\eta + t}{2} G_{ij}^2} \text{E} e^{ \frac{\eta + t}{2} \epsilon_{i}^2}\log{p}) 
\quad (\text{since } G_{ij}\perp\epsilon_{1i}, \forall i)& \\
&\leq \exp (-C_{K_1}\log{p} + K^2 \log{p})& \\
&\leq \exp (-(1 + \tau) \log{p})&
\end{align*}

Thus, we have
\begin{eqnarray*}
P \{ \Vert{G}^\top\epsilon_{1}\Vert_{\infty}\geq \eta^{-1} C_{K_1}(n\log{p})^{1/2} \}
&=&
P \{ \max_{j\in\{1,...,p\}} ( \sum_{i=1}^{n}G_{ij}\epsilon_{1i} ) \geq \eta^{-1} C_{K_1}(n\log{p})^{1/2} \} \\
&\leq&
2pP \{ \sum_{i=1}^{n}G_{ij}\epsilon_{1i}\geq \eta^{-1} C_{K_1}(n\log{p})^{1/2} \} \\
& \leq &
2p^{-\tau} \rightarrow 0
\end{eqnarray*}
when $p>n\rightarrow \infty$.
This gives us $P \{ \Vert n^{-1/2}{G}^\top\epsilon_{1}\Vert_{\infty}\geq \eta^{-1} C_{K_1}(\log{p})^{1/2} \} \rightarrow0$,
then we have $P \{ \Vert n^{-1/2}(\hat\Sigma_{SS}^{-1}\hat\Omega_C-\Sigma_{SS}^{-1}\Omega_C){G}^\top\epsilon_{1}\Vert_\infty\geq c_2(\log{p})/n^{1/2} \} \rightarrow 0$.
And since $(\log{p})/n^{1/2}\rightarrow0$,
\begin{align*}
P \{ \Vert n^{-1/2}(\hat\Sigma_{SS}^{-1}\hat\Omega_C-\Sigma_{SS}^{-1}\Omega_C){G}^\top\epsilon_{1}\Vert_\infty= 0 \} \rightarrow 1.
\end{align*}
Next we show that 
\[
\Vert \Delta_2\Vert_\infty =\Vert n^{-1/2} (\hat\Sigma_{SS}^{-1}-\Sigma_{SS}^{-1})S^\top\epsilon_2 \Vert_\infty=o_P(1).
\]
Since $ \Vert n^{-1}{S}^\top\epsilon_{2}\Vert_{L_1} = \sum_{1\leq j \leq q} \vert n^{-1} \sum_{i=1}^{n} S_{ij} \epsilon_{2i} \vert$, and by the law of large numbers, $\vert n^{-1} \sum_{i=1}^{n} S_{ij} \epsilon_{2i} \vert \rightarrow \text{E} (S_{ij} \epsilon_{2i}) = 0$ for all j, we have $ \Vert n^{-1}{S}^\top\epsilon_{2}\Vert_{L_1} \rightarrow 0$.
Combining $\Vert\hat\Sigma_{SS}^{-1}-\Sigma_{SS}^{-1}\Vert_{\infty}\leq c /{n^{1/2}}$, we have when $n\rightarrow\infty$,
$$\Vert n^{-1/2}(\hat\Sigma_{SS}^{-1}-\Sigma_{SS}^{-1})S^\top\epsilon_{2}\Vert_{\infty} 
\leq 
n^{1/2} \Vert \hat\Sigma_{SS}^{-1}-\Sigma_{SS}^{-1} \Vert_\infty 
\Vert n^{-1} S^\top\epsilon_{2}\Vert_{L_1} \rightarrow 0. $$
Finally we show that 
\[
\Vert \Delta_0\Vert_\infty =\Vert n^{1/2}\Sigma_{SS}^{-1}(\hat\Sigma_{SG}-\hat\Omega_C\hat\Sigma_{GG})(\hat\alpha-\alpha_0)\Vert_\infty=o_P(1).
\]
Since
\begin{eqnarray*}
\Vert n^{1/2}\Sigma_{SS}^{-1}(\hat\Sigma_{SG}-\hat\Omega_C\hat\Sigma_{GG})(\hat\alpha-\alpha_0)\Vert_{\infty}
&\leq&
\Vert n^{1/2}\Sigma_{SS}^{-1}(\hat\Sigma_{SG}-\hat\Omega_C\hat\Sigma_{GG})\Vert_{\infty}\Vert \hat\alpha-\alpha_0 \Vert_1\\
&\leq&
n^{1/2}\Vert\Sigma_{SS}^{-1}\Vert_{L_1}\Vert\hat\Sigma_{SG}-\hat\Omega_C\hat\Sigma_{GG}\Vert_{\infty}\Vert \hat\alpha-\alpha_0 \Vert_1
\end{eqnarray*}
From properties of scaled lasso we have $\Vert \hat\alpha-\alpha_0 \Vert_1=O \{ (\log{p}/n)^{1/2} \}$,
and by definition $\hat\Omega_C$ solves \eqref{eq:M_C} in the main text, we have 
$\Vert\hat\Sigma_{SG}-\hat\Omega_C\hat\Sigma_{GG}\Vert_{\infty}\leq\tau_n=O \{ (\log{p}/n)^{1/2} \}$.
Then we have $\Vert\Delta_0\Vert_{\infty}=O ( \log{p}/n^{1/2} )=o_P(1)$.
\end{proof}

\begin{proof}[of Theorem~\ref{thm:incomplete}]
Since
\[
n^{1/2}
\begin{pmatrix}
\hat{b} - \beta_0\\
\hat{a} - \alpha_1
\end{pmatrix}
= W + \Delta_0 + \Delta_1 + 
\begin{pmatrix}
\Delta_2 \\
0
\end{pmatrix},
\]
where
\[
W = n^{-1/2} I_2 \otimes{\Sigma_{SS}^{-1}}\Omega_{I}X^\top\epsilon_1 +
\begin{pmatrix}
n^{-1/2} \Sigma_{SS}^{-1}S^\top\epsilon_2 \\
0
\end{pmatrix}
,\] and 
\begin{eqnarray*}
\Delta_0 &=& n^{1/2} I_2 \otimes \hat \Sigma_{SS}^{-1}(\hat D-\hat\Omega_{I}\hat\Sigma_{XX})(\hat\alpha-\alpha), \\
\Delta_1 &=& n^{-1/2} I_2 \otimes({\hat\Sigma_{SS}^{-1}}\hat\Omega_{I}-{\Sigma_{SS}^{-1}}\Omega_{I})X^\top\epsilon_1, \\
\Delta_2 &=& n^{-1/2} (\hat\Sigma_{SS}^{-1}-\Sigma_{SS}^{-1})S^\top\epsilon_2.
\end{eqnarray*}

Showing that $\Vert\Delta_2\Vert_\infty=o_P(1)$ is exactly the same as in the proof of Theorem \ref{thm:complete}. By fitting a new $\ell_1$ optimization problem as in (3) instead of (4), 
$\Vert\Delta_0\Vert_\infty=o_P(1)$ can be directly extended from 
$\Vert n^{1/2}\Sigma_{SS}^{-1}(\hat\Sigma_{SG}-\hat\Omega_C\hat\Sigma_{GG})(\hat\alpha-\alpha_0)\Vert_\infty=o_P(1)$ 
and $\Vert\Delta_1\Vert_\infty=o_P(1)$ can be extended from $\Vert n^{-1/2}(\hat\Sigma_{SS}^{-1}\hat\Omega_C-\Sigma_{SS}^{-1}\Omega_C){G}^\top\epsilon_{1}\Vert_\infty=o_P(1)$, which has already been proved in Theorem \ref{thm:complete}.
So it remains to prove the asymptotic normality of W:
\[
W
=
n^{-1/2}I_2\otimes\Sigma_{SS}^{-1}\sum_{i=1}^{n} \{ \Omega_{I}X_i^\top\epsilon_{1i}+
\begin{pmatrix}
S_i^\top\epsilon_{2i} \\
0
\end{pmatrix} \}.
\]
For each $i$
let \[
W_i=W_{1i}+W_{2i}=
\Omega_{I}X_i^\top\epsilon_{1i}
+
\begin{pmatrix}
S_i^\top\epsilon_{2i} \\
0
\end{pmatrix},
\]
as before we have $\text{E}(W_{1i})=\text{E}(W_{2i})=0$, and 
\[
\text{var}(W_{1i}) = \text{E}(W_{1i}W_{1i}^\top)=\sigma_1^2\Omega_{I}\text{E}(X_i^\top X_i)\Omega_{I}^\top
=\sigma_1^2\Omega_{I}\Sigma_{XX}\Omega_{I}^\top
=\sigma_1^2D\Sigma_{XX}^{-1}D^\top,
\]
where $D=
\begin{pmatrix}
\Sigma_{SG} & 0\\
0 & \Sigma_{SS}
\end{pmatrix}
$.
By inversion of block matrix we have 
\begin{align*}
\Sigma_{XX}^{-1}=
\begin{pmatrix}
\Sigma_{GG} &\Sigma_{GS}\\
\Sigma_{SG} & \Sigma_{SS}
\end{pmatrix}^{-1}
=
\begin{pmatrix}
J^{-1} & -J^{-1}\Sigma_{GS}\Sigma_{SS}^{-1}\\
-\Sigma_{SS}^{-1}\Sigma_{SG}J^{-1} & \Sigma_{SS}^{-1} + \Sigma_{SS}^{-1}\Sigma_{SG}J^{-1}\Sigma_{GS}\Sigma_{SS}^{-1}
\end{pmatrix},
\end{align*}
where $J=\Sigma_{GG}-\Sigma_{GS}\Sigma_{SS}^{-1}\Sigma_{SG}$. Thus 
\[
\text{var}(W_{1i}) = \sigma_1^2
\begin{pmatrix}
\Sigma_{SG}J^{-1}\Sigma_{GS} & -\Sigma_{SG}J^{-1}\Sigma_{GS}\\
-\Sigma_{SG}J^{-1}\Sigma_{GS} & \Sigma_{SS}+\Sigma_{SG}J^{-1}\Sigma_{GS}
\end{pmatrix}
\]
\[
\text{var}(W_{2i}) = \text{E}(W_{2i}W_{2i}^\top)=\sigma_2^2
\begin{pmatrix}
\Sigma_{SS} & 0\\
0 & 0
\end{pmatrix};
\]
and 
\begin{eqnarray*}
\quad\text{cov}(W_{1i},W_{2i})
&=&
\text{E} \{ \Omega_{I}X_i^\top\epsilon_{1i}
(
\epsilon_{2i}S_i \quad 0
)
\} \\
&=&
\text{E} \{ \text{E} (
\Omega_{I}X_i^\top\epsilon_{1i}\epsilon_{2i}S_i \quad 0
)
\mid S_i,X_i \}
=
0.
\end{eqnarray*}
So
\begin{eqnarray*}
\text{var}(W_{i})
&=&
\text{var}(W_{1i})+\text{var}(W_{2i}) \\
&=&
\sigma_1^2
\begin{pmatrix}
\Sigma_{SG}J^{-1}\Sigma_{GS} & -\Sigma_{SG}J^{-1}\Sigma_{GS}\\
-\Sigma_{SG}J^{-1}\Sigma_{GS} & \Sigma_{SS}+\Sigma_{SG}J^{-1}\Sigma_{GS}
\end{pmatrix}
+\sigma_2^2
\begin{pmatrix}
\Sigma_{SS}^{-1} & 0\\
0 & 0
\end{pmatrix}.
\end{eqnarray*}

Since $W=n^{-1/2}I_2\otimes\Sigma_{SS}^{-1}\sum_{i=1}^{n}W_i$, by CLT, when $n\rightarrow\infty$, we have 
$$W\rightarrow{N}_{2q} \{ 0,
\begin{pmatrix}
\sigma_1^2\Gamma & -\sigma_1^2\Gamma\\
-\sigma_1^2\Gamma & \sigma_1^2(\Gamma+\Sigma_{SS}^{-1})
\end{pmatrix} 
+\sigma_2^2
\begin{pmatrix}
\Sigma_{SS}^{-1} & 0\\
0 & 0
\end{pmatrix}
\},$$
where $\Gamma=\Sigma_{SS}^{-1}\Sigma_{SG}J^{-1}\Sigma_{SG}\Sigma_{SS}^{-1}$.
\end{proof}

\begin{proof}[of Proposition~\ref{p:var} and ~\ref{p:var'}]
When $n \rightarrow\infty$,
$\text{var}(\hat\beta_s) \rightarrow (\sigma_1^2+\sigma_2^2)\Sigma_{SS}^{-1}$ and $\text{var}(\hat\beta)\rightarrow\sigma_1^2\Sigma_{SS}^{-1}U\Sigma_{SS}^{-1}+\sigma_2^2\Sigma_{SS}^{-1}$ 
where $U =\Omega_C\Sigma_{GG}\Omega_C^\top$ and $\Omega_C\Sigma_{GG}=\Sigma_{SG}.$
So 
\[
\text{var}(\hat\beta_s)-\text{var}(\hat\beta) \rightarrow \sigma_1^2\Sigma_{SS}^{-1}(\Sigma_{SS}-U)\Sigma_{SS}^{-1}
\]
and it suffices to show that $\Sigma_{SS}-U$ is positive semi-definite.
For any $p$, the Schur complement of $\Sigma_{GG}$ in $D$ is $\Sigma_{SS}-\Sigma_{SG}\Sigma_{GG}^{-1}\Sigma_{GS}$. Since D is positive definite and symmetric, its Schur complement of $\Sigma_{GG}$ in $D$ is also positive definite. Let $p\rightarrow\infty$ and we prove that $\Sigma_{SS}-U=\Sigma_{SS}-\Sigma_{SG}\Sigma_{GG}^{-1}\Sigma_{GS}$ is positive semi-definite.

Similarly, $\text{var} \{ n^{1/2}(\hat b-\beta_0) - n^{1/2}(\tilde b - \beta_0) \} \rightarrow PQP^\top$, 
where $Q = (\Sigma_{SS} - \Sigma_{SG} \Sigma_{GG}^{-1} \Sigma_{GS} )^{-1}$ is positive semi-definite, 
and $P = \Sigma_{SS}^{-1} \Sigma_{SG} \Sigma_{GG}^{-1} \Sigma_{GS} $.
Thus the $P Q P^\top$ is positive semi-definite.
\end{proof}

\end{document}